\documentclass{lmcs}
\pdfoutput=1

\usepackage{lastpage}
\renewcommand{\dOi}{14(4:16)2018}
\lmcsheading{}{1--\pageref{LastPage}}{}{}%
{Jul.~18,~2017}{Nov.~16,~2018}{}

\usepackage{amsmath,amsthm,amssymb,stmaryrd,xspace,pdfsync,mathrsfs,verbatim,bbm}

\usepackage{hyperref}
\usepackage{microtype}
\usepackage{tikz}
\usepackage[mathscr]{euscript}
\usepackage[english]{babel}
\usetikzlibrary{decorations.text,arrows,decorations.pathreplacing,shapes,fit,shadows,calc,patterns,intersections}

\newcommand{\mynote}[3][]{\todo[caption={\sf #3}, color={%
    \ifnum#2=0 green!20
    \else\ifnum#2=1 orange!30
    \else\ifnum#2=2 yellow!20
    \else\ifnum#2=3 cyan!20
    \else magenta!20\fi\fi\fi\fi}, size=\tiny, #1]{\renewcommand{\baselinestretch}{1}\selectfont\sf#3}\xspace}

\newcommand{\imprint}{imprint\xspace}
\newcommand{\imprints}{imprints\xspace}

\newcommand{\Imprints}{Imprints\xspace}

\newcommand{\tame}{multiplicative\xspace}

\newcommand{\ratms}{rating maps\xspace}

\newcommand{\Nice}{Nice\xspace}
\newcommand{\nice}{nice\xspace}

\newcommand{\mratm}{multiplicative rating map\xspace}
\newcommand{\mratms}{multiplicative rating maps\xspace}

\newcommand{\Mratms}{Multiplicative rating maps\xspace}

\newcommand{\drat}[1]{\ensuremath{\tau^{#1}_{\alpha,\rho}}\xspace}

\newcommand{\drata}{\drat{\pbpol{\Cs}}}

\newcommand{\prin}[2]{\ensuremath{\Is[#1](#2)}\xspace}

\newcommand{\ptriv}[1]{\ensuremath{\Ps_{\mathit{triv}}[#1]}\xspace}

\newcommand{\opti}[2]{\ensuremath{\Is_{#1}[#2]}\xspace}

\newcommand{\popti}[2]{\ensuremath{\Ps_{#1}[#2]}\xspace}

\newcommand{\pocopti}{\popti{\pol{\Cs}}{\alpha,\rho}}
\newcommand{\pbpopti}{\popti{\pbpol{\Cs}}{\alpha,\rho}}
\newcommand{\tpocopti}{\popti{\pol{\Cs}}{\beta,\tau}}
\newcommand{\gpocopti}{\popti{\pol{\Cs}}{\beta,\gamma}}

\newcommand{\Hb}{\ensuremath{\mathbf{H}}\xspace}

\newcommand{\Kb}{\ensuremath{\mathbf{K}}\xspace}
\newcommand{\Lb}{\ensuremath{\mathbf{L}}\xspace}

\newcommand{\Ub}{\ensuremath{\mathbf{U}}\xspace}

\newcommand{\quotienting}{quotienting\xspace}

\newcommand{\vari}{\quotienting Boolean algebra\xspace}

\newcommand{\pvari}{quotienting lattice\xspace}

\newcommand{\pvaris}{quotienting lattices\xspace}

\newcommand{\at}{\ensuremath{\textup{AT}}\xspace}

\newcommand{\bool}[1]{\ensuremath{{Bool\/}(#1)}\xspace}
\newcommand{\pol}[1]{\ensuremath{{Pol\/}(#1)}\xspace}
\newcommand{\bpol}[1]{\ensuremath{{BPol\/}(#1)}\xspace}
\newcommand{\pbpol}[1]{\ensuremath{PBPol(#1)}\xspace}

\newcommand{\polp}[2]{\ensuremath{Pol_{#2}(#1)}\xspace}
\newcommand{\polk}[1]{\polp{#1}{k}}
\newcommand{\compc}[1]{\ensuremath{{co\/}\textup{-}#1}}

\definecolor{my1}{cmyk}{0,.6,0,0}
\definecolor{my2}{cmyk}{.3,.0,.0,.0}

\newcommand{\nat}{\ensuremath{\mathbb{N}}\xspace}

\newcommand{\Cs}{\ensuremath{\mathcal{C}}\xspace}
\newcommand{\Ds}{\ensuremath{\mathcal{D}}\xspace}

\newcommand{\Is}{\ensuremath{\mathcal{I}}\xspace}

\newcommand{\Ps}{\ensuremath{\mathcal{P}}\xspace}

\newcommand{\Fs}{\ensuremath{\mathcal{F}}\xspace}
\newcommand{\Ts}{\ensuremath{\mathcal{T}}\xspace}

\newcommand{\sic}[1]{\ensuremath{\Sigma_{#1}}\xspace}
\newcommand{\siw}[1]{\ensuremath{\Sigma_{#1}(<)}\xspace}
\newcommand{\siws}[1]{\ensuremath{\Sigma_{#1}(<,+1,min,max,\varepsilon)}\xspace}
\newcommand{\pic}[1]{\ensuremath{\Pi_{#1}}\xspace}
\newcommand{\piw}[1]{\ensuremath{\Pi_{#1}(<)}\xspace}
\newcommand{\piws}[1]{\ensuremath{\Pi_{#1}(<,+1,min,max,\varepsilon)}\xspace}
\newcommand{\bsc}[1]{\ensuremath{\mathcal{B}\Sigma_{#1}}\xspace}
\newcommand{\bsw}[1]{\ensuremath{\mathcal{B}\Sigma_{#1}(<)}\xspace}
\newcommand{\bsws}[1]{\ensuremath{\mathcal{B}\Sigma_{#1}(<,+1,min,max,\varepsilon)}\xspace}

\newcommand{\sicu}{\ensuremath{\Sigma_{1}}\xspace}

\newcommand{\picu}{\ensuremath{\Pi_{1}}\xspace}

\newcommand{\bscu}{\ensuremath{\mathcal{B}\Sigma_{1}}\xspace}

\newcommand{\sicd}{\ensuremath{\Sigma_{2}}\xspace}

\newcommand{\picd}{\ensuremath{\Pi_{2}}\xspace}

\newcommand{\bscd}{\ensuremath{\mathcal{B}\Sigma_{2}}\xspace}

\newcommand{\sict}{\ensuremath{\Sigma_{3}}\xspace}
\newcommand{\siwt}{\ensuremath{\Sigma_{3}(<)}\xspace}

\newcommand{\pict}{\ensuremath{\Pi_{3}}\xspace}

\newcommand{\bsct}{\ensuremath{\mathcal{B}\Sigma_{3}}\xspace}

\newcommand{\fo}{\ensuremath{\textup{FO}}\xspace}
\newcommand{\fow}{\ensuremath{\textup{FO}(<)}\xspace}
\newcommand{\fows}{\ensuremath{\textup{FO}(<,+1,min,max,\varepsilon)}\xspace}

\newcommand{\typ}[2]{\ensuremath{[#1]_{#2}}\xspace}
\newcommand{\ctype}[1]{\typ{#1}{\Cs}}

\DeclareMathOperator{\uclos}{\uparrow}

\newcommand{\canoc}{\ensuremath{\leqslant_\Cs}\xspace}
\newcommand{\canod}{\ensuremath{\leqslant_\Ds}\xspace}

\newcommand{\polrelp}[1]{\ensuremath{\leqslant_{#1}}\xspace}
\newcommand{\polrelk}{\polrelp{k}}

\newcommand{\inv}{\ensuremath{^{-1}}}

\tikzstyle{nor}=[minimum size=0.35cm,draw,rounded rectangle,inner sep=2pt]
\tikzstyle{nod}=[minimum size=0.35cm,draw,circle,inner sep=2pt]
\tikzstyle{nok}=[minimum size=0.45cm,draw,circle,inner sep=0pt]
\tikzstyle{nof}=[minimum size=0.35cm,draw,circle,double,double
distance=1pt]
\tikzstyle{port}=[minimum size=0.35cm,draw,thick,rectangle,inner sep=2pt]
\tikzstyle{nop}=[minimum size=0.35cm,draw,thick,rectangle,inner sep=1pt,rotate=90]

\tikzstyle{nol}=[minimum size=0.35cm,draw,rounded rectangle,inner sep=1pt,rotate=90]

\tikzstyle{ar}=[line width=0.5pt,->,double]
\tikzstyle{siar}=[line width=1.5pt,->]
\tikzstyle{ars}=[line width=1.5pt,->,double]

\theoremstyle{plain}
\newtheorem{theorem}{Theorem}[section]
\newtheorem{corollary}[theorem]{Corollary}
\newtheorem{proposition}[theorem]{Proposition}
\newtheorem{example}[theorem]{Example}
\newtheorem{lemma}[theorem]{Lemma}
\newtheorem{fct}[theorem]{Fact}
\newtheorem{remark}[theorem]{Remark}
\newtheorem*{claim}{Claim}

\title{Separating regular languages with two quantifier alternations}

\author{Thomas~Place}
\address{LaBRI, Universit\'e de Bordeaux, Institut Universitaire de France}
\email{tplace@labri.fr}


\begin{document}
\begin{abstract}
	We investigate a famous decision problem in automata theory: separation. Given a class of language \Cs, the separation problem for \Cs takes as input two regular languages and asks whether there exists a third one which belongs to \Cs, includes the first one and is disjoint from the second. Typically, obtaining an algorithm for separation yields a deep understanding of the investigated class \Cs. This explains why a lot of effort has been devoted to finding algorithms for the most prominent classes.
	
	Here, we are interested in classes within concatenation hierarchies. Such hierarchies are built using a generic construction process: one starts from an initial class called the basis and builds new levels by applying generic operations. The most famous one, the dot-depth hierarchy of Brzozowski and Cohen, classifies the languages definable in first-order logic. Moreover, it was shown by Thomas that it corresponds to the quantifier alternation hierarchy of first-order logic: each level in the dot-depth corresponds to the languages that can be defined with a prescribed number of quantifier blocks. Finding separation algorithms for all levels in this hierarchy is among the most famous open problems in automata theory.
	
	Our main theorem is generic: we show that separation is decidable for the level $\frac{3}{2}$ of any concatenation hierarchy whose basis is finite. Furthermore, in the special case of the dot-depth, we push this result to the level $\frac{5}{2}$. In logical terms, this solves separation for \sict: first-order sentences having at most three quantifier blocks starting with an existential one.
\end{abstract}

\maketitle

\section{Introduction}
\label{sec:intro}
\noindent
{\bf Context.} This paper is part of a research program whose objective is to precisely understand prominent classes of regular languages and in particular those corresponding to descriptive formalisms (such as a logic). Naturally, ``understanding'' a class \Cs is an informal objective.  In the literature,  one usually approaches this question by considering a decision problem: \emph{membership}. For a given class \Cs, one looks for procedure deciding whether some input regular language belongs to \Cs. In practice, it turns out that such an algorithm yields a deep understanding of the class \Cs. Indeed, it is an effective description of \emph{all} languages in \Cs.

This approach was originally inspired by a theorem of Sch\"utzenberger~\cite{sfo}. From this theorem, one obtains an algorithm which decides whether an input regular language is star-free, \emph{i.e.}, can be expressed with a regular expression using union, complement and concatenation, but not Kleene star. This result was highly influential for several reasons:
\begin{itemize}
	\item The star-free languages are among the most prominent sub-classes of the regular languages. In particular, it was later shown by McNaughton and Papert~\cite{mnpfo} they are exactly those which can be defined in first-order logic (\fo).
	\item Being the first result of its kind, Sch\"utzenberger theorem cemented membership as the ``right'' question when aiming to ``understand'' a given class of languages.
	\item Sch\"utzenberger developed a general methodology for tackling membership problems (which he applied to the star-free languages).
\end{itemize}
Such a success motivated researchers to investigate membership for other important classes of languages. This was quite fruitful and the question is now well-understood for several important classes. Prominent examples include the class of piecewise testable languages which was solved by Simon~\cite{simon75} or the two-variables fragment of first-order logic  which was solved by Th\'erien and Wilke~\cite{twfodeux}. However, for some other classes, membership remains wide open, despite a wealth of research work spanning several decades.

\medskip
\noindent
{\bf Concatenation hierarchies.} In the paper, we investigate one of the most famous among these open questions: the \emph{dot-depth problem} (see~\cite{jep-dd45} for a recent survey on the topic). Sch\"utzenberger's results motivated Brzozowski and Cohen~\cite{BrzoDot} to define a classification of all star-free languages: the \emph{dot-depth hierarchy}. Intuitively, one classifies the star-free languages according to the number of alternations between concatenations and complements required to define them. More precisely, the dot-depth is a particular instance of a generic construction process (which was formalized later) named \emph{concatenation hierarchies}. Such a hierarchy has only one parameter: a ``level $0$ class'' called its \emph{basis}. Once it is fixed, one builds new levels by applying two generic operations: polynomial and Boolean closure. There are two kinds of levels: half levels $\frac12,\frac32,\frac52\dots$ and full levels $0,1,2,3\dots$. For any full level $n$, the next half and full levels are built as follows:
\begin{itemize}
	\item Level $n+\frac12$ is the \emph{polynomial closure} of level~$n$. Given a class \Cs, its polynomial closure \pol\Cs is the least class of languages containing \Cs and closed under union, intersection and marked concatenation ($K,L\mapsto KaL$, where $K,L \subseteq A^*$ and $a \in A$). 
	\item Level $n+1$ is the \emph{Boolean closure} of level $n+\frac12$. Given a class \Cs, its Boolean closure \bool\Cs is the least class containing \Cs and closed under union and complement.
\end{itemize}
Thus, a concatenation hierarchy is fully determined by its basis. In the paper, we carry out a generic investigation of hierarchies with a \emph{finite} basis.

There are essentially two prominent hierarchies of this kind: the dot-depth hierarchy~\cite{BrzoDot} and Straubing-Th\'erien hierarchy~\cite{StrauConcat,TheConcat} which are both classifications of the star-free languages. This status is partially explained by their connection with the quantifier alternation hierarchies within first-order logic. One may classify first-order sentences as follows. A sentence is \sic{n} either if its prenex normal form has $n$ quantifier blocks and starts with an existential one or has strictly less than $n$ quantifier blocks. Furthermore, a sentence is \bsc{n} if it is a finite Boolean combination of \sic{n} sentences. It was shown by Thomas~\cite{ThomEqu} that the dot-depth coincides with quantifier alternation: for any $n \in \nat$, dot-depth $n$ corresponds to $\bsc{n}$ and dot-depth $n+\frac{1}{2}$ to $\sic{n+1}$. This result was later lifted to the Straubing-Th\'erien hierarchy~\cite{PPOrder}: it corresponds to the alternation hierarchy within another variant of first-order logic equipped with a slightly different signature (this does not change the overall expressive power but impacts the languages that one may define at a given level of its alternation hierarchy).

\medskip

The \emph{dot-depth problem} is to obtain membership algorithms for all levels in these two hierarchies (and therefore, for the alternation hierarchies of first-order logic as well). However, progress has been slow: until recently, only level $\frac12$~\cite{arfi87,pwdelta}, level~1~\cite{simon75,knast83} and level $\frac32$~\cite{arfi87,pwdelta,gssig2} were solved. See~\cite{dgk-smallfragments} for a survey. Following these results, membership for level~2 remained open for a long time and became famous under the name ``dot-depth two problem''.

\smallskip\noindent\textbf{Separation.} Recently solutions were found for the levels $2$ and $\frac52$ by relying on a new approach~\cite{pzqalt}. The techniques involve considering a new ingredient: a stronger decision problem called \emph{separation}. Rather than asking whether a single input language belongs to the class \Cs under investigation, the \Cs-separation problem takes as input \emph{two} regular languages. It asks whether there exists a third language \emph{belonging to \Cs} which contains the first language and is disjoint from the second. The interest in separation is recent. However, it has quickly replaced membership as the central question. This newly acquired status is explained by two main reasons.

First, separation serves as the key ingredient in all recent membership results (see~\cite{PZ:Siglog15,pzcsr17} for an overview). A simple but crucial example is the membership algorithm for the level $\frac52$ in the Straubing-Th\'erien hierarchy~\cite{pzqalt}. It is based on two distinct results:
\begin{enumerate}
	\item A generic reduction from membership for the level $n+1$ to separation for the level $n$ for any {\bf half} level $n$ in the Straubing-Th\'erien hierarchy.
	\item A separation algorithm for the level $\frac32$.
\end{enumerate}
When combined, these two results yield a membership algorithm for the level $\frac52$.

However, there is another deeper reason for working with separation instead of membership. Our primary motivation is to thoroughly understand the classes that we investigate: this is why we considered membership in the first place. In this respect, while harder, separation is also far more rewarding than membership. Intuitively, this is easily explained. A membership algorithm for some class \Cs only yields benefits for the languages of \Cs: we are able to detect them and to build a description witnessing this membership. Instead, separation algorithms are universal: their benefits apply to \textbf{\emph{all}} languages. An insightful perspective is to view separation as an approximation problem: given an input pair $(L_{1},L_{2})$, we want to over-approximate $L_{1}$ by a language in \Cs, on the other hand $L_{2}$ is the specification of what an acceptable approximation is. Altogether, the key idea is that separation yields a more robust understanding of the classes.

Separation is known to be decidable for the lower levels in the Straubing-Th\'erien and dot-depth hierarchies. This was first shown to be decidable for the levels $\frac12$, $1$, $\frac32$ and $2$ in the Straubing-Th\'erien hierarchy~\cite{martens,pvzmfcs13,pzqalt,pzboolpol}. Let us point out that while these results were first formulated independently, it was recently proved in~\cite{pzboolpol} that the four of them are corollaries of only two generic theorems:
\begin{itemize}
	\item For any \emph{finite} class \Cs, \pol\Cs-separation is decidable. The levels $\frac12$ and $\frac32$ in the Straubing-Th\'erien hierarchy are of this form (by definition for the former and by a result of~\cite{pin-straubing:upper} for the latter).
	\item For any \emph{finite} class \Cs, \bool{\pol\Cs}-separation is decidable. The levels $1$ and $2$ in the Straubing-Th\'erien hierarchy are of this form (by~\cite{pin-straubing:upper} again for the latter).
\end{itemize}
This generic approach was introduced in~\cite{pzboolpol} and we shall continue it in the paper. Finally, it is known that any separation solution for a given level in the Straubing-Th\'erien hierarchy may be lifted to the corresponding level in the dot-depth via a generic reduction~\cite{pzsucc}. Thus, separation is also decidable for the levels $\frac12$, $1$, $\frac32$ and $2$ in the dot-depth.

\medskip
\noindent
\textbf{Contributions.} Our most important result in the paper is a separation algorithm for the level 5/2 in the Straubing-Th\'erien hierarchy.  Note that from a logical point of view, this is the level \sict in the quantifier alternation hierarchy of first-order logic: sentences with at most three blocks of quantifiers (i.e. two alternations between quantifiers) and an existential leftmost block. Moreover, we also obtain two important corollaries ``for free'' from previously known results. First, we are able  to lift the result to dot-depth 5/2 using the aforementioned transfer theorem~\cite{pzsucc}. Second, we obtain from the reduction of~\cite{pzqalt} that \emph{membership} is decidable for the level 7/2 in the Straubing-Th\'erien hierarchy (this may be lifted to dot-depth as well using a result of Straubing~\cite{StrauVD}).

A crucial point is that the above result is {\bf not} our main theorem but only its most important corollary. The actual main result is more general for two distinct reasons. First, we actually consider an even stronger problem that separation: \emph{covering}. This third problem was introduced in~\cite{pzcovering,pzcovering2} as a generalization of separation. It takes two objects as input, a single regular language $L$ and a finite set of regular languages \Lb. Given some class \Cs, the \Cs-covering problem asks whether there exists a \Cs-cover of $L$ (i.e. a finite set of languages in \Cs whose union includes $L$) such that no language in this \Cs-cover intersects all languages in \Lb. It is simple to show that separation is the special case of covering when the set \Lb is a singleton. We prove that covering (and therefore separation as well) is decidable for the level 5/2 in the Straubing-Th\'erien hierarchy. Considering covering has benefits: it is arguably even more rewarding than separation and it comes with an elegant framework~\cite{pzcovering,pzcovering2} designed for tackling it (which we use to formulate all algorithms presented in the paper). However, a more fundamental motivation for considering covering is that we need to: no ``direct'' separation algorithm is known for the level 5/2.

Additionally, we follow a generic approach similar to the one used in~\cite{pzboolpol} for \pol{\Cs} and \bool{\pol{\Cs}}. Specifically, our main theorem states that covering and separation are decidable for any class of the form $\pol{\bool{\pol\Cs}}$ when \Cs is a \emph{finite} class (it follows from the result of~\cite{pin-straubing:upper} mentioned above that the level 5/2 in the Straubing-Th\'erien hierarchy is such a class). Being generic, this approach yields algorithms for a whole family of classes. Moreover, we are able to pinpoint the key hypotheses which are critical in order to solve separation and covering for the level 5/2.

Finally, we shall also reprove the theorem of~\cite{pzboolpol} for classes of the form  \pol{\Cs} when \Cs is finite: the \pol{\Cs}-covering problem is decidable. In fact, we prove a stronger theorem that we require for handling $\pol{\bool{\pol\Cs}}$.

\medskip\noindent\textbf{Organization.} Section~\ref{sec:prelims} gives preliminary definitions. We introduce classes of languages and present the decision problems that we consider. In Section~\ref{sec:concat}, we define concatenation hierarchies, state our generic separation theorem and discuss its consequences. The remainder of the paper is then devoted to proving this theorem. In Section~\ref{sec:tools}, we introduce mathematical tools that we shall need. Section~\ref{sec:covers} is then devoted to presenting the framework that we use for formulating our algorithms. Finally, we present and prove our algorithms for  \pol\Cs- and \pol{\bool{\pol{\Cs}}}-separation (when \Cs is finite) in Sections~\ref{sec:polc} and~\ref{sec:pbpol} respectively

\medskip

This paper is the journal version of~\cite{pseps3}. There are several differences between the two versions. In the conference version, the point of view was purely logical: the main theorem is the decidability of \sict-separation and concatenation hierarchies are not discussed. Here, we introduce the language theoretic point of view and concatenation hierarchies which allows us to state a more general theorem: \pol{\bool{\pol\Cs}}-separation when \Cs is finite. This statement better highlights the crucial hypotheses that are needed for solving \sict-separation. Finally, while the core ideas remain the same as those used in~\cite{pseps3}, the proof arguments and their formulation have been significantly modified in order to simplify the presentation and better pinpoint the key ideas.


\section{Preliminaries}
\label{sec:prelims}
In this section, we define classes of languages as well as the decision problems that we shall consider. Moreover, we introduce some standard terminology that we shall use.

\subsection{Classes of languages}

For the whole paper, we fix an arbitrary finite alphabet $A$. We shall denote by $A^*$ the set of all words over $A$, including the empty word~$\varepsilon$. We let $A^{+} = A^{*} \setminus \{\varepsilon\}$. If $u,v \in A^*$ are words, we write $u \cdot v$ or $uv$ for their concatenation.

A subset of $A^*$ is called a \emph{language}. For the sake of avoiding clutter we often denote the singleton language $\{u\}$ by $u$. It is standard to extend the concatenation operation to languages: for $K,L \subseteq A^*$, $KL$ denotes the language $KL = \{uv \mid u \in K \text{ and } v \in L\}$. Moreover, we also consider \emph{marked concatenation}, which is less standard. Given $K,L \subseteq A^*$, a \emph{marked concatenation of $K$ with $L$} is a language of the form $KaL$ for some $a \in A$.

A \emph{class of languages} \Cs is simply a set of languages. All classes that we consider in the paper satisfy robust \emph{closure properties} which we define now.
\begin{itemize}
	\item A \emph{lattice} is a class of languages \Cs closed under finite union and finite intersection, and such that $\emptyset \in \Cs$ and $A^* \in \Cs$.
	\item A \emph{Boolean algebra} is a lattice closed under complement.
	\item Finally, we say that a class \Cs is \emph{quotienting} when it is closed under quotients, \emph{i.e.}, when for all $L \in \Cs$ and all $w \in A^*$, the following two languages both belong to \Cs,
	\[
	w^{-1}L \stackrel{\text{def}}= \{u \in A^* \mid wu \in L\} \quad \text{and} \quad Lw^{-1} \stackrel{\text{def}}= \{u \in A^* \mid uw \in L\}
	\]
\end{itemize}

In the paper, all classes that we consider are at least \emph{quotienting lattices}. Moreover, we only consider classes which are included in the class of \emph{regular languages} (which is something that we shall implicitly assume from now on). These are the languages that can be equivalently defined by monadic second-order logic, finite automata or finite monoids (we come back to this point in Section~\ref{sec:tools}).

\subsection{Separation and covering}

Our objective in the paper is to study several specific classes of languages (which we present in Section~\ref{sec:concat}). For this purpose, we shall rely on two decision problems: separation and covering. Both of them are parametrized by an arbitrary class of languages \Cs. Let us start with the definition of separation.

\medskip
\noindent
{\bf Separation.} Given three languages $K,L_1,L_2$, we say that $K$ \emph{separates} $L_1$ from $L_2$ if $L_1 \subseteq K \text{ and } L_2 \cap K = \emptyset$. Given a class of languages \Cs, we say that $L_1$ is \emph{$\Cs$-separable} from $L_2$ if some language in \Cs separates $L_1$ from $L_2$. Observe that when \Cs is not closed under complement (which is the case for all classes investigated in the paper), the definition is not symmetrical: $L_1$ could be \Cs-separable from $L_2$ while $L_2$ is not \Cs-separable from $L_1$. The separation problem associated to a given class \Cs is as follows:

\medskip

\begin{tabular}{rl}
	{\bf INPUT:}  &  Two regular languages $L_1$ and $L_2$. \\
	{\bf OUTPUT:} &  Is $L_1$ $\Cs$-separable from $L_2$ ? 
\end{tabular}

\medskip

We use separation as a mathematical tool whose purpose is to investigate classes of languages: given a fixed class \Cs, obtaining a \Cs-separation algorithm usually requires a solid understanding of~\Cs. A typical objective when considering separation is to not only get an algorithm that decides it, but also a generic method for computing a separator, if it exists.

\begin{remark}
	\Cs-separation generalizes another classical decision problem: \Cs-membership which asks whether a single regular language $L$ belongs to \Cs. Indeed, it is simple to verify that asking whether $L \in \Cs$ is equivalent to asking whether $L$ is $\Cs$-separable from its complement (in this case, the only candidate for being a separator is $L$ itself).
\end{remark}

\noindent
{\bf Covering.} Our second problem is more general and was originally defined in~\cite{pzcovering,pzcovering2}. Given a language $L$, a \emph{cover of $L$} is a \emph{\bf finite} set of languages \Kb such that,
\[
L \subseteq \bigcup_{K \in \Kb} K
\]
Moreover, given a class \Cs, a \Cs-cover of $L$  is a cover \Kb of $L$ such that all $K \in \Kb$ belong to \Cs.

Covering takes as input a language $L$ and a \emph{finite set of languages} $\Lb$. A \emph{separating cover} for the pair $(L,\Lb)$ is a cover \Kb of $L$ such that for every $K \in \Kb$, there exists $L' \in \Lb$ which satisfies $K \cap L' = \emptyset$. Finally, given a class \Cs, we say that the pair $(L,\Lb)$ is \Cs-coverable when there exists a separating \Cs-cover. The \Cs-covering problem is now defined as follows:

\medskip

\begin{tabular}{rl}
	{\bf INPUT:}  &  A regular language $L$ and a finite set of regular languages \Lb.\\
	{\bf OUTPUT:} &  Is $(L,\Lb)$ $\Cs$-coverable ? 
\end{tabular}

\begin{remark}
	This definition is slightly different from the one of~\cite{pzcovering2} in which \Lb is defined as a ``multiset'' of languages (in the sense that it may contain several instances of the same language). This is natural since each input language in \Lb is actually given by a recognizer (such as an NFA or a monoid morphism) and two distinct recognizers may define the same language. This change is harmless: it is immediate that the two definitions are equivalent. 
\end{remark}

Let us complete this definition by explaining why covering generalizes separation: the latter is special case of the former when the set \Lb is a singleton.

\begin{fct} \label{fct:septocove}
	Let \Cs be a lattice and $L_1,L_2$ two languages. Then $L_1$ is \Cs-separable from $L_2$, if and only if $(L_1,\{L_2\})$ is \Cs-coverable.
\end{fct}

\begin{proof}
	Assume first that $L_1$ is \Cs-separable from $L_2$ and let $K\in\Cs$ be a separator. Then, it is immediate that $\Kb = \{K\}$ is a separating \Cs-cover for  $(L_1,\{L_2\})$. Conversely, assume that $(L_1,\{L_2\})$ is \Cs-coverable and let \Kb be a separating \Cs-cover. Moreover, let $K$ be the union of all languages in \Kb. We have $K \in \Cs$ since \Cs is a lattice. Clearly, $L_1 \subseteq K$ since \Kb was a cover of $L_1$. Moreover, we know that no language in \Kb intersects $L_2$ since \Kb was separating. Thus, $L_2 \cap K = \emptyset$ which means that $K \in \Cs$ separates $L_1$ from $L_2$. 
\end{proof}

\begin{remark}
	While covering is a natural generalization of separation, we consider it out of necessity. For the classes that we investigate, no ``direct'' separation algorithm is known. Our techniques require considering the covering problem.
\end{remark}

\subsection{Finite lattices and stratifications}

All classes that we consider in the paper are built from an arbitrary finite lattice (i.e. one that contains finitely many languages) using a generic construction process (presented in Section~\ref{sec:concat}). Consequently, finite lattices will be important in the paper. They have several convenient properties that we present here.

\medskip
\noindent
{\bf Canonical preorder relations.} Consider a finite lattice \Cs. One may associate a \emph{canonical preorder relation over $A^*$} to \Cs. The definition is as follows. Given $w,w' \in A^*$, we write $w \canoc w'$ if and only if the following holds:
\[
\text{For all $L \in \Cs$,} \quad w \in L \ \Rightarrow\ w' \in L.
\]
It is immediate from the definition that \canoc is transitive and reflexive, making it a preorder.

\begin{example} \label{ex:ateq}
	We let \at as the class of languages consisting of all Boolean combinations of languages $A^*aA^*$, for some $a \in A$. Though this terminology is not standard, ``\at'' stands for ``alphabet testable'': $L \in \at$ if and only if membership of a word $w$ in $L$ depends only on the set of letters occurring in $w$. Clearly, \at is a finite Boolean algebra. In that case, $\leqslant_\at$ is an equivalence relation which we denote by $\sim_\at$: one may verify that $w \sim_\at w'$ if and only if $w$ and $w'$ have the same alphabet (i.e. contain the same set of letters).
\end{example}

The relation \canoc has many applications. We start with an important lemma, which relies on the fact that \Cs is finite. We say that a language $L \subseteq A^*$ is an \emph{upper set} (for \canoc) when for any two words $u,v \in A^*$, if $u \in L$ and $u \canoc v$, then $v \in L$.

\begin{lemma} \label{lem:canosatur}
	Let $\Cs$ be a finite lattice. Then, for any $L \subseteq A^*$, we have $L \in \Cs$ if and only if $L$ is an upper set for \canoc. In particular, \canoc has finitely many upper sets.
\end{lemma}

\begin{proof}
	Assume first that $L \in \Cs$. Then, for all $w \in L$ and all $w'$ such that $w \canoc w'$, we have $w' \in L$ by definition of \canoc. Hence, $L$ is an upper set.	Assume now that $L$ is an upper set. For any word $w$, we write $\uclos w$ for the  upper set $\uclos w = \{u \mid w \canoc u\}$. By definition of \canoc $\uclos w$ is the intersection of all $L \in \Cs$ such that $w \in L$. Therefore, $\uclos w \in \Cs$ since \Cs is a finite lattice (and is therefore closed under intersection). Finally, since $L$ is an upper set, we have,
	\[
	L = \bigcup_{w \in L} \uclos w.
	\]
	Hence, since \Cs is closed under union and is finite, $L$ belongs to \Cs.
\end{proof}

While Lemma~\ref{lem:canosatur} states an equivalence, we mainly use the left to right implication (or rather its contrapositive). One may apply it to show that a language does not belong to \Cs. Indeed, by the lemma, proving that $L \not\in \Cs$ is the same as proving that $L$ is not an upper set for \canoc. In other words, one needs  to exhibit $u,v \in A^*$ such that $u \canoc v$, $u \in L$ and $v \not\in L$.

\begin{example} \label{ex:notat}
	Assume that $A = \{a,b\}$ and consider the class \at of Example~\ref{ex:ateq}. The language $L = A ^*aA^*bA^*$ does not belong to \at. Indeed, $ab \sim_\at ba$, $ab \in L$ and $ba \not\in L$.	
\end{example}

Let us complete these definitions with a few additional useful results. First, as we observed for \at in Example~\ref{ex:ateq}, when the finite lattice \Cs is actually a Boolean algebra, it turns out that \canoc is an equivalence relation, which we shall denote by $\sim_\Cs$.

\begin{lemma} \label{lem:canoequiv}
	Let $\Cs$ be a finite Boolean algebra. Then for any alphabet $A$, the canonical preorder \canoc is an equivalence relation $\sim_\Cs$ which admits the following direct definition :
	\[
	\text{$w \sim_\Cs w'$ if and only if for all $L \in \Cs$,\quad $w \in L \ \Leftrightarrow\ w' \in L$}
	\]
	Thus, for any $L \subseteq A^*$, we have $L \in \Cs$ if and only if $L$ is a union of $\sim_\Cs$-classes. In particular, $\sim_\Cs$ has finite index.
\end{lemma}

\begin{proof}
	It is clear that when $w \sim_\Cs w'$, we have $w \canoc w'$ as well. We prove the reverse implication. Let $w,w' \in A^*$ be such that $w \canoc w'$. Let $L \in \Cs$ and observe that by closure under complement, we know that $A^* \setminus L \in \Cs$. Therefore, by definition of $w \canoc w'$,
	\[
	\begin{array}{lll}
	w \in L               & \Rightarrow & w' \in L               \\
	w \in A^* \setminus L & \Rightarrow & w' \in A^* \setminus L
	\end{array}
	\]
	One may now combine the first implication with the contrapositive of the second, which yields $w \in L \ \Leftrightarrow\ w' \in L$. We conclude that $w \sim_\Cs w'$: \canoc and $\sim_\Cs$ are the same relation.
\end{proof}

Another important and useful property is that when \Cs is \quotienting, the canonical preorder \canoc is compatible with word concatenation. 

\begin{lemma} \label{lem:canoquo}
	A finite lattice \Cs is \quotienting if and only if its associated canonical preorder \canoc is compatible with word concatenation. That is, for any words $u,v,u',v'$,
	\[
	u \canoc u' \quad \text{and} \quad v \canoc v' \quad \Rightarrow \quad uv \canoc u'v'.
	\]		
\end{lemma}

\begin{proof}
	First assume that \Cs is closed under quotients and let $u,u',v,v'$ be four words such that $u \canoc u'$ and $v \canoc v'$. We have to prove that $uv \canoc u'v'$. Let $L \in \Cs$ and assume that $uv \in L$. We use closure under left quotients to prove that $uv' \in L$ and then closure under right quotients to prove that $u'v' \in L$ which terminates the proof of this direction. Since $uv \in L$, we have $v \in u^{-1} \cdot L$. By closure under left quotients, we have $u^{-1} \cdot L \in \Cs$, hence, since $v \canoc v'$, we obtain that $v'\in u^{-1} \cdot L$ and therefore that $uv' \in L$. It now follows that $u \in L \cdot (v')^{-1}$. Using closure under right quotients, we obtain that $L \cdot (v')^{-1} \in \Cs$. Therefore, since $u \canoc u'$, we conclude that $u' \in L \cdot (v')^{-1}$ which means that $u'v' \in L$, as desired.
	
	Conversely, assume that \canoc is a precongruence.	Let $L \in \Cs$ and $w \in A^*$, we prove that $w^{-1} \cdot L \in \Cs$ (the proof for right quotients is symmetrical). By Lemma~\ref{lem:canosatur}, we have to prove that $w^{-1} \cdot L$ is an upper set. Let $u \in w^{-1} \cdot L$ and $u' \in A^*$ such that $u \canoc u'$. Since \canoc is a precongruence, we have $wu \canoc wu'$. Hence, since $L$ is an upper set (it belongs to \Cs) and $wu \in L$, we have $wu' \in L$. We conclude that $u' \in w^{-1} \cdot L$, which terminates the proof.
\end{proof}

\medskip
\noindent
{\bf Stratifications.} While the above notions are useful, the downside is that they only apply to \emph{finite} lattices. However, it is possible to lift their benefits to infinite classes with the notion of \emph{stratification}. Consider an arbitrary \emph{infinite} \pvari \Cs. A \emph{stratification} of \Cs is an infinite sequence $\Cs_0,\dots,\Cs_k,\dots$ of \emph{finite} \pvaris such that,
\[
\text{For all $k$,}\quad \Cs_k \subseteq \Cs_{k+1} \quad \text{and} \quad \Cs = \bigcup_{k \in \nat} \Cs_k.
\]
The point here is that once we have a stratification of \Cs in hand, we may now associate a canonical preorder \polrelk to each stratum $\Cs_k$.  Proving that a language $L$ does not belong to \Cs now amounts to proving that it does not belong to any stratum $\Cs_k$: for all $k \in \nat$, one needs to exhibit $u,v \in A^*$ such that $u \polrelk v$, $u \in L$ and $v \not\in L$.


\section{Concatenation hierarchies and main theorem}
\label{sec:concat}
In this section, we present the classes that we investigate in the paper. We actually consider a family of classes which are all \emph{built using a simple generic construction process}. As we shall see, this family includes several prominent classes coming from mathematical logic.

We start by defining the family of classes that we consider. Then, we present the main theorem of the paper: covering (and therefore separation as well) is decidable for any class in the family.  We illustrate this theorem by discussing its consequences and presenting the most prominent classes to which it applies.

\subsection{Closure operations}

Our generic theorem states that covering is decidable for any class which is built from a {\bf finite} \vari using a generic construction process that we present now. This construction is based on two operations that one may apply to classes of languages: \emph{Boolean closure and polynomial closure}.

\medskip
\noindent
{\bf Boolean closure.} As the name, suggests, the Boolean closure of a class \Cs (denoted by \bool{\Cs}) is simply the least Boolean algebra which contains \Cs.  The following lemma is immediate as one may verify that quotients commute with Boolean operations. 

\begin{lemma} \label{lem:bclos}
	For any \pvari \Cs, its Boolean closure \bool{\Cs} is a \vari.
\end{lemma}

\medskip
\noindent
{\bf Polynomial closure.} Consider a class \Cs, its \emph{polynomial closure} (denoted by \pol{\Cs}) is the least class containing \Cs and \emph{closed under union and marked concatenation}:
\[
\text{For all $K,L \in \pol{\Cs}$ and all $a \in A$}, \quad K \cup L \in \pol{\Cs} \text{ and } KaL \in \pol{\Cs}.
\]
When the input class \Cs is a \pvari (which is the only case that we consider here), its polynomial closure \pol{\Cs} has robust properties. In particular, it was shown by Arfi~\cite{arfi87} that \pol{\Cs} is also a \pvari (the difficulty here is to prove that \pol{\Cs} is closed under intersection which is not apparent in the definition). We summarize the properties that we shall need in the following theorem whose proof is available in~\cite{pzgenconcat} (see Lemma~27 and Theorem~29).

\begin{theorem} \label{thm:pclos}
	Let \Cs be a \pvari. Then, \pol{\Cs} is a \pvari closed under concatenation and marked concatenation.
\end{theorem}

For the sake of avoiding clutter, we shall write $\bpol{\Cs}$ for $\bool{\pol{\Cs}}$ and $\pbpol{\Cs}$ for $\pol{\bpol{\Cs}}$. It was shown in~\cite{pzboolpol} that for any {\bf finite} \vari \Cs, separation\footnote{In fact, while this is not explicitly stated in the paper, the arguments of~\cite{pzboolpol} apply to covering as well.} is decidable for both \pol{\Cs} and \bpol{\Cs}.

Our main theorem in the paper extends these results: we show that for any {\bf finite} \vari \Cs, separation and covering are decidable for the class \pbpol{\Cs}.  Moreover,  we also obtain a new proof for the decidability of \pol{\Cs}-covering as a side result (handling \pbpol{\Cs} requires a stronger result for \pol{\Cs} than the decidability of covering).

\begin{remark}
	The arguments that we use here for handling \pbpol{\Cs} are mostly independent from the ones used in~\cite{pzboolpol} for \bpol{\Cs}. The key idea is that both results build on our knowledge of the simpler class \pol{\Cs}. However, they do so in orthogonal directions.	  
\end{remark}

Before we state our main theorem properly, let us finish the definitions with two useful results about classes of the form \pbpol{\Cs} that we shall need later. A first result is that such classes always contain all finite languages.

\begin{lemma} \label{lem:pbpfinite}
	Let \Cs be an arbitrary lattice. Then, \bpol{\Cs} (and therefore \pbpol{\Cs} as well) contains all finite languages.
\end{lemma}

\begin{proof}
	Since \bpol{\Cs} is closed under union, it suffices to show that for any $w \in A^*$, the singleton $\{w\}$ belongs to \bpol{\Cs}. Consider a word $u \in A^*$, and let $u = a_1 \cdots a_n$ be the decomposition of $u$ as a concatenation of letters. Since $\Cs$ is a lattice, we have $A^* \in \Cs$. Thus, $L_u = A^*a_1A^*a_2 \cdots A^*a_nA^* \in \pol{\Cs}$. One may now verify that,
	\[
	\{w\} = L_w \setminus \left(\bigcup_{\{u \in A^* \mid |u| = n+1\}}L_u\right) \in \bpol{\Cs}
	\]
	This concludes the proof.
\end{proof}

The second result is used to bypass Boolean closure in the definition of \pbpol{\Cs}. Given any class \Ds, we write $\compc{\Ds}$ the \emph{complement class} of \Ds. That is, $\compc{\Ds}$ contains all complement languages $A^* \setminus L$ where $L \in \Ds$. One may verify the following fact.

\begin{fct} \label{fct:compclass}
	For any \pvari \Ds, the complement class $\compc{\Ds}$ is a \pvari.
\end{fct}

We may now state our second result:  Boolean closure may be replaced by complement in the definition of \pbpol{\Cs}.

\begin{lemma} \label{lem:comptrick}
	Let \Cs be a Boolean algebra. Then, $\pbpol{\Cs} = \pol{\compc{\pol{\Cs}}}$.
\end{lemma}

\begin{proof}
	By definition, it is clear that $\compc{\pol{\Cs}} \subseteq \bpol{\Cs}$. Hence, the right to left inclusion $\pol{\compc{\pol{\Cs}}} \subseteq \pbpol{\Cs}$ is immediate and we may concentrate on the converse one.
	
	We show that $\bpol{\Cs} \subseteq \pol{\compc{\pol{\Cs}}}$. It will then follow that we have $\pbpol{\Cs} \subseteq \pol{\pol{\compc{\pol{\Cs}}}}= \pol{\compc{\pol{\Cs}}}$ as desired. Let $L$ be a language in \bpol{\Cs}. By definition, $L$ is a Boolean combination of languages in \pol{\Cs}. Therefore, it follows from DeMorgan's laws that $L$ is built by applying unions and intersections to languages that are either in \pol{\Cs} or in $\compc{\pol{\Cs}}$. Moreover, observe that $\pol{\Cs} \subseteq \pol{\compc{\pol{\Cs}}}$ (since \Cs is a Boolean algebra, we have $\Cs \subseteq \compc{\pol{\Cs}}$) and $\compc{\pol{\Cs}} \subseteq \pol{\compc{\pol{\Cs}}}$. Therefore, since $\pol{\compc{\pol{\Cs}}}$ is closed under union and intersection, we have $L \in \pol{\compc{\pol{\Cs}}}$.
\end{proof}

\subsection{Main theorem} We are now ready to state the main theorem of the paper: \pbpol{\Cs}-covering is decidable when \Cs is a finite \vari.

\begin{theorem} \label{thm:maintheo}
	Let \Cs be a finite \vari. Then covering and separation are both decidable for \pol{\Cs} and \pbpol{\Cs}.
\end{theorem}

All remaining sections are devoted to the proof of Theorem~\ref{thm:maintheo}. We use Section~\ref{sec:tools} to introduce some mathematical tools that we shall need in our proofs: Simon's factorization theorem~\cite{simonfacto,kfacto} and a generic stratification for \pol{\Cs} when \Cs is a finite \pvari. Then, we devote  Section~\ref{sec:covers} to presenting a general framework which is designed for handling covering problems (it was originally introduced in~\cite{pzcovering,pzcovering2}). We shall use this framework to formulate our two covering algorithms. Finally, Section~\ref{sec:polc} and~\ref{sec:pbpol} are devoted to presenting and proving our covering algorithms for \pol{\Cs} and \pbpol{\Cs} respectively. However, before we turn to the proof of Theorem~\ref{thm:maintheo}, let us discuss its applications.

\subsection{Concatenation hierarchies} The polynomial and Boolean closure operations are important as they are involved in the definition of natural hierarchies of classes of languages: \emph{concatenation hierarchies}. Here, we briefly recall what they are (we refer the reader to~\cite{pinbridges,pwdelta,jep-dd45,StrauDD2,pzgenconcat} for details). Such a hierarchy depends on a single parameter: a \vari of regular languages~\Cs, called its \emph{basis}. Once the basis is chosen, the construction is uniform. Languages are classified into levels of two kinds: full levels (denoted by 0, 1, 2,$\dots$) and half levels (denoted by 1/2, 3/2, 5/2,$\dots$):
\begin{itemize}
	\item Level $0$ is the basis (\emph{i.e.}, our parameter class \Cs).
	\item Each \emph{half level} $n+1/2$, for $n\in\nat$, is the \emph{polynomial closure} of the previous full level, \emph{i.e.}, of level $n$.
	\item Each \emph{full level} $n+1$, for $n\in\nat$, is the \emph{Boolean closure} of the previous half level, \emph{i.e.}, of level $n+1/2$.
\end{itemize}

\begin{figure}[!htb]
	\centering
	\begin{tikzpicture}[scale=.9]
	\node[anchor=east] (l00) at (0.0,0.0) {{\large $0$}};
	
	\node[anchor=east] (l12) at (2.0,0.0) {\large $1/2$};
	\node[anchor=east] (l11) at (4.0,0.0) {\large $1$};
	\node[anchor=east] (l32) at (6.0,0.0) {\large $3/2$};
	\node[anchor=east] (l22) at (8.0,0.0) {\large $2$};
	\node[anchor=east] (l52) at (10.0,0.0) {\large $5/2$};
	
	\draw[very thick,->] (l00) to node[above] {$Pol$} (l12);
	\draw[very thick,->] (l12) to node[below] {$Bool$} (l11);
	\draw[very thick,->] (l11) to node[above] {$Pol$} (l32);
	\draw[very thick,->] (l32) to node[below] {$Bool$} (l22);
	\draw[very thick,->] (l22) to node[above] {$Pol$} (l52);
	
	\draw[very thick,dotted] (l52) to ($(l52)+(1.0,0.0)$);
	\end{tikzpicture}
	\caption{A concatenation hierarchy}
	\label{fig:hintro:concat}
\end{figure}

Therefore, Theorem~\ref{thm:maintheo} may be reformulated as follows: if \Cs is a {\bf finite} basis, separation and covering are decidable for the levels 1/2 and 3/2 in the associated concatenation hierarchy. There are two prominent examples of concatenation hierarchies with a finite basis:
\begin{itemize}
	\item The \emph{Straubing-Th\'erien} hierarchy~\cite{StrauConcat,TheConcat}, its basis is $\{\emptyset,A^*\}$.
	\item The \emph{dot-depth} hierarchy of Brzozowski and Cohen~\cite{BrzoDot}, its basis is $\{\emptyset,\{\varepsilon\},A^+,A^*\}$.
\end{itemize}

These two hierarchies are strict~\cite{BroKnaStrict}. Solving separation (and membership) for all their levels is a longstanding open problem. While it was already known that separation is decidable for the levels 1/2, 1, 3/2 and 2 for both hierarchies~\cite{martens,pvzmfcs13,pzqalt,pzboolpol,pzsucc}, it turns out that for these two hierarchies, we are able to lift these results to the level 5/2.

Recall that \at denotes the class of alphabet testable languages: it consists of all Boolean combinations of languages $A^*aA^*$, for some letter $a \in A$ (see Example~\ref{ex:ateq}). It is straightforward to verify that \at is a finite \vari. The following theorem was shown in~\cite{pin-straubing:upper} (see also Theorem~67 in~\cite{pzgenconcat} for a recent proof). 

\begin{theorem} \label{thm:alphatrick}
	Every level $n \geq$ 3/2 (half or full) in the Straubing-Th\'erien hierarchy corresponds exactly to the level $n-1$ in the concatenation hierarchy of basis \at.
\end{theorem}

Theorem~\ref{thm:alphatrick} implies that the level 5/2 of the Straubing-Th\'erien hierarchy is exactly the level 3/2 in the concatenation hierarchy of basis \at, i.e. the class \pbpol{\at} by definition. Therefore, the following corollary is an immediate consequence of Theorem~\ref{thm:maintheo}.

\begin{corollary} \label{cor:seventwo}
	The separation and covering problems are decidable for the level 5/2 of the Straubing-Th\'erien hierarchy
\end{corollary}

Additionally, it was shown in~\cite{pzqalt}, that for every {\bf half} level $n$ in Straubing-Th\'erien hierarchy, membership for the level $n+1$ reduces to separation for the level $n$. Therefore, we also obtain the following corollary.

\begin{corollary} \label{cor:ninetwo}
	It decidable to test whether a regular language belongs to the level 7/2 in the Straubing-Th\'erien hierarchy
\end{corollary}

Finally, it is known that these two results may be lifted to the dot-depth hierarchy using generic transfer theorems (see~\cite{StrauVD} for membership and~\cite{pzsucc,pzsucc2} for separation and covering). Therefore, we also get the two following additional corollaries.

\begin{corollary} \label{cor:seventwodd}
	The separation and covering problems are decidable for the level 5/2 of the dot-depth hierarchy.
\end{corollary}
\begin{corollary} \label{cor:ninetwodd}
	It decidable to test whether a regular language belongs to the level 7/2 in the dot-depth hierarchy
\end{corollary}

\subsection{First-order logic and quantifier alternation} We finish the section by discussing the connection with quantifier alternation hierarchies. 

Let us briefly recall the definition of first-order logic over words. One may view a finite word as a logical structure made of a sequence of positions. Each position carries a label in the alphabet~$A$ and can be quantified. We denote by ``$<$'' the linear order over the positions. We consider first-order logic, which is equipped with the following signature: 
\begin{itemize}
	\item For each $a \in A$, a unary predicate $P_a$ selecting positions labeled with the letter ``$a$''.
	\item A binary predicate ``$<$'' interpreted as the linear order.
	\item A binary predicate ``$+1$'' interpreted as the successor relation.
	\item  Unary predicates ``$min$'' and ``$max$'' selecting the leftmost and rightmost positions.
	\item A constant ``$\varepsilon$''  which holds when the word is empty.
\end{itemize}
We shall write \fows for this variant of first-order logic (we consider another variant with a restricted signature below). For every sentence $\varphi$, one may associate the language $\{w \in A^* \mid w \models \varphi\}$ of words satisfying $\varphi$.  Therefore, \fows defines the class of languages that can be defined with an \fows sentence. It is known that this class is a \vari.

Here, we are not interested in first-order logic itself: we consider its quantifier alternation hierarchy. One may classify sentences by counting their number of quantifier alternations. Let $n \in \nat$. We say that a sentence is \siws{n} (resp. \piws{n}) if it can be rewritten into a sentence in prenex normal form which has either,
\begin{itemize}
	\item \emph{exactly} $n -1$ quantifier alternations (\emph{i.e.}, exactly $n$ blocks of quantifiers) starting with an $\exists$ (resp.\ $\forall$), or
	\item \emph{strictly less} than $n -1$ quantifier alternations (\emph{i.e.}, strictly less than $n$ blocks of	quantifiers).
\end{itemize}
For example, consider the following sentence (already in prenex normal form)
\[
\forall x_1 \exists x_2 \forall x_3 \forall x_4
\ \varphi(x_1,x_2,x_3,x_4) \quad \text{(with $\varphi$ quantifier-free)}
\]
\noindent
This sentence is \piws{3}. With this definition, the $\siws{0}$ and $\piws{0}$ sentences are just the quantifier-free ones. In general, the negation of a \siws{n} sentence is not a \siws{n} sentence (it is a \piws{n} sentence). Hence it is relevant to define \bsws{n} sentences as the Boolean combinations of \siws{n} sentences. This gives a hierarchy of classes of languages as presented in Figure~\ref{fig:hiera}.

\tikzstyle{non}=[inner sep=1pt]
\tikzstyle{tag}=[draw,fill=white,sloped,circle,inner sep=1pt]
\begin{figure}[!htb]
	\centering
	\begin{tikzpicture}
	
	\node[non] (b0) at (-0.5,0.0) {$\sic{0} = \pic{0} = \bsc{0}$};
	
	\node[non] (s1) at (1.0,-0.8) {\sicu};
	\node[non] (p1) at (1.0,0.8) {\picu};
	\node[non] (b1) at (2.5,0.0) {\bscu};
	
	\node[non] (s2) at ($(b1)+(1.5,-0.8)$) {\sicd};
	\node[non] (p2) at ($(b1)+(1.5,0.8)$) {\picd};
	\node[non] (b2) at ($(b1)+(3.0,0.0)$) {\bscd};
	
	\node[non] (s3) at ($(b2)+(1.5,-0.8)$) {\sict};
	\node[non] (p3) at ($(b2)+(1.5,0.8)$) {\pict};
	\node[non] (b3) at ($(b2)+(3.0,0.0)$) {\bsct};
	
	\node[non] (s4) at ($(b3)+(1.5,-0.8)$) {\sic{4}};
	\node[non] (p4) at ($(b3)+(1.5,0.8)$) {\pic{4}};

	\draw[thick] (b0.-60) to [out=-90,in=180] node[tag] {\scriptsize
		$\subsetneq$} (s1.west);
	\draw[thick] (b0.60) to [out=90,in=-180] node[tag] {\scriptsize
		$\subsetneq$} (p1.west);
	
	\draw[thick] (s1.east) to [out=0,in=-90] node[tag] {\scriptsize
		$\subsetneq$} (b1.-120);
	\draw[thick] (p1.east) to [out=0,in=90] node[tag] {\scriptsize
		$\subsetneq$} (b1.120);

	\draw[thick] (b1.-60) to [out=-90,in=180] node[tag] {\scriptsize
		$\subsetneq$} (s2.west);
	\draw[thick] (b1.60) to [out=90,in=-180] node[tag] {\scriptsize
		$\subsetneq$} (p2.west);
	\draw[thick] (s2.east) to [out=0,in=-90] node[tag] {\scriptsize
		$\subsetneq$} (b2.-120);
	\draw[thick] (p2.east) to [out=0,in=90] node[tag] {\scriptsize
		$\subsetneq$} (b2.120);

	\draw[thick] (b2.-60) to [out=-90,in=180] node[tag] {\scriptsize
		$\subsetneq$} (s3.west);
	\draw[thick] (b2.60) to [out=90,in=-180] node[tag] {\scriptsize
		$\subsetneq$} (p3.west);
	\draw[thick] (s3.east) to [out=0,in=-90] node[tag] {\scriptsize
		$\subsetneq$} (b3.-120);
	\draw[thick] (p3.east) to [out=0,in=90] node[tag] {\scriptsize
		$\subsetneq$} (b3.120);
	
	\draw[thick] (b3.-60) to [out=-90,in=180] node[tag] {\scriptsize
		$\subsetneq$} (s4.west);
	\draw[thick] (b3.60) to [out=90,in=-180] node[tag] {\scriptsize
		$\subsetneq$} (p4.west);

	\draw[thick,dotted] ($(s4.east)+(0.1,0.0)$) to
	($(s4.east)+(0.6,0.0)$);
	\draw[thick,dotted] ($(p4.east)+(0.1,0.0)$) to
	($(p4.east)+(0.6,0.0)$);
	
	\end{tikzpicture}
	\caption{Quantifier Alternation Hierarchy}
	\label{fig:hiera}
\end{figure}

It was shown by Thomas~\cite{ThomEqu} that this hierarchy  corresponds exactly to the dot-depth hierarchy. For any $n \in \nat$, dot-depth $n$ corresponds to $\bsws{n}$ and dot-depth $n+1/2$ to $\siws{n+1}$. Naturally, this means that the logic  \piws{n+1} corresponds to the complement class of dot-depth $n+1/2$ (a language is definable in \piws{n+1} if and only if its complement has dot-depth $n+1/2$). Therefore, Corollaries~\ref{cor:seventwodd} and~\ref{cor:ninetwodd} may be translated as follows from a logical point of view.

\begin{corollary} 
	Separation and covering are decidable for $\siws{3}$.
\end{corollary}
\begin{corollary} 
	It decidable to test whether a regular language may be defined by a sentence of  $\siws{4}$.
\end{corollary}

Finally, it is possible to adapt the result of Thomas to obtain a similar correspondence between the Straubing-Th\'erien hierarchy and the quantifier alternation hierarchy within a variant of first-order logic using a smaller signature. We denote by \fow the variant of first-order logic whose signature only contains the label predicates and the linear order (i.e. $+1,min,max$ and $\varepsilon$ are disallowed).

\begin{remark}
	This restriction does not change the expressive power of first-order logic as a whole: $+1,min,max$ and $\varepsilon$ can be defined from ``$<$''. For example, the formula $y = x +1$ is equivalent to $(x < y) \wedge \neg(\exists z\ x < z < y)$. However, this definition costs quantifier alternations. Thus, \fow and \fows have different quantifier alternation hierarchies.
\end{remark}

One may define a quantifier alternation hierarchy for \fow in the natural way. We shall denote by \siw{n}, \piw{n} and \bsw{n} its levels. It was shown by Perrin and Pin~\cite{PPOrder} that this second alternation hierarchy corresponds exactly to the Straubing-Th\'erien hierarchy. Therefore, Corollaries~\ref{cor:seventwo} and~\ref{cor:ninetwo} may be translated as follows.

\begin{corollary} \label{cor:siwt}
	Separation and covering are decidable for \siwt.
\end{corollary}

\begin{corollary} \label{cor:siwq}
	It decidable to test whether a regular language is definable in \siw{4}.
\end{corollary}

\begin{remark}
	These two correspondences between alternation and concatenation hierarchies are not coincidental. It turns out that given  any basis \Cs, one may define  a set of first-order predicates $S$ such that the concatenation hierarchy of basis \Cs corresponds exactly to the alternation hierarchy within the variant of first-order logic with signature is $S$ (see~\cite{pzgenconcat}).
\end{remark}


\section{Mathematical tools}
\label{sec:tools}
In this section, we present mathematical tools which we shall use later when proving the correctness of our covering algorithms. First, we briefly recall the algebraic definition of regular languages based on monoid morphisms. We complete this definition with the \emph{factorization forest theorem of Simon}, a combinatorial result about finite monoids. Finally, we present a generic stratification for classes of the form \pol{\Cs} when \Cs is a finite \pvari.

\subsection{Monoids and regular languages} A \emph{semigroup} is a set $S$ equipped with an associative multiplication (usually denoted by ``$\cdot$''). Recall that an idempotent within a semigroup $S$ is an element $e \in S$ such that $ee = e$. It is well-known that when a semigroup $S$ is \emph{finite}, there exists a number $\omega(S)$ (denoted by $\omega$ when $S$ is understood) such that for any $s \in S$, $s^\omega$ is an idempotent.

A \emph{monoid} $M$ is a semigroup in which there exists a neutral element denoted $1_M$. Observe that $A^*$ is a monoid with concatenation as the multiplication and $\varepsilon$ as the neutral element. Hence, given a monoid $M$, we may consider morphisms $\alpha: A^* \to M$ (i.e. $\alpha(\varepsilon) = 1_M$ and $\alpha(uv) = \alpha(u) \cdot \alpha(v)$ for all $u,v \in A^*$). Given a morphism $\alpha$, we say that a language $L \subseteq A^*$ is \emph{recognized} by $\alpha$ if there exists $F \subseteq M$ such that $L = \alpha\inv(F)$. We call $F$ the accepting set of $L$. It is well-known that a language $L \subseteq A^*$ is regular if and only if it is recognized by a morphism into a \emph{finite} monoid. Moreover, one may compute such a morphism from any representation of $L$ (such as an automaton).

\begin{remark}
	For all morphisms $\alpha: A^* \to M$ that we consider in the paper, $M$ will be finite. Hence, for the sake of avoiding clutter, we make this implicit: whenever we say that ``$\alpha: A^* \to M$ is a morphism'', this implicitly means that $M$ is a finite monoid.
\end{remark}

\medskip
\noindent
{\bf \Cs-compatible morphisms.} We complete these definitions with an additional notion which is specific to the paper. Our main objective is to prove Theorem~\ref{thm:maintheo}: we want an algorithm for \pbpol{\Cs}-covering where \Cs is an arbitrary finite \vari. This will require working with morphisms which satisfy a specific property called \emph{\Cs-compatibility}. We fix the finite \vari \Cs for the definition. 

Recall that since \Cs is a finite Boolean algebra, one may associate a canonical equivalence $\sim_\Cs$ over $A^*$ to \Cs. Two words are equivalent when they belong to the same languages in \Cs:
\[
w \sim_\Cs w' \quad \text{if and only if} \quad \forall L\in\Cs,\ w \in L \Leftrightarrow w' \in L.
\]
Given a word $w$, we denote by \ctype{w} its $\sim_\Cs$-equivalence class. Recall that by Lemma~\ref{lem:canoequiv}, $\sim_\Cs$ has finite index and the languages in \Cs are exactly the unions of $\sim_\Cs$-classes. Moreover, since \Cs is \quotienting, we know from Lemma~\ref{lem:canoquo} that $\sim_\Cs$ is a congruence for word concatenation. It follows that the quotient set ${A^*}/{\sim_\Cs}$ is a finite monoid and the map $w \mapsto \ctype{w}$ is a morphism from $A^*$ to ${A^*}/{\sim_\Cs}$.

We may now define \Cs-compatibility. Consider an arbitrary morphism $\alpha: A^* \to M$. We say that $\alpha$ is \emph{\Cs-compatible} when for any $s \in M$, there exists a $\sim_\Cs$-class denoted by $\ctype{s}$ such that $\alpha\inv(s) \subseteq \ctype{s}$ (i.e. $\ctype{w} = \ctype{s}$ for any $w \in A^*$ such that $\alpha(w) = s$).

\begin{remark}
	Given $s \in M$, the $\sim_{\Cs}$-class \ctype{s} is fully determined by $\alpha$ when $\alpha\inv(s) \neq \emptyset$ ($\ctype{s} = \ctype{w}$ for any $w \in \alpha\inv(s)$). If $\alpha\inv(s) = \emptyset$, we may choose any $\sim_\Cs$-class as \ctype{s}. When we consider a \Cs-compatible morphism, we implicitly assume that the map $s \mapsto \ctype{s}$ is fixed.
\end{remark}

We now prove that for any regular language $L$, one may compute a \Cs-compatible morphism which recognizes $L$.

\begin{lemma} \label{lem:compat}
	Given a regular language $L \subseteq A^*$, one may compute a \Cs-compatible morphism $\alpha: A^* \to M$ recognizing $L$.
\end{lemma}

\begin{proof}
	Since $L$ is regular, we may compute a finite monoid $N$ and a morphism $\beta: A^* \to N$ (not necessarily \Cs-compatible) recognizing $L$. Since we know that the quotient set ${A^*}/{\sim_\Cs}$ is a finite monoid, the Cartesian product $M = N \times ({A^*}/{\sim_\Cs})$ is a finite monoid for the componentwise multiplication. It now suffices to define the morphism $\alpha: A^* \to M$ by $\alpha(w) = (\beta(w),\ctype{w})$ for any $w \in A^*$. Clearly, $\alpha$ is a morphism which recognizes $L$ and one may verify that it is \Cs-compatible.
\end{proof}

\subsection{Factorization forests} The factorization theorem of Simon is a combinatorial result which applies to \emph{finite semigroups}. Here, we briefly recall this theorem. We refer the reader to~\cite{kfacto,bfacto,cfacto} for more details and a proof.

Consider a morphism $\alpha: A^* \rightarrow M$. An \emph{$\alpha$-factorization forest} is an ordered unranked rooted tree whose nodes are labeled by words in $A^*$. For any inner node $x$ with label $w \in A^*$, if $w_1,\dots,w_n \in A^*$ are the labels of its children listed from left to right, then $w = w_1\cdots w_n$. Moreover, all nodes $x$ in the forest must be of the three following kinds:

\begin{itemize}
	\item \emph{Leaves} which are labeled by either a single letter or
	the empty word.
	\item \emph{Binary inner nodes} which have exactly two children.
	\item \emph{Idempotent inner nodes} which have three or more children. The labels $w_1,\dots,w_n$ of these children must satisfy $\alpha(w_1) = \cdots = \alpha(w_n) = e$ where $e$ is an idempotent element of $M$.
\end{itemize}

Given a word $w \in A^*$, an \emph{$\alpha$-factorization forest for $w$} is an $\alpha$-factorization forest whose root is labeled by $w$. The \emph{height} of a factorization forest is the largest $h \in \nat$ such that it contains a branch with $h$ inner nodes (a single leaf has height $0$). We shall also consider the notion of \emph{idempotent height}. The idempotent height of an $\alpha$-factorization forest is the largest number $m \in \nat$ such that there exists a branch containing $m$ idempotent nodes in the forest. We shall use the following notations. Let $s \in M$ and $h,m \in \nat$:
\begin{enumerate}
	\item $F^\alpha(s,h,m) \subseteq A^*$ denotes the language of all words $w \in \alpha^{-1}(s)$ admitting an $\alpha$-factorization forest of height at most $h$ and idempotent height at most $m$.
	\item $F^\alpha_B(s,h,m) \subseteq A^*$ denotes the language of all words $w \in \alpha^{-1}(s)$ admitting an $\alpha$-factorization forest of height at most $h$, of idempotent height at most $m$ and whose root is a \emph{binary node}.
	\item $F^\alpha_I(s,h,m) \subseteq A^*$ denotes the language of all words $w \in \alpha^{-1}(s)$ admitting an $\alpha$-factorization forest of height at most $h$, of idempotent height at most $m$ and whose root is a \emph{idempotent node}.	
\end{enumerate}

We have the following fact.

\begin{fct} \label{fct:factounion}
	For every morphism $\alpha: A^* \rightarrow M$, every $s \in M$ and every $h,m \in \nat$, we have,
	\[
	F^\alpha(s,h,m) = F^\alpha_B(s,h,m) \cup F^\alpha_I(s,h,m) \cup F^\alpha(s,0,0)
	\]
\end{fct}

\begin{proof}
	The right to left inclusion is immediate by definition. For the converse inclusion, consider $w \in F^\alpha(s,h,m)$. By definition, $w \in \alpha^{-1}(s)$ and admits an $\alpha$-factorization forest of height at most $h$ and idempotent height at most $m$. The root of this forest is either a binary node, an idempotent node or a leaf. In the first case, $w \in F^\alpha_B(s,h,m)$. In the second one, $w \in F^\alpha_I(s,h,m)$. Finally, the third case implies that $w \in F^\alpha(s,0,0)$.	
\end{proof}

We turn to the factorization forest theorem of Simon: there exists a bound depending only on $M$ such that any word admits an $\alpha$-factorization forest of height at most this bound.

\begin{theorem}[\cite{simonfacto,kfacto}] \label{thm:facto}
	Consider a morphism $\alpha: A^* \rightarrow M$. For all words $w \in A^*$, there exists an $\alpha$-factorization forest for $w$ of height at most $3|M|-1$.
\end{theorem}

Observe that an immediate consequence of Theorem~\ref{thm:facto} is that for any $s \in M$, we have $\alpha^{-1}(s) = F^\alpha(s,3|M|-1,3|M|-1)$. We shall use Theorem~\ref{thm:facto} conjointly with a second result on factorization forests which takes the idempotent height into account. Given a word $w \in A^*$, an \emph{infix} of $w$ is a second word $u \in A^*$ such that $w = v_1uv_2$ for some $v_1,v_2 \in A^*$.

\begin{proposition} \label{prop:idheight}
	Consider a morphism $\alpha: A^* \to M$ and let $h,m \in \nat$. Let $w \in A^*$ admitting an $\alpha$-factorization forest of height at most $h$ and idempotent height at most $m$. Then any infix of $w$ admits an $\alpha$-factorization forest of height at most $h + 2$ and idempotent height at most $m$.
\end{proposition}

Let us prove Proposition~\ref{prop:idheight}. The argument is an induction on the $\alpha$-factorization forest of $w$: we show that it may be ``repaired'' into another forest for some infix of $w$ without adding any idempotent node. We first consider the special case of prefixes and suffixes.

\begin{lemma} \label{lem:idheight}
	Let $h,m \in \nat$ and consider a word $w \in A^*$ admitting an  $\alpha$-factorization forest of height at most $h$ and idempotent height at most $m$. Then any prefix or suffix of $w$ admits an $\alpha$-factorization forest of height at most $h+1$ and idempotent height at most $m$.
\end{lemma}

Observe that any infix of a word $w \in A^*$ is by definition the prefix of a suffix of $w$. Hence, Proposition~\ref{prop:idheight} is obtained by applying Lemma~\ref{lem:idheight} twice (once for prefixes and once for suffixes). Therefore, we may concentrate on proving Lemma~\ref{lem:idheight}. We only treat the case of prefixes (the argument for suffixes is symmetrical).

Let $w \in A^*$ admitting an $\alpha$-factorization forest of height at most $h$ and idempotent height at most $m$. We shall denote this forest by \Fs. Consider a prefix $u$ of $w$. We construct an $\alpha$-factorization forest for $u$ using structural induction on \Fs. There are three cases depending on the root node of \Fs. If the root of \Fs is a leaf node, then $w = a$ or $w = \varepsilon$. In particular, $0 \leq h$ and $0 \leq m$. Since $u$ is a prefix of $w$, we have $u = a$ or $u = \varepsilon$. Hence $u$ admits an $\alpha$-factorization forest of height $0 \leq h$ and idempotent height $0 \leq m$.
	
If the root of \Fs is a binary node, it has two children labeled with words $w_1,w_2 \in A^*$ such that $w = w_1w_2$. Moreover, $w_1,w_2$ admit $\alpha$-factorization forests of heights at most $h-1$ and idempotent heights at most $m$. Since $u$ is a prefix of $w = w_1w_2$ by hypothesis, there are two possibles cases. First, $u$ may be a prefix of $w_1$. In this case, it suffices to apply induction to the forest of $w_1$ to get the desired forest for $u$. Otherwise, there exists a prefix $u'$ of $w_2$ such that $u = w_1u'$. We may apply induction to the forest of $w_2$ to get a forest of height at most $h-1+1 = h$ and idempotent height at most $m$ for $u'$. Using one binary node, one may then combine the forests of $w_1$ and $u'$ into a single forest for $u =w_1 u'$. By construction, this new forest has height at most $h + 1$ and idempotent height at most $m$.

We finish with the case when the root of \Fs is an idempotent node. Its children are labeled with words $w_1,\dots,w_n$ such that $w = w_1 \cdots w_n$ and $\alpha(w_1) = \cdots = \alpha(w_n) = e$ for some idempotent $e \in M$. Moreover, the words $w_1,\cdots,w_n$ admit $\alpha$-factorization forests of heights at most $h-1$ and idempotent heights at most $m-1$. Let $i$ be the smallest natural such that $u$ is a prefix of $w_1 \cdots w_i$. If $i = 1$, $u$ is a prefix of $w_1$ and it suffices to apply induction to the forest of $w_1$ to get the desired forest for $u$. Otherwise, there exists a prefix $u'$ of $w_i$ such that $u = w_1 \cdots w_{i-1}u'$. We know that, 
\begin{itemize}
\item $u'$ admits an $\alpha$-factorization forest of height at most $(h-1) + 1 = h$ and idempotent height at most $m-1$ (this is by	induction).
\item $w_{1} \cdots w_{i-1}$ admit an $\alpha$-factorization forest of height at most $h$ and idempotent height at most $m$ (the root is an idempotent node whose children are labeled with $w_{1}, \dots, w_{i-1}$).
\end{itemize}
These two forests can be combined into a single forest for $u$ with one binary node. By definition, this forest has height at most $h+1$ and idempotent height at most $m$.

\subsection{A stratification for \pol{\Cs}}

We turn to the second mathematical tool that we shall need: a stratification for classes of the form \pol{\Cs} when \Cs is a finite \pvari.

\begin{remark}
	We do not present a stratification for \pbpol{\Cs} (even if our ultimate objective is to solve \pbpol{\Cs}-covering). Indeed, it turns out that the stratification of \pol{\Cs} that we present here suffices to handle both \pol{\Cs} and \pbpol{\Cs}.
\end{remark}

We fix an arbitrary finite \pvari \Cs and define a finite \pvari \polk{\Cs} for each $k \in \nat$. The definition uses induction on $k$ and counts the number of marked concatenations that are necessary to define each language in \pol{\Cs}.
\begin{itemize}
	\item When $k = 0$, we simply define $\polp{\Cs}{0} = \Cs$.
	\item When $k \geq 1$, \polk{\Cs} is the least lattice such that:
	\begin{enumerate}
		\item $\polp{\Cs}{k-1} \subseteq \polk{\Cs}$.
		\item For any $a \in A$ and $L_1,L_2 \in \polp{\Cs}{k-1}$, we have $L_1 a L_2 \in \polk{\Cs}$.
	\end{enumerate}
\end{itemize}

This concludes the definition. Since \Cs was a finite lattice, it is immediate that all classes \polk{\Cs} are finite lattices as well. Moreover, by Theorem~\ref{thm:pclos}, we indeed have,
\[
\text{For all $k \in \nat$, } \polk{\Cs} \subseteq \polp{\Cs}{k+1} \quad \text{and} \quad \pol{\Cs} = \bigcup_{k \in \nat} \polk{\Cs}
\]
Finally, it is straightforward to verify that all classes \polk{\Cs} are \quotienting (the argument is identical to the one used for proving that the whole class \pol{\Cs} is \quotienting, see Lemma~27 in~\cite{pzgenconcat} for a proof). Thus, we did define a stratification of \pol{\Cs}. We now prove a few properties of this stratification that we shall need.

\medskip

Consider the canonical preorder relations associated to the strata. For any $k \in \nat$, we denote by \polrelk the preorder associated to \polk{\Cs}. Recall that for any $w,w' \in A^*$, we have,
\[
w \polrelk w'   \quad \text{ if and only if} \quad  \text{for all languages $L \in \polk{\Cs}$,\ $w \in L \Rightarrow w' \in L$} 
\]
We showed in Lemma~\ref{lem:canosatur} that the languages in \polk{\Cs} are exactly the upper sets for the relation $\polrelk$. Moreover, since the lattices \polk{\Cs} are \quotienting, we know from Lemma~\ref{lem:canoquo} that the relations \polrelk are compatible with word concatenation.

\medskip

We present two specific properties of the preorders \polrelk. We start with an alternate definition of \polrelk which is easier to manipulate when proving these two properties. Recall that we write \canoc for the canonical preorder associated to the finite \pvari \Cs.

\begin{lemma}\label{lem:hintro:preoinduc}
	Let $k \in \nat$. For any $w,w' \in A^*$, we have $w \polrelk w'$ if and only if the two following properties hold:
	\begin{enumerate}
		\item $w \canoc w'$
		\item If $k \geq 1$, for any decomposition $w = uav$ with $u,v \in A^*$ and $a \in A$, there exist $u',v' \in A^*$ such that $w' = u'av'$, $u \polrelp{k-1} u'$ and $v \polrelp{k-1} v'$.
	\end{enumerate}
\end{lemma}

\begin{proof}
	Assume first that $w \polrelk w'$. We have to prove that the two items in the lemma hold. For the first item, observe that by definition, $\Cs \subseteq \polk{\Cs}$. Therefore, $w \polrelk w' \Rightarrow w \canoc w'$. We turn to the second item. Assume that $k \geq 1$ and consider a decomposition $w = uav$ of $w$. We have to find an appropriate decomposition of $w'$. Let $K_u = \{u' \in A^* \mid u \polrelp{k-1} u'\}$ and $K_v = \{v' \in A^* \mid v \polrelp{k-1} v'\}$. By definition $K_u,K_v$ are upper sets for \polrelp{k-1} and it follows from Lemma~\ref{lem:canosatur} that $K_u,H_v \in \polp{\Cs}{k-1}$. Hence, $K_u aK_v \in \polk{\Cs}$ by definition. Moreover, since $w = uav \in  K_u aK_v$ and $w \polrelk w'$, it follows that $w' \in K_u aK_v$. Therefore, we obtain $u' \in K_u$ and $v' \in K_v$ such that $w' = u'av'$. It is then immediate by definition of $K_u$ and $K_v$ that $u \polrelp{k-1} u'$ and $v \polrelp{k-1} v'$.
	
	We turn to the other direction. Let $k \in \nat$. Assuming that the two assertions in the lemma hold for $k$, we prove that $w \polrelk w'$. The argument is an induction  using induction on $k$. When $k = 0$, this is immediate from the first item since $\polp{\Cs}{0} = \Cs$ (and therefore, $\polrelp{0}$ and \canoc are the same relation). We now assume that $k \geq 1$. Let $L \in \polk{\Cs}$, we have to prove that $w \in L \Rightarrow w' \in L$. By definition, $L$ is constructed by applying finitely many unions and intersections to the two following kinds of languages:
	\begin{enumerate}
		\item Languages in \polp{\Cs}{k-1}.
		\item Languages of the form $L_1aL_2$ with $L_1,L_2 \in \polp{\Cs}{k-1}$.	
	\end{enumerate}
	We use a sub-induction on this construction.
	\begin{itemize}
		\item When $L \in \polp{\Cs}{k-1}$ the implication is immediate. Since \polrelk is finer than \polrelp{k-1}, the assertions in the lemma also hold for $k-1$. We get $w \polrelp{k-1} w'$ by induction on $k$.
		\item Assume now that $L = L_1aL_2$ with $L_1,L_2 \in \polp{\Cs}{k-1}$. If $w \in L = L_1aL_2$, then it admits a decomposition $w = uav$ with $u \in L_1$ and $v \in L_2$. By the second item, we obtain $u',v' \in A^*$ such that $w' = u'av'$, $u \polrelp{k-1} u'$ and $v \polrelp{k-1} v'$. In particular, since $L_1,L_2 \in \polp{\Cs}{k-1}$, it follows by definition of $\polrelp{k-1}$ that $u' \in L_1$ and $v' \in L_2$, i.e. $w' =u'av' \in L_1aL_2 = L$.
		\item Finally, if $L = L_1 \cup L_2$ or $L = L_1 \cap L_2$, induction yields that $w \in L_1 \Rightarrow w' \in L_1$ and $w \in L_2 \Rightarrow w' \in L_2$ and therefore, $w \in L \Rightarrow w' \in L$.			
	\end{itemize}
	This terminates the proof of Lemma~\ref{lem:hintro:preoinduc}.
\end{proof}

We now use Lemma~\ref{lem:hintro:preoinduc} to present and prove two characteristic properties of the preorders \polrelk which we shall use multiple times. Let us first present the following fact which introduces a characteristic natural number of \Cs that is involved in both properties.

\begin{fct}\label{fct:omegapower}
	There exists a natural number $p \geq 1$ such that for any word $u \in A^*$ and natural numbers $m,m' \geq 1$, we have $u^{pm}  \canoc  u^{pm'}$.
\end{fct}

\begin{proof}
	Let $\sim$ be the equivalence generated by \canoc. That is, for any $w,w' \in A^*$, $w \sim w'$ if and only if $w \canoc w'$ and $w' \canoc w$. We know from Lemmas~\ref{lem:canosatur} and~\ref{lem:canoquo} that \canoc has finitely many upper sets and is compatible with word concatenation. Therefore, $\sim$ is a congruence of finite index for word concatenation. It follows that the quotient set ${A^*}/{\sim}$ is a finite monoid. We let $p = \omega({A^*}/{\sim})$ (the idempotent power of ${A^*}/{\sim}$). The fact is now immediate.
\end{proof}

We shall call the natural number $p \geq 1$ described in Fact~\ref{fct:omegapower} the \emph{period of the \pvari \Cs}. We are now ready to state the first of our two properties.

\begin{lemma} \label{lem:hintro:propreo1}
	Let $p$ be the period of \Cs and $k \in \nat$. Then, for any $m,m' \geq 2^{k+1}-1$ and any word $u \in A^*$, we have $u^{pm} \polrelk u^{pm'}$.
\end{lemma}

\begin{proof}
	Let $m,m' \geq 2^{k+1}-1$ and $u$ some word, we prove that $u^{pm} \polrelk u^{pm'}$. This amounts to proving that the two items in Lemma~\ref{lem:hintro:preoinduc} hold. The argument is an induction on $k$. For the first item, it suffices to prove that $u^{pm} \canoc u^{pm'}$. This is immediate by choice of $p$ in Fact~\ref{fct:omegapower}. This concludes the case $k = 0$.
	
	We now consider the second item (which may only happen when $k \geq 1$). Consider a decomposition $u^{pm}  = w_1aw_2$. We have to find a decomposition $u^{pm'} = w'_1aw'_2$ such that $w_1 \polrelp{k-1} w'_1$ and $w_2 \polrelp{k-1} w'_2$. By definition, the letter $a$ in the decomposition $u^{pm}  = w_1aw_2$ falls within some factor $u^p$ of $u^{pm}$. Let us refine the decomposition to isolate this factor. We have $u^{pm}  = u^{pm_1}v_1av_2u^{pm_2}$ where,
	\begin{itemize}
		\item $m = m_1+1+m_2$
		\item $v_1av_2 = u^p$.
		\item $u^{pm_1}v_1 = w_1$ and $v_2u^{pm_2} = w_2$.
	\end{itemize}
	Since $m \geq 2^{k+1}-1$ by hypothesis and $m = m_1+1+m_2$, either $m_1 \geq 2^{k}-1$ or $m_2 \geq 2^{k}-1$ (possibly both). By symmetry, let us assume that $m_1 \geq 2^{k} - 1$. We use the following claim.
	
	\begin{claim}
		There exist $m'_1,m'_2 \geq 1$ such that $m' = m'_1+1+m'_2$, $u^{pm_1} \polrelp{k-1} u^{pm'_1}$ and $u^{pm_2} \polrelp{k-1} u^{pm'_2}$.
	\end{claim}
	
	\begin{proof}
		There are two cases depending on whether $m_2 \geq 2^{k}-1$ or not. Assume first that, $m_2 \geq 2^{k}-1$. Since $m' \geq 2^{k+1}-1$, we may choose $m'_1,m'_2 \geq 2^k-1$ such that $m' = m'_1 + 1 + m'_2$. It is now immediate from induction on $k$ that $u^{pm_1} \polrelp{k-1}{n} u^{pm'_1}$ and $u^{pm_2} \polrelp{k-1} u^{pm'_2}$. Otherwise, $m_2 < 2^{k}-1$. We let $m'_2 = m_2$ and $m'_1 = m' - 1 - m'_2$. Clearly, $m'_1 \geq 2^{k} -1$ since $m' \geq 2^{k+1}-1$. Hence, we get $u^{pm_1} \polrelp{k-1} u^{pm'_1}$ from induction on $k$. Furthermore, $u^{pm_2} \polrelp{k-1} u^{pm'_2}$ is immediate since $m_2 = m'_2$ by definition.
	\end{proof}
	
	We may now finish the proof of Item~2. Let $m'_1,m'_2 \geq 1$ be as defined in the claim. We let $w'_1 = u^{pm'_1}v_1$ and $w'_2 = v_2u^{pm'_2}$. Clearly, $w'_1 aw'_2 = u^{pm'}$ since $v_1 av_2 = u^p$ and $m' = m'_1+1+m'_2$. Moreover, since \polrelp{k-1} is compatible with multiplication, we have
	\[
	\begin{array}{rll}
	w_1 = u^{pm_1}v_1 & \polrelp{k-1} & u^{pm'_1}v_1 = w'_1 \\
	w_2 = v_2u^{pm_2} & \polrelp{k-1} & v_2u^{pm'_2} = w'_2
	\end{array}
	\]
	This terminates the proof of Item~2.
\end{proof}

We turn to the second lemma which states a characteristic property of classes built with polynomial closure.

\begin{lemma} \label{lem:hintro:propreo2}
	Let $p$ be the period of \Cs and $k \in \nat$. Let $u,v \in A^*$ such that $u^p \canoc v$. Then, for any $m,m'_1,m'_2 \geq 2^{k+1}-1$, we have $u^{pm} \polrelk u^{pm'_1}vu^{pm'_2}$.
\end{lemma}

\begin{proof}
	The proof is similar to that of Lemma~\ref{lem:hintro:propreo1}. Let $k \in \nat$, $u,v$ satisfying $v \canoc u^p$, and $m,m'_1,m'_2 \geq 2^{k+1}-1$. We prove that $u^{pm} \polrelk u^{pm'_1}vu^{pm'_2}$. This amounts to proving that the two items in Lemma~\ref{lem:hintro:preoinduc} hold. The argument is an induction on $k$.
	
	For Item~1, we prove that $u^{pm} \canoc u^{pm'_1}vu^{pm'_2}$. By hypothesis on $u,v$, we know that $u^p \canoc v$. Therefore, since \canoc is compatible with concatenation, we get that $u^{p(m'_1+ 1 + m'_2)} \canoc u^{pm'_1}vu^{pm'_2}$. Finally, we obtain by choice of $p$ in Fact~\ref{fct:omegapower} that $u^{pm} \canoc u^{p(m'_1+ 1 + m'_2)}$. This finishes the proof of Item~1 (and the case $k = 0$) by transitivity.
	
	We now consider the second item (which may only happen when $k \geq 1$). Consider a decomposition $u^{pm}  = w_1aw_2$. We have to find a decomposition $u^{pm'_1}vu^{pm'_2} = w'_1aw'_2$ such that $w_1 \polrelp{k-1} w'_1$ and $w_2 \polrelp{k-1} w'_2$. By definition, the letter $a$ in the decomposition $u^{pm}  = w_1aw_2$ falls within some factor $u^p$ of $u^{pm}$. Let us refine the decomposition to isolate this factor. We have $u^{pm}  = u^{pm_1}v_1av_2u^{pm_2}$ where,
	\begin{itemize}
		\item $m = m_1+1+m_2$
		\item $v_1av_2 = u^p$.
		\item $u^{pm_1}v_1 = w_1$ and $v_2u^{pm_2} = w_2$.
	\end{itemize}
	Since $m \geq 2^{k+1}-1$ by hypothesis and $m = m_1+1+m_2$, either $m_1 \geq 2^{k}-1$ or $m_2 \geq 2^{k}-1$ (possibly both). By symmetry, let us assume that $m_1 \geq 2^{k} - 1$. We use the following claim.
	
	\begin{claim}
		There exist $\ell'_1,\ell'_2 \in \nat$ such that $m'_2 = \ell'_1+1+\ell'_2$, $u^{pm_1} \polrelp{k-1} u^{pm'_1}vu^{p\ell'_1}$ and $u^{pm_2} \polrelp{k-1} u^{p\ell'_2}$.
	\end{claim}
	
	\begin{proof}
		There are two cases depending on whether $m_2 \geq 2^{k}-1$ or not. Assume first that, $m_2 \geq 2^{k}-1$. Since $m'_2 \geq 2^{k+1}-1$, we may choose $\ell'_1,\ell'_2 \geq 2^k-1$ such that $m'_2 = \ell'_1 + 1 + \ell'_2$. That $u^{pm_1} \polrelp{k-1} u^{pm'_1}vu^{p\ell'_1}$ follows from induction on $k$. Moreover, we know that the inequality $u^{pm'_2} \polrelp{k-1} u^{p\ell'_2}$ holds thanks to Lemma~\ref{lem:hintro:propreo1}.
		
		Otherwise, $m_2 < 2^{k}-1$. We let $\ell'_2 = m_2$ and $\ell'_1 = m'_2 - 1 - \ell'_2$. Clearly, $\ell'_1 \geq 2^{k} -1$ since $m'_2 \geq 2^{k+1}-1$. Hence, we get $u^{pm_1} \polrelp{k-1} u^{pm'_1}vu^{p\ell'_1}$ from induction on $k$. Furthermore, $u^{pm_2} \polrelp{k-1} u^{p\ell'_2}$ is immediate since $m_2 = \ell'_2$ by definition.	
	\end{proof}
	
	We may now finish the proof of Item~2. Let $\ell'_1,\ell'_2 \geq 1$ be as defined in the claim. We let $w'_1 = u^{pm'_1}vu^{p\ell'_1}v_1$ and $w'_2 = v_2u^{p\ell'_2}$. Clearly, $w'_1 aw'_2 = u^{pm'_1}vu^{pm'_2}$ since $v_1 av_2 = u^p$ and $m_2' = \ell'_1+1+\ell'_2$. Moreover, since \polrelp{k-1} is compatible with multiplication, we have
	\[
	\begin{array}{rll}
	w_1 = u^{pm_1}v_1 & \polrelp{k-1} & u^{pm'_1}vu^{p\ell'_1}v_1 = w'_1 \\
	w_2 = v_2u^{pm_2} & \polrelp{k-1} & v_2u^{p\ell'_2} = w'_2
	\end{array}
	\]
	This terminates the proof of Item~2.
\end{proof}


\section{Framework: \ratms and optimal covers}
\label{sec:covers}
We now present the general framework that we use for formulating our covering algorithms. The notions that we introduce here were originally designed in~\cite{pzcovering2} and we only recall what we need for presenting our results. We refer the reader to~\cite{pzcovering2} for details and background on these notions.

Given a lattice \Ds, our input for the \Ds-covering problem is a pair $(L,\Lb)$ where $L$ is a regular language and \Lb a finite set of regular languages: we want to know whether there exists a \Ds-cover of $L$ which is separating for \Lb. The key idea behind our framework is to replace the set \Lb by a (more general) algebraic object called \emph{\mratm}. Intuitively, \mratms are designed to measure the quality of \Ds-covers. Given a \mratm $\rho$ and a language $L$, we use $\rho$ to rank the existing \Ds-covers of $L$. This leads to the definition of ``{\bf optimal}'' \Ds-cover of $L$. We are able to reformulate our problem with these notions. The key idea is that instead of deciding whether $(L,\Lb)$ is \Ds-coverable, we compute an optimal \Ds-cover for $L$ for a \mratm $\rho$ that we build from $\Lb$. We have two motivations for relying on this approach,
\begin{enumerate}
	\item It yields elegant formulations for covering algorithms. Beyond the two that we present in the paper, we refer the reader to~\cite{pzcovering2} for more examples. 
	\item More importantly, recall that our main goal in the paper is \pbpol{\Cs}-covering.   As we already explained, this will require to first prove a preliminary result for \pol{\Cs} which is stronger than the decidability of \pol{\Cs}-covering (this is why we reprove the result of~\cite{pzboolpol} for \pol{\Cs}-covering along the way). The framework presented here is exactly what we need in order to precisely formulate this stronger result.
\end{enumerate}

We start by defining \mratms. Then, we explain how they are used to measure the quality of a cover and define optimal covers. Finally, we connect these notions to the covering problem. Let us point out that several statements presented here are without proof, we refer the reader to~\cite{pzcovering2} for these proofs.

\subsection{\Mratms}

In order to present \mratms, we need to introduce a new algebraic structure: \emph{hemirings}. A hemiring is a set $R$ equipped with two binary operations called addition (``$+$'') and multiplication (``$\cdot$'') respectively. Moreover, the following axioms have to be satisfied:
\begin{itemize}
	\item $R$ is a commutative monoid for addition. The neutral element is denoted by $0_R$.
	\item $R$ is a semigroup for multiplication.
	\item Multiplication distributes over addition: for all $r,s,t \in R$ we have,
	\[
	\begin{array}{lll}
	t \cdot (r + s) & = & (t \cdot r) + (t \cdot s) \\
	(r + s) \cdot t & = & (r \cdot t) + (s \cdot t)
	\end{array}
	\]
	\item $0_R$ is a zero for the multiplication: for any $r \in R$:
	\[
	0_R \cdot r = r \cdot 0_R = 0_R
	\]	
\end{itemize}

\begin{remark}
	Hemirings are a generalization of the more standard notion of semiring which additionally asks for the multiplication to have a neutral element.  
\end{remark}

Finally, a hemiring $R$ is \emph{idempotent} when all elements are idempotents for addition: for all $r \in R$, we have $r + r = r$. In the paper, we only work with idempotent hemirings. Observe that when $R$ is such a hemiring, one may define a \emph{canonical partial order} on $R$. Given $r,s \in R$, we shall write $r \leq s$ when $s+r = s$. One may verify that this is indeed a partial order (the fact that $R$ is idempotent is required here) which is compatible with both addition and multiplication.

\begin{example}
	The most simple example of an idempotent hemiring (which is crucial here) is the set $2^{A^*}$ of all languages over $A$. Indeed, it suffices to use union as the addition (the neutral element is $\emptyset$) and language concatenation as the multiplication. Observe that the canonical partial order is inclusion ($L \subseteq H$ if and only if $H \cup L = H$). In fact, $2^{A^*}$ has even more structure: it is a semiring ($\{\varepsilon\}$ is neutral for multiplication).
\end{example}

We may now define \mratms. We call a hemiring morphism $\rho: 2^{A^*} \to R$ into a \emph{finite} idempotent hemiring $R$ a \emph{\mratm}. Specifically, $\rho$ has to satisfy the following axioms: 
\begin{enumerate}
	\item\label{itm:bgen:fzer} $\rho(\emptyset) = 0_R$.
	\item\label{itm:bgen:fadd} For any $K_1,K_2 \subseteq A^*$, $\rho(K_1 \cup K_2) = \rho(K_1) + \rho(K_2)$.	
	\item\label{itm:bgen:fmul} For any $K_1,K_2 \subseteq A^*$, $\rho(K_1K_2) = \rho(K_1) \cdot \rho(K_2)$.
\end{enumerate}

\begin{remark}
	In~\cite{pzcovering2}, the definition of \mratms is slightly more restrictive: $R$ must be a finite idempotent semiring and $\rho$ must be a semiring morphism. This change is harmless and could actually be avoided: working with semirings suffices to make the connection with covering as seen in~\cite{pzcovering2}. We make it for convenience: the proof of our \pbpol{\Cs}-covering algorithm (in Section~\ref{sec:pbpol}) involves auxiliary \mratms and allowing hemirings simplifies their presentation (see Remark~\ref{rem:standmult}).
\end{remark}

For the sake of improved readability, when applying a \mratm $\rho$ to a singleton language $K = \{w\}$, we shall write $\rho(w)$ for $\rho(\{w\})$. Note that $\rho$ is increasing for the canonical orders on $2^{A^*}$ and $R$ (this is true for any morphism of idempotent hemirings).

\begin{fct} \label{fct:fadd}
	Consider a \mratm $\rho: 2^{A^*} \to R$. Then, for any two languages $K_1,K_2 \subseteq A^*$, the following property holds,	$K_1 \subseteq K_2 \Rightarrow \rho(K_1) \leq \rho(K_2)$.
\end{fct}

\medskip
\noindent
{\bf \Cs-compatible \mratms.} We lift the notion of \Cs-compatibility to \mratms. Consider a \mratm $\rho: 2^{A^*} \to R$ and a finite \vari \Cs. We say that $\rho$ is \Cs-compatible when for any $r \in R$, there exists a $\sim_{\Cs}$-class \ctype{r} such that for any $w \in A^*$ satisfying $\rho(w) = r$, we have $w \in \ctype{r}$. When we manipulate \Cs-compatible \mratms, we shall implicitly assume that the map $r \mapsto \ctype{r}$ is fixed.

\begin{remark}
	This definition is distinct from the one of~\cite{pzcovering2}. We make this choice because our definition is simpler and suffices for presenting our results. However, this is harmless: one may show that the two definitions coincide for \nice \mratms (defined below) and this is the only situation in which we use \Cs-compatibility here.
\end{remark}

\medskip
\noindent
{\bf \Nice \mratms.} We now define a special class of \mratms which is crucial: it is used for the connection with covering. Consider a \mratm $\rho: 2^{A^*} \to R$. We say that $\rho$ is {\bf \nice} if for any language $K \subseteq A^*$, there exist words $w_1,\dots,w_n \in K$ such that,
\[
\rho(K) =  \rho(w_1) + \cdots + \rho(w_n)
\]
\begin{remark}
	There exist \mratms which are not \nice. For example, let $R = \{0,1,2\}$ be the semiring equipped with the following addition and multiplication. For $i,j \in R$, we let $i + j = max(i,j)$, $0 \cdot i = i \cdot 0 = 0$, $1 \cdot i = i \cdot 1 = i$ and $2 \cdot 2 = 2$. Moreover, consider the \mratm $\rho: 2^{A^*} \to R$ defined $\rho(\emptyset) = 0$, $\rho(K) = 1$ when $K$ is non-empty and finite, and $\rho(K) = 2$ when $K$ is infinite. Clearly, $\rho(A^*) = 2$. However, given finitely many words $w_1,\dots,w_n \in A^*$, we have $\rho(w_1) + \cdots + \rho(w_n) = 1 \neq 2$.	 
\end{remark}

A key point is that \nice \mratms are finitely representable (which is not the case in general). Let us explain why. Given any \mratm $\rho: 2^{A^*} \to R$ we associate a canonical morphism $\beta: A^* \to M$ (where $M$ is a monoid computed from $R$). We then argue than when $\rho$ is \nice, it is characterized by this canonical morphism $\beta$.

Consider an arbitrary \mratm $\rho: 2^{A^*} \to R$. We associate a canonical monoid morphism $\beta$ to $\rho$. Essentially, $\beta$ is the restriction of $\rho$ to $A^*$. However, since $R$ need not be a monoid for multiplication, we also have to restrict the output set $R$.

We let $R_{A^*} \subseteq R$ as the set $R_{A^*} = \{\rho(w) \mid w \in A^*\}$. Clearly, since $\rho$ is a \mratm, $R_{A^*}$ is a sub-semigroup of $R$ for multiplication.  Moreover, it is a monoid whose neutral element is $\rho(\varepsilon)$. We now define the canonical morphism $\beta$ associated to $\rho$ as the following monoid morphism $\beta: A^* \to R_{A^*}$:
\[\begin{aligned}[t]
  \beta: A^* &\to R_{A^*} \\
  w   &\mapsto \rho(w)
\end{aligned}\]
When $\rho: 2^{A^*} \to R$ is \nice, it is fully determined by the associated canonical morphism $\beta: A^* \to R_{A^*}$ and the finite hemiring $R$. Indeed, the definition of \nice \mratms exactly states that $\rho(K) = \sum_{w \in K} \rho(w)$. Hence, $\beta$ and the addition of $R$ determine $\rho(K)$ for any language $K$. Consequently, it makes sense to speak of algorithms which take a \nice \mratm as input.

\begin{remark}
	Observe that when we have a \nice \mratm $\rho$ in hand, it is possible to evaluate $\rho(K)$ when $K$ is a regular language. Indeed, $\rho(K)$ is the sum of all elements $\rho(w)$ for $w \in K$ which is simple to evaluate (see Lemma~5.7 in~\cite{pzcovering2} for details).
\end{remark}

Finally, we have the following fact which is immediate from the definitions. 

\begin{fct} \label{fct:canoiscompat}
	Consider a \Cs-compatible \mratm $\rho: 2^{A^*} \to R$. Then the associated canonical morphism $\beta: A^* \to R_{A^*}$ is \Cs-compatible.
\end{fct}

\subsection{Connection with finite sets of regular languages.}  We now explain how one may associate a (computable) \mratm to a finite set of regular languages. This will later be useful to reformulate the covering problem with our framework.

Consider a finite set of languages \Lb and a \mratm $\rho: 2^{A^*} \to R$. Observe that $2^\Lb$ is an idempotent commutative monoid for union. We say that \emph{$\rho$ extends $\Lb$} when there exists a morphism $\delta: R \to 2^\Lb$ for the addition of $R$ (i.e. $\delta(0_R) = \emptyset$ and $\delta(r_1+r_2) = \delta(r_1) \cup \delta(r_2)$ for any $r_1,r_2 \in R$) such that for every language $K$,
\[
\delta(\rho(K)) = \{L \in \Lb \mid K \cap L \neq \emptyset\} \in 2^\Lb
\]
We call $\delta$ an \emph{extending morphism}. Let us point out that $\delta$ is finitely representable since both \Lb and $R$ are finite. We connect this definition to covering with the following fact.

\begin{fct} \label{fct:extend}
	Consider a finite set of languages \Lb together with a \mratm $\rho: 2^{A^*} \to R$ which extends \Lb for the extending morphism $\delta: R \to 2^\Lb$.  Then, a finite set of languages \Kb is separating for \Lb if and only if $\delta(\rho(K)) \neq \Lb$ for every $K \in \Kb$.
\end{fct}

\begin{proof}
	Immediate from the definitions. Recall that \Kb is separating for \Lb if for every $K \in \Kb$, there exits $L \in \Lb$ such that $K \cap L = \emptyset$.
\end{proof}

We now show that for any finite set of regular languages \Lb, one may compute an extending \nice \Cs-compatible \mratm.

\begin{proposition} \label{prop:extend}
	Let \Cs be a finite \vari. Given as input a finite set of regular languages \Lb, one may compute a \nice \Cs-compatible \mratm $\rho: 2^{A^*} \to R$ which extends \Lb and the corresponding extending morphism $\delta: R \to 2^\Lb$.
\end{proposition}

\begin{proof}
	Let $\Lb = \{L_1,\dots,L_n\}$. By Lemma~\ref{lem:compat}, for any $i \leq n$, one may compute a \Cs-compatible morphism $\alpha_i: A^* \to M_i$ which recognizes $L_i$. We let $F_i$ as the corresponding accepting set: $L_i = \alpha_i\inv(F_i)$. One may verify that for any $i \leq n$, the set $2^{M_i}$ is a finite idempotent hemiring. The addition is union and the multiplication is obtained by lifting the one of $M_i$: for $S,T \in 2^{M_i}$ their multiplication is $S \cdot T = \{st \mid s \in S \text{ and } t \in T\}$. We define,
	\[
	R = 2^{M_1} \times \cdots \times 2^{M_n}
	\]
	It is straightforward to verify that $R$ is a finite idempotent hemiring for the componentwise addition and multiplication. We define our \mratm $\rho: 2^{A^*} \to R$ as follows:
	\[
	\begin{array}{llll}
	\rho: & 2^{A^*} & \to     & R                               \\
	& K       & \mapsto & (\alpha_1(K),\dots,\alpha_n(K))
	\end{array}
	\]
	One may verify that $\rho$ is indeed a \nice \Cs-compatible \mratm (\Cs-compatibility comes from the fact that we started from \Cs-compatible morphisms). 
	
	It remains to explain why $\rho$ extends \Lb and how to compute an extending morphism $\delta: R \to 2^\Lb$. We define $\delta: R \to 2^\Lb$ as follows. Consider $r= (S_1,\dots,S_n) \in R$, for $i \leq n$, we have $L_i \in \delta(r)$ if and only if $S_i \cap F_i \neq \emptyset$. Clearly, one may compute $\delta$ from the morphism $\alpha_i$. Let us verify that it is an extending morphism. Let $K \subseteq A^*$, we show that $\delta(\rho(K)) = \{L \in \Lb \mid K \cap L \neq \emptyset\}$. Given $L_i \in \Lb$, we have to show that $L_i \in \delta(\rho(K))$ if and only if $K \cap L_i \neq \emptyset$. By definition, $L_i \in \delta(\rho(K))$ if and only if $\alpha_i(K) \cap F_i \neq \emptyset$. Since $L_i = \alpha_i\inv(F_i)$, this is equivalent to $K \cap L_i \neq \emptyset$, finishing the proof.
\end{proof}

\subsection{\Imprints and optimal covers}

We may now explain how we use \mratms to measure the quality of covers. This is based on a new notion called ``\imprints''. Consider a \mratm $\rho: 2^{A^*} \to R$ (possibly not \nice). For any finite set of languages \Kb, the $\rho$-\imprint of \Kb (denoted by $\prin{\rho}{\Kb} \subseteq R$) is the following set:
\[
\prin{\rho}{\Kb} = \{r \in R \mid \text{there exists $K \in \Kb$ such that $r \leq \rho(K)$}\} \subseteq R
\]
When using this notion, we will always have some language $L \subseteq A^*$ in hand and our objective will be to find the ``best possible'' cover \Kb of $L$. Intuitively, $\rho$-\imprints are designed for this purpose: given a candidate cover \Kb of $L$, we use the $\rho$-\imprint of \Kb to measure its ``quality''.

\medskip

This leads to the notion of optimality. Assume that some arbitrary lattice \Ds is fixed and consider a language $L$. An \emph{optimal \Ds-cover of $L$ for $\rho$} is a \Ds-cover of $L$ which has the smallest possible $\rho$-\imprint (with respect to inclusion). That is, \Kb is an optimal \Ds-cover of $L$ for $\rho$ if and only if,
\[
\prin{\rho}{\Kb} \subseteq \prin{\rho}{\Kb'} \quad \text{for every \Ds-cover $\Kb'$ of $L$}
\]
In general, there can be infinitely many optimal \Ds-covers of $L$ for $\rho$. However, there always exists at least one (we need the fact that \Ds is a lattice for this, see Lemma~4.15 in~\cite{pzcovering2} for the proof).

\begin{lemma} \label{lem:bgen:opt}
	Let \Ds be a lattice. Then, for any \mratm $\rho: 2^{A^*} \to R$ and any language $L \subseteq A^*$, there exists an optimal \Ds-cover of $L$ for $\rho$.
\end{lemma}

The proof of Lemma~\ref{lem:bgen:opt} is non-constructive. In fact, given $L$ and $\rho: 2^{A^*} \to R$, computing an actual optimal \Ds-cover of $L$ for $\rho$ is a difficult problem in general. As seen in Theorem~\ref{thm:bgen:main} below, getting such algorithm yields a procedure for \Ds-covering. Before we can present this theorem, we need a key observation about optimal \Ds-covers.

\medskip
\noindent{\bf Optimal \imprints.} By definition, all optimal \Ds-covers of $L$ for $\rho$ have the same $\rho$-\imprint. Hence, this unique $\rho$-\imprint is a \emph{canonical} object for \Ds, $L$ and $\rho$. We call it the \emph{\Ds-optimal $\rho$-\imprint on $L$} and we denote it by $\opti{\Ds}{L,\rho}$:
\[
\opti{\Ds}{L,\rho} = \prin{\rho}{\Kb} \subseteq R \quad \text{for every optimal \Ds-cover \Kb of $L$ for $\rho$}.
\]

We complete this definition with a few properties of optimal \imprints. We start with two facts which will be useful (the proofs are available in~\cite{pzcovering2}, see Facts~4.16 and~4.17).

\begin{fct} \label{fct:inclus1}
	Consider a \mratm $\rho: 2^{A^*} \to R$ and a language $L$. Let $\Cs,\Ds$ be two lattices such that $\Cs \subseteq \Ds$. Then, $\opti{\Ds}{L,\rho} \subseteq \opti{\Cs}{L,\rho}$.
\end{fct}

\begin{fct} \label{fct:inclus2}
	Consider a \mratm $\rho: 2^{A^*} \to R$ and \Ds a lattice. Let $H,L$ be two languages such that $H \subseteq L$. Then, $\opti{\Ds}{H,\rho} \subseteq \opti{\Ds}{L,\rho}$.
\end{fct}

Moreover, we have the following important lemma which connects optimal \imprints to multiplication. Again, this lemma is proved in~\cite{pzcovering2}, (see Lemma~5.8)

\begin{lemma} \label{lem:closmult}
	Consider a \pvari \Ds. Let $L_1,L_2 \subseteq A^*$ be two languages and $\rho: 2^{A^*} \to R$ a \mratm. Then, for any $r_1 \in \opti{\Ds}{L_1,\rho}$ and $r_2 \in \opti{\Ds}{L_2,\rho}$, we have $r_1r_2 \in \opti{\Ds}{L_1L_2,\rho}$.	
\end{lemma}

Finally, we present a useful lemma which gives an alternate definition of \Ds-optimal \imprints when \Ds is a finite lattice. We shall use it in proofs together with stratifications. Given a preorder relation $\leqslant$ defined on $A^*$, a word $w \in A^*$ and a language $K \subseteq A^*$, we write $w \leqslant K$ to denote the fact that $w \leqslant u$ {\bf for all} $u \in K$.

\begin{lemma} \label{lem:finitedef}
	Assume that \Ds is a {\bf finite} lattice. Let $\rho:2^{A^*} \to R$ be a \mratm and $L$ a language. Given $r \in R$, the following are equivalent:
		\begin{enumerate}
			\item $r \in \opti{\Ds}{L,\rho}$.
			\item There exist $w \in L$ and $K \subseteq A^*$ such that $w \canod K$ and $r \leq \rho(K)$.
		\end{enumerate}	
\end{lemma}

\begin{proof}
	We start with preliminary terminology. For every $w \in A^*$, we let $K_w \subseteq A^*$ as the least upper set for \canod which contains $w$: $K_w = \{u \in A^* \mid w \canod v\}$. Recall that by Lemma~\ref{lem:canosatur}, the upper sets of \canod are exactly the languages in \Ds. In particular, there are finitely many languages $K_w$ (even though there are infinitely many words $w \in A^*$). It follows that the set $\Kb = \{K_u \mid u \in L\}$ is a \Ds-cover of $L$. We may now prove the lemma.
	
	Assume first that $r \in \opti{\Ds}{L,\rho}$. Since \Kb is a \Ds-cover of $L$, this implies that $r \in \prin{\rho}{\Kb}$ by definition of \Ds-optimal \imprints. Hence, we have $K \in \Kb$ such that $r \leq \rho(K)$. By definition, $K = K_w$ for some $w \in L$ which means that $w \canod K$ and the second assertion in the lemma holds. Conversely, assume that there exist $w \in L$ and $K \subseteq A^*$ such that $w \canod K$ and $r \leq \rho(K)$. Consider an arbitrary \Ds-optimal cover \Hb of $L$. By definition, $\prin{\rho}{\Hb} = \opti{\Ds}{L,\rho}$. Therefore, it suffices to show that $r \in \prin{\rho}{\Hb}$. Since \Hb is a cover of $L$ and $w \in L$, we have $H \in \Hb$ such that $w \in H$. Moreover, since $H \in \Ds$, it is a upper set for \canod by Lemma~\ref{lem:canosatur}. Consequently, $w \canod K$ and $w \in H$ imply that $K \subseteq H$ and we get that $r \leq \rho(K) \leq \rho(H)$. We conclude that $r \in \prin{\rho}{\Hb}$, finishing the proof.
\end{proof}

%
%

\subsection{Connection with the covering problem} We may now connect these definitions to the covering problem. We do so with the following theorem.

\begin{theorem} \label{thm:bgen:main}
	Consider a lattice \Ds. Let $L \subseteq A^*$ be a language and \Lb a finite set of languages. Moreover, let $\rho: 2^{A^*} \to R$ be a \mratm extending \Lb for the extending morphism $\delta: R \to 2^\Lb$. The following properties are equivalent:
	\begin{enumerate}
		\item $(L,\Lb)$ is \Ds-coverable.
		\item We have $\Lb \not\in \delta(\opti{\Ds}{L,\rho})$.
		\item Any optimal \Ds-cover of $L$ for $\rho$ is separating for \Lb.
	\end{enumerate} 
\end{theorem}

\begin{proof}
	We show that $(1) \Rightarrow (2) \Rightarrow (3) \Rightarrow (1)$. Let us start with $(1) \Rightarrow (2)$. Assume that $(L,\Lb)$ is \Ds-coverable and let \Kb be a \Ds-cover of $L$ which is separating for \Lb. We show that $\Lb \not\in \delta(\opti{\Ds}{L,\rho})$. Given $r \in \opti{\Ds}{L,\rho}$, we have to prove that $\Lb \neq \delta(r)$. Since \Kb is a \Ds-cover of $L$, we have $\opti{\Ds}{L,\rho} \subseteq \prin{\rho}{\Kb}$. Hence, $r \in \prin{\rho}{\Kb}$ which yields $K \in \Kb$ such that $r \leq \rho(K)$, i.e. $\rho(K) + r = \rho(K)$. Since $\delta$ is a morphism this implies that $\delta(\rho(K)) \cup \delta(r) = \delta(\rho(K))$, i.e. $\delta(r) \subseteq \delta(\rho(K))$. Moreover, since \Kb is separating for \Lb, it follows from Fact~\ref{fct:extend} that $\delta(\rho(K)) \neq \Lb$. Consequently, $\delta(r) \subseteq \delta(\rho(K))$ yields $\Lb \neq \delta(r)$, finishing the proof.
	
	We turn to the direction $(2) \Rightarrow (3)$. Assume that $\Lb \not\in \delta(\opti{\Ds}{L,\rho})$. Consider an optimal \Ds-cover \Kb of $L$ for $\rho$. We show that \Kb is separating for \Lb. By definition of \Kb, for any $K \in \Kb$, we have $\rho(K) \in \opti{\Ds}{L,\rho}$ which means that $\delta(\rho(K)) \neq \Lb$ by hypothesis. Thus, it follows from Fact~\ref{fct:extend} that \Kb is separating for \Lb.
		
	We finish with the direction $(3) \Rightarrow (1)$. Since we know from Lemma~\ref{lem:bgen:opt} that there always exists an optimal \Ds-cover of $L$ for $\rho$, Item~(3) yields that there exists a \Ds-cover of $L$ which is separating for \Lb. Hence, $(L,\Lb)$ is \Ds-coverable.
\end{proof}

Theorem~\ref{thm:bgen:main} formally connects our framework to the covering problem. Let \Ds be some lattice. Given as input a regular language $L$ and a finite set of regular languages \Lb, Proposition~\ref{prop:extend} states that we may compute a \nice \mratm $\rho$ extending \Lb. By Theorem~\ref{thm:bgen:main}, deciding whether $(L,\Lb)$ is \Ds-coverable now amounts to computing \opti{\Ds}{L,\rho} (we outline the reduction precisely below). Additionally, getting a separating \Ds-cover (if there exists one) reduces to computing an optimal \Ds-cover of $L$ for $\rho$.

This is essentially what our algorithms do (for $\Ds =\pol{\Cs}$ or $\Ds =\pbpol{\Cs}$). However, rather than directly computing \opti{\Ds}{L,\rho}, we will work with some morphism $\alpha: A^* \to M$ recognizing $L$ and compute all sets \opti{\Ds}{\alpha\inv(s),\rho} for $s \in M$ simultaneously (this is how we exploit the hypothesis that $L$ is regular). As shown in the following proposition, this is sufficient information to compute \opti{\Ds}{L,\rho}.

\begin{proposition} \label{prop:isreg}
	Let \Ds be a lattice. Let $L$ be a regular language recognized by a morphism $\alpha: A^* \to M$ for the accepting set $F$ (i.e. $L = \alpha\inv(F)$). Moreover, let $\rho: 2^{A^*} \to R$ be a \mratm. For any $s \in M$, let $\Kb_s$ be an optimal \Ds-cover of $\alpha\inv(s)$ for $\rho$. The two following properties hold:
	\begin{itemize}
		\item We have $\opti{\Ds}{L,\rho} = \bigcup_{s \in F} \opti{\Ds}{\alpha\inv(s),\rho}$.
		\item The set $\Kb = \bigcup_{s \in F} \Kb_s$ is an optimal \Ds-cover of $L$.
	\end{itemize}	
\end{proposition}

\begin{proof}
	We start with the second item. Since $L = \alpha\inv(F)$ and $\Kb_s$ is a \Ds-cover of $\alpha\inv(s)$ for any $s \in M$, we know that $\Kb = \bigcup_{s \in F} \Kb_s$ is a \Ds-cover of $L$. Let us show that it is optimal. Consider another \Ds-cover $\Kb'$ of $L$, we show that $\prin{\rho}{\Kb} \subseteq \prin{\rho}{\Kb'}$. Let $r \in \prin{\rho}{\Kb}$. By definition, $r \leq \rho(K)$ for some $K \in \Kb$. Moreover, $K \in \Kb_s$ for some $s \in F$. Since $\Kb_s$ is an optimal \Ds-cover of $\alpha\inv(s)$ for $\rho$, it follows that $r \in \opti{\Ds}{\alpha\inv(s),\rho}$. Finally, since $\Kb'$ is a \Ds-cover of $L$, it is also a \Ds-cover of $\alpha\inv(s) \subseteq L$ ($L = \alpha\inv(F)$ and $s \in F$). Hence, $\opti{\Ds}{\alpha\inv(s),\rho} \subseteq \prin{\rho}{\Kb'}$ which yields $r \in \prin{\rho}{\Kb'}$, finishing the proof.
	
	We finish with the first item. We just proved that $\Kb = \bigcup_{s \in F} \Kb_s$ is an optimal \Ds-cover of $L$. Thus, $\opti{\Ds}{L,\rho} = \prin{\rho}{\Kb}$. Moreover, it is immediate from the definition that $\prin{\rho}{\Kb}  = \bigcup_{s \in F} \prin{\rho}{\Kb_s}$. Hence, since $\Kb_s$ is an optimal \Ds-cover of $\alpha\inv(s)$ for $\rho$ which means that $\prin{\rho}{\Kb_s} = \opti{\Ds}{\alpha\inv(s),\rho}$, we obtain $\opti{\Ds}{L,\rho} = \bigcup_{s \in F} \opti{\Ds}{\alpha\inv(s),\rho}$.	
\end{proof}

Altogether, this means that we shall be looking for algorithms which take a morphism $\alpha: A^* \to M$ and a \nice \mratm $\rho: 2^{A^*} \to R$ as input and compute all \Ds-optimal $\rho$-\imprints $\opti{\Ds}{\alpha\inv(s),\rho}$ for $s \in M$. It will be convenient to have a single notation which records all these objects.

Given a lattice \Ds, a morphism $\alpha: A^* \to M$ and a \mratm $\rho:2^{A^*} \to R$, we write \popti{\Ds}{\alpha,\rho} for the following set,
\[
\popti{\Ds}{\alpha,\rho} = \{(s,r) \in M \times R \mid r \in \opti{\Ds}{\alpha\inv(s),\rho}\} \subseteq M \times R
\]
We call \popti{\Ds}{\alpha,\rho} the \emph{\Ds-optimal $\alpha$-pointed $\rho$-\imprint}. Clearly, it encodes all sets $\opti{\Ds}{\alpha\inv(s),\rho}$ for $s \in M$.  We may now summarize the connection between covering and our framework with the following proposition.

\begin{proposition} \label{prop:thereduction} 	
 	Consider a lattice \Ds and some finite \vari \Cs. Assume that there exists an algorithm for the following computational problem:
 	\begin{center}
 	\begin{tabular}{ll}
 		{\bf Input:}    & A \Cs-compatible morphism $\alpha: A^* \to M$ and, \\
 		                & a \nice \Cs-compatible \mratm $\rho:2^{A^*} \to R$.\\
 		{\bf Output:} & Compute the \Ds-optimal $\alpha$-pointed $\rho$-\imprint, \popti{\Ds}{\alpha,\rho}.
 	\end{tabular}
 	\end{center}
	Then, \Ds-covering is decidable.
\end{proposition}

We prove Proposition~\ref{prop:thereduction} by describing a procedure for \Ds-covering. Our input is a pair $(L,\Lb)$ where all languages in $\{L\} \cup \Lb$ are regular: we need to decide whether $(L,\Lb)$ is \Ds-coverable. The reduction to the problem described in Proposition~\ref{prop:thereduction} is based on the following steps:
	\begin{enumerate}
		\item Compute a \Cs-compatible morphism $\alpha: A^* \to M$ recognizing $L$. This is possible by Lemma~\ref{lem:compat}. We let $F \subseteq M$ as the accepting set: $L = \alpha\inv(F)$.
		\item Compute a \nice \Cs-compatible \mratm $\rho: 2^{A^*} \to R$ which extends \Lb and an extending morphism $\delta: R \to 2^{\Lb}$. This is possible by Proposition~\ref{prop:extend}.
		\item Use the procedure given by Proposition~\ref{prop:thereduction} to compute the \Ds-optimal $\alpha$-pointed $\rho$-\imprint, \popti{\Ds}{\alpha,\rho}.
		\item Compute \opti{\Ds}{L,\rho} from \popti{\Ds}{\alpha,\rho}. Recall that by Proposition~\ref{prop:isreg}, we have,
		\[
		\opti{\Ds}{L,\rho} = \bigcup_{s \in F} \opti{\Ds}{\alpha\inv(s),\rho}
		\]
		Moreover, for any $s \in M$, $\opti{\Ds}{\alpha\inv(s),\rho} = \{r \in R \mid (s,r) \in \popti{\Ds}{\alpha,\rho}\}$ by definition.
		\item By Theorem~\ref{thm:bgen:main}, $(L,\Lb)$ is \Ds-coverable if and only if we have $\Lb \not\in \delta(\opti{\Ds}{L,\rho})$. We may now check this condition.
	\end{enumerate}
This concludes the proof of Proposition~\ref{prop:thereduction}.

\medskip

In view of Proposition~\ref{prop:thereduction}, when we tackle \pol{\Cs}- and \pbpol{\Cs}-covering in the next two sections, we shall do so by presenting algorithms which compute \pocopti and \pbpopti from a \Cs-compatible morphism $\alpha: A^*: \to M$ and a \nice \mratm $\rho: 2^{A^*} \to R$ (which also needs to be \Cs-compatible in the case of \pbpol{\Cs}).

A key point is that we present these two algorithms as elegant characterization theorems. This is a design principle of our framework. Given an arbitrary \pvari \Ds (such as \pol{\Cs} or \pbpol{\Cs}), the idea is to characterize $\popti{\Ds}{\alpha,\rho} \subseteq M \times R$ as the least subset of $M \times R$ which contains basic elements and is closed under operations within a finite list (which depends on \Ds). In turn, this yields a least fixpoint algorithm for computing \popti{\Ds}{\alpha,\rho}: one starts from the set of basic elements and saturates it with the operations. 

\begin{remark}
	There lies the difference between our results for \pol{\Cs} and \pbpol{\Cs}. Our characterization of \pocopti holds for {\bf any} \mratm $\rho$  (and in particular for those which are not \nice).  On the other hand, our characterization of \pbpopti only holds for \nice \Cs-compatible \mratms.	In other words, our results provide more information about \pol{\Cs}. Of course, this is useless for \pol{\Cs}-covering: considering \nice \Cs-compatible \mratms suffices by Proposition~\ref{prop:thereduction}. However, it turns out that having a characterization of \pocopti which holds for {\bf all} \mratms (even if it is non-effective in general) is crucial for handling \pbpol{\Cs}-covering.
\end{remark}

\subsection{Generic properties of optimal pointed \imprints}

We finish with an important lemma which presents properties of optimal pointed \imprints (i.e. the objects that we now want to compute). These properties are \emph{generic}: they are satisfied by \popti{\Ds}{\alpha,\rho} for any \pvari \Ds. In particular, they are involved in our characterizations of both \pocopti and \pbpopti.

Given any \mratm $\rho: 2^{A^*} \to R$ and any morphism $\alpha: A^* \to M$, we shall write \ptriv{\alpha,\rho} for the following set,
\[
\ptriv{\alpha,\rho} = \{(\alpha(w),r) \in M \times R \mid w \in A^* \text{ and } r \leq \rho(w)\}
\]
It is straightforward to verify that given a morphism $\alpha: A^* \to M$ and a {\bf \nice} \mratm $\rho$ as input, one may compute \ptriv{\alpha,\rho}. Our key lemma is as follows.

\begin{lemma} \label{lem:satur}
	Let \Ds be a \pvari, $\alpha: A^* \to M$ a morphism and $\rho: 2^{A^*} \to R$ a \mratm. Then, $\popti{\Ds}{\alpha,\rho} \subseteq M \times R$ contains the set \ptriv{\alpha,\rho} and satisfies the two following closure properties:
	\begin{enumerate}
		\item {\bf Downset.} For any $(s,r) \in \popti{\Ds}{\alpha,\rho}$ and $r ' \leq r$, $(s,r') \in \popti{\Ds}{\alpha,\rho}$.
		\item {\bf Multiplication.} For any $(s_1,r_1),(s_2,r_2) \in \popti{\Ds}{\alpha,\rho}$, $(s_1s_2, r_1r_2) \in \popti{\Ds}{\alpha,\rho}$.
	\end{enumerate}	
\end{lemma}

\begin{proof}
	Let us first show that \popti{\Ds}{\alpha,\rho} contains \ptriv{\alpha,\rho}. Let $(s,r) \in \ptriv{\alpha,\rho}$. By definition of \popti{\Ds}{\alpha,\rho}, we have to show that $r \in \opti{\Ds}{\alpha\inv(s),\rho}$. Let \Kb be an optimal \Ds-cover of $\alpha\inv(s)$ for $\rho$,   we have to show that $r \in \prin{\rho}{\Kb}$. By definition of \ptriv{\alpha,\rho}, there exists $w \in A^*$ such that $s = \alpha(w)$ and $r \leq \rho(w)$. Consequently $w \in \alpha\inv(s)$ and there exists $K \in \Kb$ such that $w \in K$. Thus, $r \leq \rho(w) \leq \rho(K)$ which means that $r \in \prin{\rho}{\Kb}$ by definition.
	
	It remains to prove that \popti{\Ds}{\alpha,\rho} is closed under downset and multiplication. We start with downset. Consider $(s,r) \in \popti{\Ds}{\alpha,\rho}$ and $r ' \leq r$. By definition, we have $r \in \opti{\Ds}{\alpha\inv(s),\rho}$. Since \opti{\Ds}{\alpha\inv(s),\rho} is a $\rho$-\imprint, it is immediate by definition that $r' \leq r$ implies that $r'  \in  \opti{\Ds}{\alpha\inv(s),\rho}$. Therefore, $(s,r') \in \popti{\Ds}{\alpha,\rho}$.
	
	Closure under multiplication follows from Lemma~\ref{lem:closmult}. Let $(s_1,r_1),(s_2,r_2) \in \popti{\Ds}{\alpha,\rho}$. It follows that $r_1 \in \opti{\Ds}{\alpha\inv(s_1),\rho}$ and $r_2 \in \opti{\Ds}{\alpha\inv(s_2),\rho}$. Therefore, since \Ds is a \pvari, Lemma~\ref{lem:closmult} yields $r_1r_2 \in \opti{\Ds}{\alpha\inv(s_1)\alpha\inv(s_2),\rho}$. Moreover, since $\alpha$ is a morphism, it is immediate that $\alpha\inv(s_1)\alpha\inv(s_2) \subseteq \alpha\inv(s_1s_2)$. Hence, Fact~\ref{fct:inclus2} yields $r_1r_2 \in \opti{\Ds}{\alpha\inv(s_1s_2),\rho}$ which means that 
	$(s_1s_2, r_1r_2) \in \popti{\Ds}{\alpha,\rho}$ by definition.
\end{proof}


\section{Covering for \pol{\Cs}}
\label{sec:polc}
In this section, we present the first of our two covering algorithms. Given an arbitrary finite \vari \Cs (which we fix for the section), we show that \pol{\Cs}-covering is decidable. As announced, we use the framework introduced in Section~\ref{sec:covers}: we present an effective characterization of \pol{\Cs}-optimal pointed \imprints. More precisely, our main theorem is a description of the set \pocopti for any \Cs-compatible morphism $\alpha: A^* \to M$ and any \mratm $\rho$ (\nice or not). Furthermore, this description yields an algorithm for computing \pocopti when $\rho$ is \nice.

\begin{remark} \label{rem:half:whatratms}
	Naturally, the algorithm only applies to \nice \mratms. It does not even make sense to speak of an algorithm which takes arbitrary \mratms as input: we cannot represent them finitely in general. However, the characterization itself applies to {\bf any} \mratm. While this is useless for \pol{\Cs}-covering, having this result will be crucial in Section~\ref{sec:pbpol} when we tackle \pbpol{\Cs}-covering. 
\end{remark}

We start by presenting an auxiliary result about factorization forests that we shall need for proving our characterization (we choose to isolate it from the main argument as we shall reuse it later when considering \pbpol{\Cs}). We then present our characterization and devote the remainder of the section to its proof.

\subsection{Covering languages of factorization forests}

When proving the difficult direction in our characterization of \pol{\Cs}-optimal pointed \imprints, we shall work with two objects: a morphism $\alpha: A^* \rightarrow M$ and a \mratm $\rho: 2^{A^*} \to R$. Our objective will be to build \pol{\Cs}-covers of all languages $\alpha\inv(s)$ for $s \in M$ which satisfy specific properties with respect to $\rho$. We shall use the factorization forest theorem: any word admits an $\alpha$-factorization forest of height at most $3|M|-1$ (see Section~\ref{sec:tools}, Theorem~\ref{thm:facto}). Specifically, we use an induction on the height of factorization forests. Our proof is based on a construction which takes as input some height $h$ and \pol{\Cs}-covers for the restrictions of the languages $\alpha\inv(s)$ to words admitting $\alpha$-factorization forests of height at most $h$. It combines them into larger \pol{\Cs}-covers which also contain the words admitting forests of height at most $h+1$.

The following lemma formalizes an important part of this construction which we shall later reuse for handling \pbpol{\Cs}. We assume that the morphism $\alpha: A^* \rightarrow M$ and the \mratm $\rho: 2^{A^*} \to R$ are fixed. Recall that given $s \in M$ and $h,m \in \nat$, $F^\alpha(s,h,m)$  (resp. $F^\alpha_I(s,h,m)$) denotes the language of all $w \in \alpha^{-1}(s)$ admitting an $\alpha$-factorization forest of height at most $h$ and idempotent height at most $m$ (resp. whose root is an idempotent node). 

\begin{lemma} \label{lem:inducidem}
	Let $h,m \geq 1$ and an idempotent $e \in M$. Let \Ub be a cover of $F^\alpha(e,h-1,m-1)$ and $V \subseteq A^*$ containing $\alpha^{-1}(e)$. There exists a cover \Kb of $F^\alpha_I(e,h,m)$ such that any $K \in \Kb$ is a concatenation $K = K_1 \cdots K_n$ with each $K_i$ having one of the two following forms:
	\begin{enumerate}
		\item $K_i$ is a language in \Ub, or,
		\item $K_i = U_1 \cdots U_\ell V U'_1 \cdots U'_{\ell'}$ where $U_1,\cdots,U_\ell, U'_1,\dots,U'_{\ell'} \in \Ub$ and there exists some {\bf idempotent} $f \in R$ for multiplication such that $\rho(U_1 \cdots U_\ell) = \rho(U'_1 \cdots U'_{\ell'}) = f$.
	\end{enumerate}
\end{lemma}

We now prove Lemma~\ref{lem:inducidem}. Consider $h,m \geq 1$ and an idempotent $e \in M$. Finally, we let \Ub as a cover of $F^\alpha(e,h-1,m-1)$ and $V \subseteq A^*$ such that $\alpha^{-1}(e) \subseteq V$. We construct a cover \Kb of $F^\alpha_I(e,h,m)$ satisfying the properties described in the lemma.

For the proof, we fix $k = 2^{3|R|}$. Given any $n \geq 1$, we define $\Kb_n$ as the set of all languages $K \subseteq A^*$ which are of the form $K = K_1 \cdots K_{p}$ with $p \leq n$ where each $K_i$ is of one of the two following kinds:
\begin{enumerate}
	\item $K_i \in \Ub$, or,
	\item $K_i = U_1 \cdots U_\ell V U'_1 \cdots U'_{\ell'}$ where $\ell,\ell' \leq k$, $U_1,\cdots,U_\ell, U'_1,\dots,U'_{\ell'} \in \Ub$ and there exists some {\bf idempotent} $f \in R$ such that $\rho(U_1 \cdots U_\ell) = \rho(U'_1 \cdots U'_{\ell'}) = f$.
\end{enumerate}

Observe that by definition all sets $\Kb_n$ are finite (each $K \in \Kb_n$ is a concatenation of at most $n \times (2k+1)$ languages in the finite set $\Ub \cup \{V\}$). Moreover all languages in $\Kb_n$ are of the form described in Lemma~\ref{lem:inducidem}. Thus, it suffices to prove that there exists $n \geq 1$ such that $\Kb_n$ is a cover of $F^\alpha_I(e,h,m)$. This is what we do.

We start with some terminology. Observe that by definition of $\alpha$-factorization forests, for any word $w \in F^\alpha_I(e,h,m)$, there exists a sequence of words $w_1,\dots,w_p \in F^\alpha(e,h-1,m-1)$ such that $w = w_1 \cdots w_p$. Given any such sequence $w_1,\dots,w_p \in F^\alpha(e,h-1,m-1)$, we associate a number that we call its \emph{index} (we shall use this number as an induction parameter). For the definition, we fix an arbitrary linear order on \Ub, the cover of $F^\alpha(e,h-1,m-1)$. Moreover, for any $w\in F^\alpha(e,h-1,m-1)$, we write $U[w]$ for the smallest language in \Ub containing $w$.

Consider a sequence $w_1,\dots,w_p \in F^\alpha(e,h-1,m-1)$ and an idempotent $f \in R$. We say that $f$ \emph{occurs in the sequence $w_1,\dots,w_\ell$} when there exists $i \leq p$ and $q \leq k$ such that,
\[
\rho(U[w_i] \cdots U[w_{i+q-1}]) = f
\]
Finally, we define the index of the sequence $w_1,\dots,w_p$ as the number of \emph{distinct idempotents} $f \in R$ that occur in $w_1,\dots,w_p$. Clearly, the index of $w_1,\dots,w_p$ is bounded by $|R|$ (the key point is that this bound is independent from the length $p$ of the sequence). We shall need the following fact which is proved using the factorization forest theorem.

\begin{fct} \label{fct:idemishere}
	Consider a sequence $w_1,\dots,w_p \in F^\alpha(e,h-1,m-1)$ of length $p \geq k =  2^{3|R|}$. Then the index of $w_1,\dots,w_p$ is at least one: an idempotent $f \in R$ occurs in $w_1,\dots,w_p$.
\end{fct}

\begin{proof}
	For all $i \leq k$, we let $r_i = \rho(U[w_i]) \in R$. Consider the morphism $\beta: R^* \to R$ which is defined by $\beta(r) = r$ for all $r \in R$. It is immediate from Theorem~\ref{thm:facto} that the word $(r_1) \cdots (r_k) \in R^*$ admits a $\beta$-factorization forest of height at most $3|R| - 1$. Since $k = 2^{3|R|}$, this forest has to contain at least one idempotent node. Thus, we have $i \leq k$ and $q \leq k$ such that $r_i \cdots r_{i+q-1}$ is an idempotent $f$ of $R$. This concludes the proof.
\end{proof}

We may now come back to the proof of Lemma~\ref{lem:inducidem} and exhibit $n \geq 1$ such that $\Kb_n$ is a cover of $F^\alpha_I(e,h,m)$. We use the following lemma.

\begin{lemma} \label{lem:inducase3}
	Let $d \in \nat$ and consider a sequence $w_1,\dots,w_p \in F^\alpha(e,h-1,m-1)$ whose index is at most $d$. There exists a language $K \in \Kb_{(k+1)(d+1)}$ containing $w_1 \cdots w_p$.
\end{lemma}

Lemma~\ref{lem:inducidem} is an immediate consequence of Lemma~\ref{lem:inducase3}: we obtain that $\Kb_n$ is a cover of $F^\alpha_I(e,h,m)$ for $n = (k+1)(|R|+1)$. Indeed, by definition, for any $w \in F^\alpha_I(e,h,m)$ there exist $w_1,\dots,w_p \in F^\alpha(e,h-1,m-1)$ such that $w = w_1 \cdots w_p$. Since the index of $w_1,\dots,w_p$ is at most $|R|$, it follows from Lemma~\ref{lem:inducase3} that there exist $K \in \Kb_n$ which contains $w = w_1 \cdots w_p$.

\medskip

It remains to prove Lemma~\ref{lem:inducase3}.	Let $d \in \nat$ and consider a sequence $w_1,\dots,w_p \in F^\alpha(e,h-1,m-1)$ whose index is at most $d$. We construct $K \in \Kb_{(k+1)(d+1)}$ which contains $w_1 \cdots w_p$. The argument is an induction on $d$.

Assume first that $d = 0$. In that case, there exists no idempotent $f \in R$ occurring in $w_1,\dots,w_p$. It is immediate from Fact~\ref{fct:idemishere} that $p < k$. Thus, we may simply choose,
\[
K = U[w_1] \cdots U[w_p] 
\]
Clearly $K \in \Kb_{(k+1)(d+1)}$ since $p < k$ and we have $w_1 \cdots w_p \in K$.

\medskip

Assume now that $d \geq 1$: some idempotent $f \in R$ occurs in $w_1,\dots,w_p$. Thus, there exists $i \leq p$ and $q \leq k$ such that,
\[
\rho(U[w_i] \cdots U[w_{i+q-1}]) = f
\]
Moreover, it follows from Fact~\ref{fct:idemishere}, that we may choose $i$ and $q$ so that $i + q- 1 \leq k$. Consider the sequence $w_{i+q+1},\dots,w_p$. We distinguish two sub-cases depending on whether $f$ occurs in $w_{i+q+1},\dots,w_p$ as well or not.

\medskip
\noindent
{\bf Case~1: The idempotent $f$ does not occur in $w_{i+q+1},\dots,w_p$.} In that case, the index of $w_{i+q+1},\dots,w_p$ is at most $d-1$. Thus, induction yields a language $K' \in \Kb_{(k+1)d}$ containing $w_{i+q+1} \cdots w_p$. We define our new language $K$ as follows:
\[
K = U[w_1] \cdots U[w_{i+q}] \cdot K'
\]
By definition, we have $w_1 \cdots w_p \in K$. Moreover, $K$ is the concatenation of a language in $\Kb_{k+1}$ (i.e. $U[w_1] \cdots U[w_{i+q}]$ since $i + q- 1 \leq k$) with a language of $\Kb_{(k+1)d}$ (i.e. $K'$). Thus, since $k+1 + (k+1)d = (k+1)(d+1)$, we obtain that $K \in \Kb_{(k+1)(d+1)}$.

\medskip
\noindent
{\bf Case~2: The idempotent $f$ occurs in $w_{i+q+1},\dots,w_p$.} By definition, this means that we get $j \geq i+q+1$ and $r \leq k$, such that,
\[
\rho(U[w_j] \cdots U[w_{j+r-1}]) = f
\]
We work with the largest such $j$. In other words, we choose $j$ so that $f$ does not occur in $w_{j+1},\dots,w_p$ and, therefore neither in $w_{j+r},\dots,w_p$. Hence, the index of $w_{j+r},\dots,w_p$ is at most $d-1$ and induction yields a language $K' \in \Kb_{(k+1)d}$ containing $w_{j+r} \cdots w_p$. We define our new language $K$ as follows:
\[
K = U[w_1] \cdots U[w_{i-1}] \cdot (U[w_i] \cdots U[w_{i+q-1}] \cdot V \cdot U[w_j] \cdots U[w_{j+r-1}]) \cdot K'
\]
Clearly, $U[w_1] \cdots U[w_{i-1}] \cdot (U[w_i] \cdots U[w_{i+q-1}] \cdot V \cdot U[w_j] \cdots U[w_{j+r-1}])\in \Kb_{k+1}$ since $i \leq k+1$. Thus, $K$ is the concatenation of a language in $\Kb_{k+1}$ with another language in $\Kb_{(k+1)d}$ and we get $K \in \Kb_{(k+1)(d+1)}$.

Finally, $w_1 \cdots w_{i+q-1} \in U[w_1] \cdots U[w_{i+q-1}]$ and $w_j \cdots w_p \in U[w_j] \cdots U[w_{j+r-1}] \cdot K'$ by definition. Moreover, since $w_{i+q},\dots,w_{j-1} \in F^\alpha(e,h-1,m-1)$ and $e$ is idempotent, we have $w_{i+q} \cdots w_{j-1} \in \alpha^{-1}(e) \subseteq V$. Altogether, we get $w_1 \cdots w_p\in K$, concluding the proof.

\subsection{Characterization}

We begin with the property characterizing \pol{\Cs}-optimal pointed \imprints. Consider a \Cs-compatible morphism $\alpha: A^* \to M$ and a \mratm $\rho: 2^{A^*} \to R$ (note that there is no constraint on $\rho$). In particular, recall that since $\alpha$ is \Cs-compatible, for any $s \in M$, $\ctype{s}$ is well-defined as a $\sim_\Cs$-class containing $\alpha\inv(s)$.

We say that a subset $S \subseteq M \times R$ is \emph{\pol{\Cs}-saturated} (for $\alpha$ and $\rho$) when it contains \ptriv{\alpha,\rho} and is closed under the following operations: 
\begin{enumerate}
	\item \emph{Downset}: For any $(s,r) \in S$ and $r ' \leq r$, we have $(s,r') \in S$.
	\item \emph{Multiplication}: For any $(s_1,r_1),(s_2,r_2) \in S$, we have $(s_1s_2,r_1r_2) \in S$.
	\item\label{op:half:polclos} \emph{\pol{\Cs}-closure}: For any pair of multiplicative idempotents $(e,f) \in S$, we have,
	\[
	(e,f \cdot \rho(\ctype{e}) \cdot f) \in S
	\]
\end{enumerate}

We may now state the main theorem of the section: the \pol{\Cs}-optimal $\alpha$-pointed $\rho$-\imprint \pocopti is the least \pol{\Cs}-saturated subset of~$M \times R$ (for inclusion).

\begin{theorem}[Characterization of \pol{\Cs}-optimal \imprints] \label{thm:half:mainpolc}
	Let $\alpha: A^* \to M$ be a \Cs-compatible morphism and $\rho:2^{A^*} \to R$ a \mratm. Then, $\pocopti$ is the least \pol{\Cs}-saturated subset of $M \times R$.
\end{theorem}

When $\rho$ is \nice, Theorem~\ref{thm:half:mainpolc} yields an algorithm for computing \pocopti from $\alpha$ and $\rho$. Indeed, one may compute the least \pol{\Cs}-saturated subset of $M \times R$ with a least fixpoint procedure: one starts from \ptriv{\alpha,\rho} and saturates this set with the three operations in the definition (it is clear that all may be implemented). Combined with Proposition~\ref{prop:thereduction}, this yields the desired corollary.

\begin{corollary} \label{cor:half:corpolc}
	Let \Cs be a finite \vari. The \pol{\Cs}-covering problem is decidable.
\end{corollary}

It remains to prove Theorem~\ref{thm:half:mainpolc}. The remainder of the section is devoted to this argument.  Let $\alpha: A^* \to M$ as a \Cs-compatible morphism and $\rho:2^{A^*} \to R$ as a \mratm. We  show that \pocopti is the least \pol{\Cs}-saturated subset of $M \times R$.

We separate the proof in two main steps. Intuitively, they correspond respectively to soundness and completeness of the least fixpoint procedure obtained from the theorem. First, we show that \pocopti is \pol{\Cs}-saturated. This corresponds to soundness (the least fixpoint procedure only computes elements in \pocopti). Then we show that \pocopti is included in any \pol{\Cs}-saturated subset $S  \subseteq  M \times R$. This corresponds to completeness (the least fixpoint procedure computes all elements  in \pocopti).

\subsection{Soundness}

We show that \pocopti is \pol{\Cs}-saturated. Since \pol{\Cs} is a \pvari, we already know from Lemma~\ref{lem:satur} that \pocopti contains \ptriv{\alpha,\rho} and is closed under downset and multiplication. Thus, we may focus on proving that \pocopti satisfies \pol{\Cs}-closure. Consider an idempotent $(e,f) \in \pocopti$. We show that,
\[
(e,f \cdot \rho(\ctype{e}) \cdot f) \in \pocopti
\]
By definition of \pocopti, this means proving that $f \cdot \rho(\ctype{e}) \cdot f \in \opti{\pol{\Cs}}{\alpha\inv(e),\rho}$. We fix an arbitrary optimal \pol{\Cs}-cover \Kb of $\alpha\inv(e)$ and show that $f \cdot \rho(\ctype{e}) \cdot f \in \prin{\rho}{\Kb}$.  We rely on the generic stratification of \pol{\Cs} defined in Section~\ref{sec:tools} (which we may use since \Cs is a finite \vari). Recall that the strata are denoted by \polk{\Cs} and the associated canonical preorders by \polrelk. By definition \Kb is a finite set of languages in \pol{\Cs}. Therefore, there exists some stratum $k \in \nat$ such that all languages in \Kb belong to \polk{\Cs}. Consequently, we have $\opti{\polk{\Cs}}{\alpha\inv(e),\rho} \subseteq \prin{\rho}{\Kb}$ and it now suffices to show that $f \cdot \rho(\ctype{e}) \cdot f \in \opti{\polk{\Cs}}{\alpha\inv(e),\rho}$. By Lemma~\ref{lem:finitedef}, we have to exhibit $w \in \alpha\inv(e)$ and $K \subseteq A^*$ such that $w \polrelk K$ and $f \cdot \rho(\ctype{e}) \cdot f \leq \rho(K)$. This is what we do.

\medskip

We first use our hypothesis that $(e,f) \in \pocopti$ to define $w$ and $K$. This implies that $f \in \opti{\pol{\Cs}}{\alpha\inv(e),\rho}$ and we get $f \in \opti{\polk{\Cs}}{\alpha\inv(e),\rho}$ by Fact~\ref{fct:inclus1}. Therefore, Lemma~\ref{lem:finitedef} yields $v \in \alpha\inv(e)$ and $H \subseteq A^*$ such that $v \polrelk H$ and $f \leq \rho(H)$. Let $p \in \nat$ be the period of \Cs (see Fact~\ref{fct:omegapower}) and $\ell = p2^{k+1}$. We define,
\[
w = v^{\ell}
\]
Clearly, $w \in \alpha\inv(e)$ since $v \in \alpha\inv(e)$ and $e \in M$ is idempotent. Finally, we let,
\[
K = H^\ell \cdot \ctype{e} \cdot H^\ell
\]
By definition of $K$ and since $\rho$ is a \mratm,
\[
(\rho(H))^\ell \cdot \rho(\ctype{e}) \cdot (\rho(H))^\ell = \rho(K)
\]
Moreover, since $f \leq \rho(H)$ by hypothesis and $f$ is idempotent, we obtain,
\[
f \cdot \rho(\ctype{e}) \cdot f \leq \rho(K)
\]
It remains to show that $w \polrelk K$. The proof is based on  Lemma~\ref{lem:hintro:propreo2}. Consider $u \in K$. We show that $w \polrelk u$. This will imply that $w \polrelk K$. By definition, we have $K = H^\ell \cdot \ctype{e} \cdot H^\ell$ and by hypothesis, we have $v \polrelk H$. Thus, since $u \in K$ and \polrelk is compatible with concatenation (see Lemma~\ref{lem:canoquo}), we obtain that there exists some $x \in \ctype{e}$ such that,
\[
v^{\ell} \cdot x \cdot v^{\ell} \polrelk u
\]
Therefore, it now suffices to show that $w \polrelk v^{\ell} \cdot x \cdot v^{\ell}$. It will then be immediate from transitivity that $w \polrelk u$ as desired. Recall that $p$ denotes the period of $\Cs$ and $\ell = p2^{k+1}$. Since $v \in \alpha\inv(e)$ and $e$ is an idempotent of $M$, we have $v^p \in \alpha\inv(e) \subseteq \ctype{e}$. Thus, $x$ and $v^p$ belong to the same $\sim_\Cs$-class: \ctype{e}. We get $v^p \sim_\Cs x$ and  Lemma~\ref{lem:hintro:propreo2} yields that
\[
v^{\ell} \polrelk v^{\ell} \cdot x \cdot v^{\ell}
\]
Since $w = v^\ell$ by definition this is exactly the desired property which concludes the proof.

\subsection{Completeness}

We turn to completeness. Recall that we have a \Cs-compatible morphism $\alpha: A^* \to M$ and a \mratm $\rho:2^{A^*} \to R$. Consider $S \subseteq M \times R$ which is \pol{\Cs}-saturated. Our objective is to prove that $\pocopti \subseteq S$.

\medskip

We present a generic construction which builds \pol{\Cs}-covers $\Kb_s$ of $\alpha\inv(s)$ for all $s \in M$ such that for any $K \in \Kb_s$, we have $(s,\rho(K)) \in S$. By definition of \prin{\rho}{\Kb_s} and since $S$ is closed under downset, (it is \pol{\Cs}-saturated), this will imply,
\[
\prin{\rho}{\Kb_s} \subseteq \{r \mid (s,r) \in S\}
\]
Thus, since $\opti{\pol{\Cs}}{\alpha\inv(s),\rho} \subseteq \prin{\rho}{\Kb_s}$ by definition of \pol{\Cs}-optimal $\rho$-\imprints, we shall obtain,
\[
\opti{\pol{\Cs}}{\alpha\inv(s),\rho} \subseteq \{r \mid (s,r) \in S\}
\]
Finally, by definition of \pocopti, we shall obtain the desired inclusion: $\pocopti \subseteq S$ (recall that $\pocopti  = \{(s,r) \mid r \in \opti{\pol{\Cs}}{\alpha\inv(s),\rho}\}$).

\begin{remark}
	This construction yields a generic method for building optimal \pol{\Cs}-covers of the sets $\alpha\inv(s)$ for $s \in M$. Indeed, we may apply it in the special case when $S = \pocopti$ since we already showed above that \pocopti is \pol{\Cs}-saturated. In this case, we get \pol{\Cs}-covers $\Kb_s$ of $\alpha\inv(s)$	for all $s \in M$ such that $\prin{\rho}{\Kb_s}  = \opti{\pol{\Cs}}{\alpha\inv(s),\rho}$.
\end{remark}

The construction is based on the factorization forest theorem of Simon which we presented in Section~\ref{sec:tools}. Recall that for any $s \in M$ and $h,m \in \nat$, we write $F^\alpha(s,h,m)$ for the language of all words in $\alpha^{-1}(s)$ which admit an $\alpha$-factorization forest of height at most $h$ and idempotent height at most $m$. We use the following proposition.

\begin{proposition} \label{prop:compolc}
	Let $s \in M$ and $h,m \in \nat$. There exists a \pol{\Cs}-cover \Kb of $F^\alpha(s,h,m)$ such that for any $K \in \Kb$, we have $(s,\rho(K)) \in S$.
\end{proposition}

Before proving the proposition, we use it to finish the proof of Theorem~\ref{thm:half:mainpolc}. Consider $s \in M$. Our objective is to build a \pol{\Cs}-cover $\Kb_s$ of $\alpha\inv(s)$ such that,
\[
(s,\rho(K)) \in S \quad \text{for all $K \in \Kb_s$}
\]
Let $h = m = 3|M| - 1$. We obtain from the factorization forest theorem (i.e. Theorem~\ref{thm:facto}) that $\alpha^{-1}(s) = F^\alpha(s,h,m)$. Thus, Proposition~\ref{prop:compolc} yields the desired \pol{\Cs}-cover $\Kb_s$ of $\alpha^{-1}(s)$ such that for any $K \in \Kb_s$, we have $(s,\rho(K)) \in S$.

\medskip

It remains to prove Proposition~\ref{prop:compolc}. Let $h,m \in \nat$ and $s \in M$. Our objective is to build a \pol{\Cs}-cover \Kb of $F^\alpha(s,h,m)$ which satisfies the following property:
\begin{equation} \label{eq:goalpolc}
\text{For all $K \in \Kb$,} \quad (s,\rho(K)) \in S
\end{equation}
The construction is an induction on $h$. We start by considering the base case: $h = 0$.

\begin{remark}
	The proof is actually independent from the idempotent height $m$: this parameter is only useful for the \pbpol{\Cs} proof which we present in Section~\ref{sec:theproof}. We are forced to make it apparent because we intend to use Lemma~\ref{lem:inducidem} whose statement is designed to accommodate both \pol{\Cs} and \pbpol{\Cs}.
\end{remark}

\subsubsection{Base case: Leaves}

Assume that $h = 0$. It follows that all words in $F^\alpha(s,0,m)$ are either empty or made of a single letter $a \in A$. We distinguish two cases depending on whether $s = 1_M$ or not. If $s \neq 1_M$, we define,
\[
\Kb = \{\ctype{\varepsilon}a\ctype{\varepsilon} \mid a \in A \text{ and } \alpha(a) = s\}
\]
Otherwise, $s = 1_M$ and we define,
\[
\Kb = \{\ctype{\varepsilon}a\ctype{\varepsilon} \mid a \in A \text{ and } \alpha(a) = s\} \cup \{\ctype{\varepsilon}\}
\]
Clearly, all languages in \Kb belong to \pol{\Cs} (they are marked concatenations of the language $\ctype{\varepsilon} \in \Cs$ with itself). Moreover, it is also immediate that \Kb  is a cover of $F^\alpha(s,0,m)$. Indeed, a word $w \in F^\alpha(s,0,m)$ is either empty or a single letter and $\alpha(w) = s$. Thus, if $w = \varepsilon$, we have $w \in \ctype{\varepsilon}$ and if $w = a \in A$, we have $w \in \ctype{\varepsilon}a\ctype{\varepsilon}$.

It remains to show that \Kb satisfies~\eqref{eq:goalpolc}. Given $K \in \Kb$, we have to show that $(s,\rho(K)) \in S$. By definition of \Kb this is an immediate consequence of the following lemma.

\begin{lemma} \label{lem:half:leaflemma}
	We have $(1_M,\rho(\ctype{\varepsilon})) \in S$. Moreover, given any letter $a \in A$, we have  $(\alpha(a),\rho(\ctype{\varepsilon}a\ctype{\varepsilon})) \in S$.
\end{lemma}

\begin{proof}
	We begin with the first property: $(1_M,\rho(\ctype{\varepsilon})) \in S$. By definition, we know that $(1_M,\rho(\varepsilon)) = (\alpha(\varepsilon),\rho(\varepsilon)) \in \ptriv{\alpha,\rho}$. Therefore, since $S$ is \pol{\Cs}-saturated, we have $(1_M,\rho(\varepsilon)) \in S$. Moreover, $(1_M,\rho(\varepsilon))$ is clearly a pair of idempotents. Therefore, since $S$ is \pol{\Cs}-saturated, we get from \pol{\Cs}-closure that,
	\[
	(1_M,\rho(\ctype{1_M})) = (1_M,\rho(\varepsilon) \cdot \rho(\ctype{1_M}) \cdot \rho(\varepsilon)) \in S
	\]
	Since $\varepsilon \in \alpha\inv(1_M) \subseteq \ctype{1_M}$, we have 
	$\ctype{\varepsilon} = \ctype{1_M}$ and we get $(1_M,\rho(\ctype{\varepsilon})) \in S$ as desired.
	
	\medskip
	
	It remains to prove that given $a \in A$, we have  $(\alpha(a),\rho(\ctype{\varepsilon}a\ctype{\varepsilon})) \in S$. By definition, we have $(\alpha(a),\rho(a)) \in \ptriv{\alpha,\rho}$. Hence, $(\alpha(a),\rho(a))  \in S$ since $S$ is \pol{\Cs}-saturated. Since $(1_M,\rho(\ctype{\varepsilon})) \in S$, it now follows from closure under multiplication that,
	\[	
	(1_M \cdot \alpha(a) \cdot 1_M,\rho(\ctype{\varepsilon}a\ctype{\varepsilon})) \in S
	\]
	We get as desired that $(\alpha(a),\rho(\ctype{\varepsilon}a\ctype{\varepsilon})) \in S$.
\end{proof}

\subsubsection{Inductive case}

We now assume that $h \geq 1$. Recall that our objective is to build a \pol{\Cs}-cover \Kb of $F^\alpha(s,h,m)$ which satisfies~\eqref{eq:goalpolc}. We decompose $F^\alpha(s,h,m)$ as the union of three languages that we cover independently. Recall that $F^\alpha_B(s,h,m)$ (resp. $F^\alpha_I(s,h,m)$) denotes the language of all words in $\alpha^{-1}(s)$ admitting an $\alpha$-factorization forest of height of at most $h$, of idempotent height at most $m$ and whose root is a \emph{binary node} (resp. \emph{idempotent node}). The construction is based on the two following lemmas.

\begin{lemma} \label{lem:pbinary}
	There exists a \pol{\Cs}-cover $\Kb_B$ of $F^\alpha_B(s,h,m)$ such that for all $K \in \Kb_B$, we have $(s,\rho(K)) \in S$.
\end{lemma}

\begin{lemma} \label{lem:pidem}
	There exists a \pol{\Cs}-cover $\Kb_I$ of $F^\alpha_I(s,h,m)$ such that for all $K \in \Kb_I$, we have $(s,\rho(K)) \in S$.
\end{lemma}

Before we show these two results, let us use them to finish the inductive case. Let $\Kb_B$ and $\Kb_I$ be as defined in the two above lemmas. Fact~\ref{fct:factounion} yields the following:
\[
F^\alpha(s,h,m) = F^\alpha_B(s,h,m) \cup F^\alpha_I(s,h,m) \cup F^\alpha(s,0,0)
\]
Since $h \geq 1$, we obtain from induction on $h$ that we have a \pol{\Cs}-cover $\Kb'$ of $F^\alpha(s,0,0)$ such that for all $K \in \Kb$, we have $(s,\rho(K)) \in S$. Thus, it suffices to define $\Kb = \Kb_B \cup \Kb_I \cup \Kb'$. By definition, we know that \Kb is a \pol{\Cs}-cover of $F^\alpha(s,h,m)$ which satisfies~\eqref{eq:goalpolc} as desired. It remains to prove the two lemmas.

\medskip
\noindent
{\bf Proof of Lemma~\ref{lem:pbinary}.} Using induction we obtain that for all $t \in M$, there exists a \pol{\Cs}-cover $\Ub_t$ of $F^\alpha(t,h-1,m)$ such that for any $U \in \Ub_t$, we have $(t,\rho(U)) \in S$.  We use these new \pol{\Cs}-covers $\Ub_t$ to build the desired $\Kb_B$. We define,
\[
\Kb_B = \{K_1K_2 \mid \text{there exist $t_1,t_2 \in M$ such that $s = t_1t_2$, $K_1 \in \Ub_{t_1}$ and $K_2 \in \Ub_{t_2}$}\}
\]

We need to verify that $\Kb_B$ satisfies the desired properties. We start by proving that it is a \pol{\Cs}-cover of $F_B^\alpha(s,h,m)$. Clearly, all languages in $\Kb_B$ belong to \pol{\Cs} (this follows from Theorem~\ref{thm:pclos} since they are all concatenations of two languages in \pol{\Cs}). Hence, it suffices to verify that $\Kb_B$ is a cover of $F_B^\alpha(s,h,m)$. Let $w \in F_B^\alpha(s,h,m)$, we exhibit $K \in \Kb_B$ such that $w \in K$. By definition, $w$ is the root label of some $\alpha$-factorization forest of height at most $h$, of idempotent height at most $m$ and whose root is a binary node. This exactly says that $w$ admits a decomposition $w = w_1w_2$ with $w_1 \in F^\alpha(\alpha(w_1),h-1,m)$, $w_2 \in F^\alpha(\alpha(w_2),h-1,m)$. Since $\Ub_{\alpha(w_1)}$ and $\Ub_{\alpha(w_2)}$ are covers of $F^\alpha(\alpha(w_1),h-1,m)$ and $F^\alpha(\alpha(w_2),h-1,m)$ respectively, there exist $K_1 \in \Ub_{\alpha(w_1)}$ and $K_2 \in \Ub_{\alpha(w_2)}$ such that $w_1 \in K_1$ and $w_2 \in K_2$. Thus, $w = w_1w_2 \in K_1K_2$ which is an element of $\Kb_B$ by definition since $\alpha(w_1)\alpha(w_2) =  \alpha(w) = s$. 

It remains to verify that for any $K \in \Kb_B$, we have $(s,\rho(K)) \in S$. By construction, $K = K_{1}K_{2}$ where $K_1 \in \Ub_{t_1}$ and $K_2 \in \Ub_{t_2}$ for $t_1,t_2 \in M$ such that $s = t_1t_2$. Moreover, we know that $(t_1,\rho(K_1)) \in S$ and $(t_2,\rho(K_2)) \in S$ by definition of $\Ub_{t_1}$ and $\Ub_{t_2}$. Therefore, since $S$ is \pol{\Cs}-saturated, we obtain from closure under multiplication that $(s,\rho(K)) =  (t_1t_2,\rho(K_1K_2)) \in S$.

\medskip
\noindent
{\bf Proof of Lemma~\ref{lem:pidem}.} Recall that our objective here is to construct a \pol{\Cs}-cover $\Kb_I$ of $F^\alpha_I(s,h,m)$ such that for all $K \in \Kb_I$, we have $(s,\rho(K)) \in S$. Observe that we may assume without loss of generality that $s$ is an idempotent of $M$. Indeed, otherwise we have $F^\alpha_I(s,h,m) = \emptyset$ and we may simply choose $\Kb_I = \emptyset$. We shall write $s = e \in M$ to underline the fact that $s$ is idempotent: we have to cover $F^\alpha_I(e,h,m)$. This is where we use Lemma~\ref{lem:inducidem}. Applying it requires a cover \Ub of $F^\alpha(e,h-1,m-1)$ and a language $V$ which contains $\alpha^{-1}(e)$. Let us first define these objects.

\medskip

We build \Ub by induction. More precisely, it is immediate from induction on the height $h$ that there exists a \pol{\Cs}-cover \Ub of $F^\alpha(e,h-1,m-1)$ such that for any $U \in \Ub$, we have $(e,\rho(U)) \in S$. Furthermore, use $\ctype{e}$ as the language $V$ (note that $\ctype{e}$ contains $\alpha^{-1}(e)$ by definition of \Cs-compatible morphisms).

We now have everything we need for applying Lemma~\ref{lem:inducidem}. We obtain a cover $\Kb_I$ of $F^\alpha_I(e,h,m)$ such that any $K \in \Kb_I$ is a concatenation $K = K_1 \cdots K_n$ where each $K_i$ is of one of the two following kinds:
\begin{enumerate}
	\item $K_i$ is a language in \Ub, or,
	\item $K_i = U_1 \cdots U_\ell \ctype{e} U'_1 \cdots U'_{\ell'}$ where $U_1, \dots,U_\ell,U'_1,\dots,U'_{\ell'} \in \Ub$ and there exists an \emph{idempotent} $f \in R$ such that $\rho(U_1 \cdots U_\ell) = \rho(U'_1 \cdots U'_{\ell'}) = f$.
\end{enumerate}

Clearly, any $K \in \Kb_I$ belongs to \pol{\Cs}: this follows from Theorem~\ref{thm:pclos} since it is a concatenation of languages in \pol{\Cs} by definition ($\ctype{e} \in \Cs \subseteq \pol{\Cs}$ by Lemma~\ref{lem:canoequiv} since it is $\sim_\Cs$-class). Therefore, $\Kb_I$ is a \pol{\Cs}-cover of $F^\alpha_I(e,h,m)$. It remains to prove that for any $K \in \Kb_I$, we have $(e,\rho(K)) \in S$. We need the following fact which is proved using \pol{\Cs}-closure.

\begin{fct} \label{fct:polclos}
	Consider an idempotent $f \in R$ such that $f = \rho(U_1 \cdots U_\ell)$ with $U_1,\dots,U_\ell \in \Ub$. Then, we have $(e,f \cdot \rho(\ctype{e}) \cdot f) \in S$.
\end{fct}

\begin{proof}
	By definition of \Ub, we know that $(e,\rho(U_i)) \in S$ for all $i \leq \ell$. Therefore, since $S$ is \pol{\Cs}-saturated and $e$ is idempotent, closure under multiplication yields that,
	\[
	(e,f) = (e^\ell,\rho(U_1 \cdots U_\ell)) \in S
	\]
	Since $(e,f)$ is an idempotent, we get from \pol{\Cs}-closure that $(e,f \cdot \rho(\ctype{e}) \cdot f) \in S$.	
\end{proof}

We may now prove that $(e,\rho(K)) \in S$ for any $K \in \Kb_I$. Consider $K \in \Kb_I$. By definition $K = K_1 \cdots K_n$ where all languages $K_i$ are as described in the two items above. Clearly, $(e,\rho(K_i)) \in S$ for any $i \leq n$. If $K_i$ is as described in the first item ($K_i \in \Ub$), this is by definition of \Ub. Otherwise, $K_i$ is as described in the second item and this is by Fact~\ref{fct:polclos}. Therefore, since $S$ is \pol{\Cs}-saturated and $e$ is idempotent, we obtain from closure under multiplication that,
\[
(e,\rho(K)) = (e^n,\rho(K_1) \cdots \rho(K_n))  \in S 
\]
This concludes the argument for Lemma~\ref{lem:pidem}.


\section{Covering for \pbpol{\Cs}}
\label{sec:pbpol}
We now turn to our main result: \pbpol{\Cs}-covering is decidable for any finite \vari \Cs. Again, our approach is based on optimal \imprints: we present an effective characterization  of \pbpol{\Cs}-optimal pointed \imprints. For the sake of avoiding clutter, we shall assume that \Cs is fixed for the section.

This characterization is more involved than the one we already obtained for \pol{\Cs}. First, it applies to a more restricted class of \mratms. Specifically, we present a characterization of \pbpopti which holds when $\alpha: A^* \to M$ is a \Cs-compatible morphism and $\rho$ is a \emph{\nice \Cs-compatible} \mratm. In other words, we now need $\rho$ to be \nice and \Cs-compatible (which was not the case for \pol{\Cs}).

\begin{remark}
	Of course, having a characterization restricted to this special case is enough to obtain the desired \pbpol{\Cs}-covering algorithm by Proposition~\ref{prop:thereduction}. However, the fact that we now require $\rho$ to be \nice (which was not the case for \pol{\Cs}) is significant. This explains why the arguments of this paper do not extend to higher levels in concatenation hierarchies. The proof of our characterization for \pbpol{\Cs} relies heavily on the fact that we have a characterization for \pol{\Cs} which holds for {\bf all} \mratms.
\end{remark}

A second point is that our characterization of \pbpol{\Cs}-optimal pointed \imprints actually involves two distinct objects:
\begin{itemize}
	\item As desired, it describes \pbpopti, the \pbpol{\Cs}-optimal $\alpha$-pointed $\rho$-\imprint.
	\item It also describes a second object: \tpocopti, which is the \pol{\Cs}-optimal $\beta$-pointed $\tau$-\imprint where $\beta$ is the canonical morphism associated to $\rho$ (i.e. its restriction to $A^*$), and $\tau$ is an auxiliary \mratm built from $\alpha$ and $\rho$.
\end{itemize}

The key idea here is that our descriptions of \pbpopti and \tpocopti are mutually dependent. Reformulated from an algorithmic point of view, this means that we get a least fixpoint procedure which computes \pbpopti and \tpocopti simultaneously (although $\tau$ is not \nice and can only be used implicitly in the algorithm).

\begin{remark}
	The presence of this second object \tpocopti explains why we shall need to reuse our characterization of \pol{\Cs}-optimal pointed \imprints (i.e. Theorem~\ref{thm:half:mainpolc}). In particular, the \mratm $\tau$ will {\bf not} be \nice. Hence, it will be important that Theorem~\ref{thm:half:mainpolc} holds for all \mratms.
\end{remark}

The section is organized as follows. First, we explain how the auxiliary \mratm $\tau$ is defined from $\alpha$ and $\rho$. Then, we present our characterization of \pbpol{\Cs}-optimal pointed \imprints. Finally, we concentrate on proving our characterization (note that we postpone the difficult direction of this proof to the next section).

\subsection{Auxiliary \mratm}

Let \Ds be a \pvari which is \emph{closed under concatenation}. Consider a \mratm $\rho: 2^{A^*} \to R$ and a morphism $\alpha: A^* \to M$. We build a new \mratm $\drat{\Ds}: 2^{A^*} \to 2^{M \times R}$ (which we shall often simply write ``$\tau$'' when $\alpha$, $\rho$ and \Ds are clear from the context). We shall later use this definition in the special case when $\Ds = \pbpol{\Cs}$. This makes sense since we showed that $\pbpol{\Cs}$ is a \pvari closed under concatenation in Theorem~\ref{thm:pclos}. Let us first explain why the set $2^{M \times R}$ is a hemiring.

Since $2^{M \times R}$ is a set of subsets, it is an idempotent commutative monoid for union. Thus, we simply use union as our addition (the neutral element is $\emptyset$ and the order is inclusion). It remains to define our multiplication.  Given $T_1,T_2 \in 2^{M \times R}$, we define,
\[
T_1 \cdot T_2 = \{(s_1s_2,r) \mid \text{there exists $(s_1,r_1) \in T_1$ and $(s_2,r_2) \in T_2$ such that $r \leq r_1r_2$}\}
\]
One may verify that this is indeed a semigroup multiplication which distributes over the addition (i.e. union) and that $\emptyset$ (the neutral element for union) is a zero. It follows that $2^{M \times R}$ is an idempotent hemiring.

\begin{remark} \label{rem:standmult}
	Our multiplication is not the most immediate one: $T_1 \cdot T_2$ is {\bf not} the set of all multiplications between elements of $T_1$ and $T_2$. It contains more elements: we make sure that $T_1 \cdot T_2$ is closed under downset (if $(s,r) \in T_1 \cdot T_2$ and $r' \leq r$, then $(s,r') \in T_1 \cdot T_2$). We shall need this for proving that $\drat{\Ds}: 2^{A^*} \to 2^{M \times R}$ is a \mratm.
	
	Observe that for this multiplication, $2^{M \times R}$ is a hemiring but not a semiring. There exists no neutral element for multiplication since a set $T \in 2^{M \times R}$ which is not closed under downset cannot be equal to any multiplication.	
\end{remark}

We are now ready to define our new \mratm $\drat{\Ds}: 2^{A^*} \to 2^{M \times R}$. We use the following definition:
\[
\begin{array}{llll}
	\drat{\Ds}: & 2^{A^*} & \to     & 2^{M \times R}                                                                 \\
	      & K       & \mapsto & \{(s,r) \in M \times R \mid \text{$r \in \opti{\Ds}{K \cap \alpha\inv(s),\rho}$}\}
\end{array}
\]
Let us verify that $\drat{\Ds}$ is indeed a \mratm. Note that it is important here that \Ds is a \pvari closed under concatenation.

\begin{lemma} \label{lem:tauismratm}
	The map $\drat{\Ds}$ is a \mratm.
\end{lemma}

\begin{proof}
	We have to show that $\drat{\Ds}$ is a hemiring morphism. Let us first consider addition (which is union for both $2^{A^*}$ and $2^{M \times R}$). Clearly $\drat{\Ds}(\emptyset) = \emptyset$. Indeed,
	\[
	\drat{\Ds}(\emptyset) = \{(s,r) \in M \times R \mid r \in \opti{\Ds}{\emptyset,\rho}\} = \{(s,r) \in M \times R \mid r \in \emptyset\} = \emptyset
	\]
	We now show that for any $K_1,K_2 \subseteq A^*$, we have $\drat{\Ds}(K_1 \cup K_2) = \drat{\Ds}(K_1) \cup \drat{\Ds}(K_2)$. We start with the right to left inclusion. Let $(s,r) \in \drat{\Ds}(K_1) \cup \drat{\Ds}(K_2)$ and by symmetry, assume that $(s,r) \in \drat{\Ds}(K_1)$. It follows that $r \in \opti{\Ds}{K_1 \cap \alpha\inv(s),\rho}$. Since $K_1 \cap \alpha\inv(s) \subseteq (K_1 \cup K_2) \cap \alpha\inv(s)$, it then follows from Fact~\ref{fct:inclus2} that $r \in \opti{\Ds}{(K_1 \cup K_2) \cap \alpha\inv(s),\rho}$ which exactly says that $(s,r) \in \drat{\Ds}(K_1 \cup K_2)$. It remains to treat the left to right inclusion. Let $(s,r) \in \drat{\Ds}(K_1 \cup K_2)$. By definition, $r \in \opti{\Ds}{(K_1 \cup K_2) \cap \alpha\inv(s),\rho}$. Consider optimal \Ds-covers (for $\rho$) $\Ub_1,\Ub_2$ of $(K_1 \cap \alpha\inv(s))$ and $(K_2 \cap \alpha\inv(s))$ respectively. By definition, $\Ub_1 \cup \Ub_2$ is a \Ds-cover of $(K_1 \cup K_2) \cap \alpha\inv(s)$.  Thus, we have $\opti{\Ds}{(K_1 \cup K_2) \cap \alpha\inv(s),\rho} \subseteq \prin{\rho}{\Ub_1 \cup \Ub_2}$ and we obtain $r \in \prin{\rho}{\Ub_1 \cup \Ub_2}$. Therefore, there exists $U \in \Ub_1 \cup \Ub_2$ such that $r \leq \rho(U)$. By symmetry assume that $U \in \Ub_1$, we show that $(s,r) \in \drat{\Ds}(K_1)$ (when $U \in \Ub_2$, one may show that $(s,r) \in \drat{\Ds}(K_2)$). By definition, we have $r \in \prin{\rho}{\Ub_1}$ and since $\Ub_1$ is an optimal \Ds-cover of $K_1 \cap \alpha\inv(s)$, we know that $\prin{\rho}{\Ub_1} = \opti{\Ds}{K_1 \cap \alpha\inv(s),\rho}$. Thus, $r \in \opti{\Ds}{K_1 \cap \alpha\inv(s),\rho}$ which exactly means that $(s,r) \in \drat{\Ds}(K_1)$.

\medskip

	This concludes the proof for addition. We turn to multiplication. We show that $\drat{\Ds}(K_1K_2) = \drat{\Ds}(K_1) \cdot \drat{\Ds}(K_2)$. We start with the inclusion $\drat{\Ds}(K_1) \cdot \drat{\Ds}(K_2) \subseteq \drat{\Ds}(K_1K_2)$. Let $(s,r) \in \drat{\Ds}(K_1) \cdot \drat{\Ds}(K_2)$. Thus, we have $(s_1,r_1) \in \drat{\Ds}(K_1)$ and $(s_2,r_2) \in \drat{\Ds}(K_2)$ such that $s = s_1s_2$ and $r \leq r_1r_2$. By definition of $\drat{\Ds}$, we have	$r_1 \in \opti{\Ds}{K_1 \cap \alpha\inv(s_1),\rho}$ and $r_2 \in \opti{\Ds}{K_2 \cap \alpha\inv(s_2),\rho}$. Since \Ds is a \pvari, the following result is immediate from Lemma~\ref{lem:closmult},
	\[
	r_1r_2 \in \opti{\Ds}{(K_1 \cap \alpha\inv(s_1)) \cdot (K_2 \cap \alpha\inv(s_2)),\rho}
	\]
	Observe that $(K_1 \cap \alpha\inv(s_1)) \cdot (K_2 \cap \alpha\inv(s_2)) \subseteq K_1K_2 \cap \alpha\inv(s_1)\alpha\inv(s_2)$. Moreover, since $\alpha$ is a morphism, we have $\alpha\inv(s_1)\alpha\inv(s_2) \subseteq \alpha\inv(s_1s_2)$. Altogether, this means that we have $(K_1 \cap \alpha\inv(s_1)) \cdot (K_2 \cap \alpha\inv(s_2)) \subseteq K_1K_2 \cap \alpha\inv(s_1s_2)$.  Therefore, we then obtain from Fact~\ref{fct:inclus2} that $r_1r_2 \in \opti{\Ds}{K_1K_2 \cap \alpha\inv(s_1s_2),\rho}$. Thus, we also have $r \in \opti{\Ds}{K_1K_2 \cap \alpha\inv(s_1s_2),\rho}$ by definition of \imprints since $r \leq r_1r_2$. It follows that $(s,r) \in \drat{\Ds}(K_1K_2)$ by definition of $\drat{\Ds}$.
	
	\medskip
	
	We finish with the converse inclusion. Let $(s,r) \in \drat{\Ds}(K_1K_2)$. By definition, we have $r \in \opti{\Ds}{K_1K_2 \cap \alpha\inv(s),\rho}$. For any $t \in M$ and $i \in \{1,2\}$, we define $\Ub_{i,t}$ as an optimal \Ds-cover of $K_i \cap \alpha\inv(t)$. Consider the following finite set of languages \Ub,
	\[
	\Ub = \{U_1U_2 \mid \text{there exists $s_1,s_2 \in M$ s.t. $s_1s_2 = s$, $U_1 \in \Ub_{1,s_1}$ and $U_2 \in \Ub_{2,s_2}$}\}
	\]
	Observe that \Ub is a \Ds-cover of $K_1K_2 \cap \alpha\inv(s)$. Clearly all languages in \Ub belong to \Ds since \Ds is closed under concatenation by hypothesis. Let us show that \Ub is a cover of $K_1K_2 \cap \alpha\inv(s)$. Consider $w \in K_1K_2 \cap \alpha\inv(s)$, we exhibit $U \in \Ub$ such that $w \in U$. Since $w \in K_1K_2$, we have $w = w_1w_2$ with $w_1 \in K_1$ and $w_2 \in K_2$. Let $s_1 = \alpha(w_1)$ and $s_2 = \alpha(w_2)$. Altogether, this means that $w_1 \in K_1 \cap \alpha\inv(s_1)$ and $w_2 \in K_2 \cap \alpha\inv(s_2)$. Therefore, we have $U_1 \in \Ub_{1,s_1}$ and $U_2 \in \Ub_{2,s_2}$ such that $w_1 \in U_1$ and $w_2 \in U_2$. This yields $w \in U_1U_2$. Finally, $s_1s_2 = \alpha(w) = s$ which yields that $U_1U_2 \in \Ub$ by definition.
	
	We may now finish the argument and show that $(s,r) \in \drat{\Ds}(K_1) \cdot \drat{\Ds}(K_2)$.	 Recall that $r \in \opti{\Ds}{K_1K_2 \cap \alpha\inv(s),\rho}$. Thus, since \Ub is a \Ds-cover of $K_1K_2 \cap \alpha\inv(s)$, we have $r \in \prin{\rho}{\Ub}$. It follows that there exists $U \in \Ub$ such that $r \leq \rho(U)$. By definition of \Ub, $U = U_1U_2$  with $U_1 \in \Ub_{1,s_1}$ and $U_2 \in \Ub_{2,s_2}$ where $s_1,s_2 \in M$ satisfy $s_1s_2 = s$. Let $r_1 = \rho(U_1)$ and $r_2 = \rho(U_2)$. Since $\Ub_{1,s_1}$ and $\Ub_{2,s_2}$ are optimal \Ds-covers of $K_1 \cap \alpha\inv(s_1)$ and $K_2 \cap \alpha\inv(s_2)$ respectively, we have $r_1 \in \opti{\Ds}{K_1 \cap \alpha\inv(s_1),\rho}$ and $r_2 \in \opti{\Ds}{K_2 \cap \alpha\inv(s_2),\rho}$. It follows that $(s_1,r_1) \in \drat{\Ds}(K_1)$ and $(s_2,r_2) \in \drat{\Ds}(K_2)$. Since $r \leq \rho(U) = \rho(U_1) \cdot \rho(U_2) = r_1r_2$, it follows by definition of our multiplication that $(s,r) \in \drat{\Ds}(K_1) \cdot \drat{\Ds}(K_2)$ which concludes the proof (note that we used the hypothesis that $\drat{\Ds}(K_1) \cdot \drat{\Ds}(K_2)$ is closed under downset by definition).
\end{proof}
 
\begin{remark} \label{ex:weak:notnice}
	An important point is that $\drat{\Ds}$ is {\bf not} \nice in general even when $\rho$ itself is \nice. Moreover, given a (regular) language $K \subseteq A^*$, computing its image $\drat{\Ds}(K)$ is a difficult task: one needs to compute all \Ds-optimal $\rho$-imprints $\opti{\Ds}{K \cap \alpha\inv(s),\rho}$ for $s \in M$.
	
	Let us provide an example in which $\drat{\Ds}$ is not \nice. Consider the class \Ds which contains all finite languages and $A^*$. One may verify that \Ds is a \pvari. Moreover, we let $M = \{1_M\}$ as the trivial monoid and $\alpha: A^* \to M$ as the only possible morphism. Finally, let  $N = \{1_N,s\}$ be the monoid whose multiplication is defined by $ss = 1_N$ (and $1_N$ is a neutral element). Clearly, $2^N$ is a finite hemiring: the addition is union and the multiplication is obtained by lifting the one of $N$. We let $\rho: 2^{A^*} \to 2^N$ as the \nice \mratm defined by $\rho(a) = \rho(b) = \{s\}$ (for every language $K$, $1_N \in \rho(K)$ if $K$ contains a word of even length and $s \in \rho(K)$ if $K$ contains a word of odd length). We show that in this case the \mratm $\drat{\Ds}$ is not \nice. By definition, for every $K \subseteq A^*$, we have,
 	\[
 	\drat{\Ds}(K) = \{(1_M,r) \mid \text{$r \in \opti{\Ds}{K,\rho}$}\}
 	\]
 	Since the only infinite language in \Ds is $A^*$, a \Ds-cover of $A^*$ must contain it. Hence, since $\rho(A^*) = \{1_N,s\}$, we have $\opti{\Ds}{A^*,\rho} = \{\emptyset,\{1_N\},\{s\},\{1_N,s\}\}$. However, given some word $w \in A^*$, since $\{w\} \in \Ds$, it is simple to verify that,
 	\[
 	\opti{\Ds}{\{w\},\rho} = \left\{\begin{array}{ll}
 	\{\emptyset,\{1_N\}\} & \text{if $w$ has even length} \\
 	\{\emptyset,\{s\}\} & \text{if $w$ has odd length} 
 	\end{array}\right.
 	\]
 	Altogether, it follows that,
 	\[
 	\bigcup_{w \in A^*} \drat{\Ds}(w) = \{(1_M,\emptyset),(1_M,\{1_N\}),(1_M,\{s\})\} \neq \drat{\Ds}(A^*)
 	\]
 	Therefore, $\drat{\Ds}$ is not \nice.	
 \end{remark}
 
 \begin{remark}
	We use these definitions in the case when $\Ds = \pbpol{\Cs}$ and consider the \mratm $\tau = \drata$. We intend to consider \tpocopti, the \pol{\Cs}-optimal $\beta$-pointed $\tau$-\imprint (where $\beta$ is the canonical morphism associated to $\rho$ but this is unimportant for the moment). Since we proved a characterization of \pol{\Cs}-optimal pointed \imprints (Theorem~\ref{thm:half:mainpolc}), we already have a lot of information on \tpocopti. However, what we {\bf cannot} do is use this characterization to directly compute \tpocopti from $\alpha$ and $\rho$. Indeed, implementing the third operation (\pol{\Cs}-closure), requires evaluating $\tau(V)$ when $V$ is some $\sim_\Cs$-class.
\end{remark}

The main consequence of these observations is that we never manipulate $\drat{\Ds}$ in algorithms. It is only an object that we shall use when proving the correction of our characterization of \pbpol{\Cs}-optimal pointed \imprints.

\subsection{Characterization}

We now turn to our characterization of \pbpol{\Cs}-optimal pointed \imprints. Let $\alpha: A^* \to M$ be a \Cs-compatible morphism and $\rho: 2^{A^*} \to R$ be a \Cs-compatible \mratm. We define a notion of \emph{\pbpol{\Cs}-saturated subset of $M \times R$}. Our theorem then states that when $\rho$ is \nice, the least such subset is exactly \pbpopti, the \pbpol{\Cs}-optimal $\alpha$-pointed \imprint on $\rho$.

\begin{remark}
	We do not need the hypothesis that $\rho$ is \nice to define \pbpol{\Cs}-saturated subsets of $M \times R$: the definition makes sense for any \Cs-compatible \mratm. However, we shall need this hypothesis to prove that the least one is \pbpopti.
\end{remark}

\begin{remark}
	Since $\alpha$ is \Cs-compatible, for any $s \in M$, \ctype{s}, is well-defined as a $\sim_\Cs$-class containing $\alpha\inv(s)$.  Moreover, since $\rho$ is \Cs-compatible, for any $r \in R$, $\ctype{r}$ is well-defined as a $\sim_{\Cs}$-class such that for any $w \in A^*$ satisfying $r = \rho(w)$, we have $w \in \ctype{r}$.
\end{remark}

As announced, it turns out that our characterization simultaneously describes two sets. The first one is $\pbpopti \subseteq M \times R$: this is the set that we want to compute. The second one is $\tpocopti \subseteq R \times 2^{M \times R}$ (where $\beta: A^* \to R_{A^*}$ is the canonical morphism associated to $\rho$ and $\tau: 2^{A^*} \to  2^{M \times R}$ is the \mratm \drata built from \pbpol{\Cs}, $\rho$ and $\alpha$): this is an auxiliary object needed for the computation. This is reflected in the definition of \pbpol{\Cs}-saturated objects: it describes the properties satisfied by these two sets. More precisely, the definition applies to pairs $(S,\Ts)$ where $S \subseteq M \times R$ and $\Ts \subseteq R \times 2^{M \times R}$.

\begin{remark}
	We use the multiplication on $2^{M \times R}$ defined at the beginning of the section. It is not the most natural one. Given $T_1,T_2 \in 2^{M \times R}$, their multiplication is,
	\[
	T_1 \cdot T_2 = \{(s_1s_2,r) \mid \text{there exists $(s_1,r_1) \in T_1$ and $(s_2,r_2) \in T_2$ such that $r \leq r_1r_2$}\}
	\]
\end{remark}

Given a pair $(S,\Ts)$ with $S \subseteq M \times R$ and $\Ts \subseteq R \times 2^{M \times R}$ we say that $(S,\Ts)$ is \pbpol{\Cs}-saturated (for $\alpha$ and $\rho$) when the following conditions are satisfied:

\begin{itemize}
	\item The set $S$ contains \ptriv{\alpha,\rho} and is closed under the following operations:
	\begin{enumerate}
		\item {\bf Downset:} If $(s,r) \in S$ and $r' \leq r$, then $(s,r') \in S$.
		\item {\bf Multiplication:} for any $(s_1,r_1),(s_2,r_2) \in S$, we have $(s_1s_2,r_1r_2) \in S$. 
		\item {\bf \pbpol{\Cs}-closure.}\label{op:pbp:lower} For any $(r,T) \in \Ts$ and any idempotent $(e,f) \in T \subseteq M \times R$,
		\[
		(e, f \cdot (r + \rho(\varepsilon))  \cdot f) \in S
		\]
	\end{enumerate}
	\item For any $w \in A^*$, if $r = \rho(w)$ and $s = \alpha(w)$, then $(r,\{(s,r') \mid r' \leq r\}) \in \Ts$. Moreover, \Ts is closed under the following operations:
	
	\begin{enumerate}
		\setcounter{enumi}{3}
		\item {\bf Downset:} If $(r,T) \in \Ts$ and $T' \subseteq T$, then $(r,T') \in \Ts$.
		\item {\bf Multiplication:} For any $(r_1,T_1),(r_2,T_2) \in \Ts$, we have $(r_1r_2,T_1T_2) \in \Ts$.
		\item {\bf Nested closure.}\label{op:pbp:upper} For any idempotent $(f,E) \in \Ts$,
		\[
		(f,E \cdot T \cdot E) \in \Ts \quad \text{where $T = \{(s,r) \in S \mid \ctype{f} = \ctype{s}\}$}	
		\]	
	\end{enumerate} 
\end{itemize}

\begin{remark}
	Observe that the conditions on $S$ and $\Ts$ are mutually dependent. On one hand, the \pbpol{\Cs}-closure operation generates elements of $S$ from elements of \Ts. On the other hand,  nested closure does the opposite, it builds element of \Ts from those in $S$.
\end{remark}

Finally, we say that a subset $S \subseteq M \times R$ is \pbpol{\Cs}-saturated (for $\alpha$ and $\rho$), when there exists another subset $\Ts \subseteq R \times 2^{M \times R}$ such that the pair $(S,\Ts)$ is \pbpol{\Cs}-saturated.

We may now state the main theorem of this section. It turns out that {\bf when $\rho$ is \nice}, the least \pbpol{\Cs}-saturated subset of $M \times R$ is \pbpopti.

\begin{theorem}[Characterization of \pbpol{\Cs}-optimal \imprints] \label{thm:pbp:algo}
	Let $\alpha: A^* \to M$ be a \Cs-compatible morphism and $\rho: 2^{A^*} \to R$ a {\bf \nice} \Cs-compatible \mratm. Then, \pbpopti is the least \pbpol{\Cs}-saturated subset of $M \times R$. 
\end{theorem}

By definition, of \pbpol{\Cs}-saturated sets Theorem~\ref{thm:pbp:algo} yields a least fixpoint algorithm for computing the set \pbpopti from a \Cs-compatible morphism $\alpha: A^* \to M$ and a \nice \Cs-compatible \mratm $\rho 2^{A^*} \to R$. Let us briefly explain how to proceed.

One may compute the least pair $(S,\Ts)$ with $S \subseteq M \times R$ and $\Ts \subseteq R \times 2^{M \times R}$ which is \pbpol{\Cs}-saturated. We use a least fixpoint procedure which starts from the pair $(\ptriv{\alpha,\rho},\Ts')$ where $\Ts'$ contains all pairs $(r,\{(s,r') \mid r' \leq r\})$ where $r = \rho(w)$ and $s = \alpha(w)$ for some $w \in A^*$. Then, we saturate this starting pair with the six operations in the definition which can clearly be implemented. Once we have this least \pbpol{\Cs}-saturated pair $(S,\Ts)$ in hand, we obtain from Theorem~\ref{thm:pbp:algo} that $S$ is exactly the desired set \pbpopti.

\begin{remark}
	Theorem~\ref{thm:pbp:algo} does not describe the least subset $\Ts \subseteq R \times 2^{M \times R}$ such that the pair $(\pbpopti,\Ts)$ is \pbpol{\Cs}-saturated (which we also compute with our least fixpoint procedure). This is because, we are mainly interested in \pbpopti: \Ts is only an auxiliary object that is required for the computation. However, a byproduct of our proof is that \Ts is exactly \tpocopti where $\beta: A^* \to R_{A^*}$ is the canonical morphism which is associated to $\rho$ and $\tau: 2^{A^*} \to  2^{M \times R}$ is the \mratm \drata.
\end{remark}

\begin{remark}
	The hypothesis that $\rho$ is \nice in Theorem~\ref{thm:pbp:algo} is mandatory: the statement does not hold for arbitrary \mratms. This fact is rather intuitive. One may verify from its definition that the least \pbpol{\Cs}-saturated subset of $M \times R$ may only contain pairs $(s,r)$ where the element $r \in R$ is built by summing and multiplying elements of the form $\rho(w) \in R$ (i.e. images of singletons). Thus, considering \pbpol{\Cs}-saturated subsets only makes sense when the \mratm $\rho$ is characterized by the images of singletons: this is exactly the definition of \nice \mratms.
\end{remark}

We complete the characterization by stating the desired corollary. It is now immediate from Proposition~\ref{prop:thereduction} that \pbpol{\Cs}-covering is decidable.

\begin{corollary} \label{cor:pbp:main}
	Let \Cs be a finite \vari. Then, \pbpol{\Cs}-covering is decidable.
\end{corollary}

We turn to the proof of Theorem~\ref{thm:pbp:algo}. In this section, we concentrate on the ``easy'' direction which corresponds to soundness of the least fixpoint procedure: for $\alpha: A^* \to M$ and $\rho: 2^{A^*} \to R$ satisfying the conditions of the theorem, we show that $\pbpopti \subseteq M \times R$ is indeed \pbpol{\Cs}-saturated. The converse direction (i.e. that \pbpopti is the least such subset of $M \times R$) requires more work and we postpone it to the next section.

\subsection{Soundness}

We fix a \Cs-compatible morphism $\alpha: A^* \to M$ and a \nice \Cs-compatible \mratm $\rho: 2^{A^*} \to R$ for the proof. We show that \pbpopti is \pbpol{\Cs}-saturated. 

\begin{remark}
	We do not use the hypothesis that $\rho$ is \nice here. This is only needed for the other direction of the proof.
\end{remark}

By definition, we first need to choose a subset $\Ts \subseteq R \times 2^{M \times R}$ and then prove that the pair $(\pbpopti,\Ts)$ is \pbpol{\Cs}-saturated. Let $\beta: A^* \to R_{A^*}$ be the canonical morphism which is associated to $\alpha$ (recall that $R_{A^*} \subseteq R$ is the monoid $R_{A^*} = \{\rho(w) \mid w \in A^*\}$). Note that $\beta$ is \Cs-compatible by Fact~\ref{fct:canoiscompat}. Moreover, we let $\tau: 2^{A^*} \to  2^{M \times R}$ as the \mratm \drata:
\[
\begin{array}{llll}
\tau: & 2^{A^*} & \to     & 2^{M \times R}                                                                 \\
& K       & \mapsto & \{(s,r) \in M \times R \mid \text{$r \in \opti{\pbpol{\Cs}}{K \cap \alpha\inv(s),\rho}$}\}
\end{array}
\]
Consider the set $\tpocopti \subseteq R_{A^*} \times 2^{M \times R} \subseteq R \times 2^{M \times R}$. As expected, we prove that,
\[
\text{The pair $(\pbpopti,\tpocopti)$ is \pbpol{\Cs}-saturated}
\]
Since \pol{\Cs} and \pbpol{\Cs} are both \pvaris, we already know from Lemma~\ref{lem:satur} that \pbpopti contains \ptriv{\alpha,\rho} and is closed under downset and multiplication. Moreover, we know that \tpocopti is closed under downset and multiplication. Therefore, we may concentrate on proving the three following properties:
\begin{enumerate}	
	\item For any $w \in A^*$, if $r = \rho(w)$ and $s = \alpha(w)$, then $(r,\{(s,r') \mid r' \leq r\}) \in \tpocopti$.
	\item \tpocopti satisfies nested closure (i.e. Operation~\eqref{op:pbp:upper}).
	\item \pbpopti satisfies \pbpol{\Cs}-closure (i.e. Operation~\eqref{op:pbp:lower}).
\end{enumerate}

\subsubsection*{First property} Consider $w \in A^*$, $r = \rho(w)$ and $s = \alpha(w)$. We have to show that the pair $(r,\{(s,r') \mid r' \leq r\})$ belongs to \tpocopti.

Since $\beta(w) = \rho(w)$ (this is the definition of $\beta$), we have $r = \beta(w)$. We show that, 
\[
\{(s,r') \mid r' \leq r\} \subseteq \tau(w)
\]
By definition, it will follow that $(r,\{(s,r') \mid r' \leq r\}) \in \ptriv{\beta,\tau}$ which concludes the proof since $\ptriv{\beta,\tau} \subseteq \tpocopti$ by Lemma~\ref{lem:satur}.

Let $r' \leq r$, we have to show that $(s,r') \in \tau(w)$, i.e. $r' \in \opti{\pbpol{\Cs}}{\{w\} \cap \alpha\inv(s),\rho}$ by definition of $\tau$. Since $s = \alpha(w)$, it is immediate that we have $\{w\} \cap \alpha\inv(s) = \{w\}$. Therefore, we need to show that $r' \in \opti{\pbpol{\Cs}}{\{w\},\rho}$. This is immediate: any cover \Kb of $\{w\}$ must contain a language $K$ such that $w \in K$ which implies that $r' \in \prin{\rho}{\Kb}$ since $r' \leq r = \rho(w)$. This is in particular true for optimal \pbpol{\Cs}-covers of $\{w\}$ which yields $r' \in \opti{\pbpol{\Cs}}{\{w\},\rho}$.

\subsubsection*{Second property} We prove that nested closure is satisfied. Consider an idempotent $(f,E) \in \tpocopti \subseteq R \times 2^{M \times R}$. We have to show that,
\[
(f,E \cdot T \cdot E) \in \tpocopti \quad \text{where $T = \{(s,r) \in \pbpopti \mid \ctype{f} = \ctype{s}\}$}	
\]
We use Theorem~\ref{thm:half:mainpolc}, our characterization of \pol{\Cs}-optimal pointed \imprints. It states that \tpocopti is \pol{\Cs}-saturated (for $\tau$). In particular, it satisfies \pol{\Cs}-closure which yields,
\[
(f,E \cdot \tau(\ctype{f}) \cdot E) \in \tpocopti
\]
Therefore, it suffices to show that $T = \tau(\ctype{f})$ to conclude the proof. This is what we do. By definition of $\tau$, we know that,
\[
\tau(\ctype{f}) = \{(s,r) \in M \times R \mid \text{$r \in \opti{\pbpol{\Cs}}{\ctype{f} \cap \alpha\inv(s),\rho}$}\}
\]
Since $\alpha: A^* \to M$ is \Cs-compatible, we know that any $\alpha\inv(s)$ is included in the $\sim_{\Cs}$-class \ctype{s}. Therefore, the intersection $\ctype{f} \cap \alpha\inv(s)$ is either equal to $\alpha\inv(s)$ (when $\ctype{f} = \ctype{s}$) or empty (when $\ctype{f} \neq \ctype{s}$). In the latter case, \opti{\pbpol{\Cs}}{\ctype{f} \cap \alpha\inv(s),\rho} is empty. Therefore, we obtain that,
\[
\tau(\ctype{f}) = \{(s,r) \in M \times R \mid \text{$\ctype{f} = \ctype{s}$ and $r \in \opti{\pbpol{\Cs}}{\alpha\inv(s),\rho}$}\}
\]
Moreover, we have $\pbpopti = \{(s,r) \in M \times R \mid \text{$r \in \opti{\pbpol{\Cs}}{\alpha\inv(s),\rho}$} \}$ by definition. Therefore, this can be reformulated as,
\[
\tau(\ctype{f}) = \{(s,r) \in \pbpopti \mid \ctype{f} = \ctype{s}\} = T
\]
This concludes the proof.

\subsubsection*{Third property} We show that \pbpol{\Cs}-closure is satisfied. Consider a pair $(r,T) \in \tpocopti \subseteq R \times 2^{M \times R}$. Moreover, let $(e,f) \in T \subseteq M \times R$ be an idempotent. We have to show that,
\[
(e, f \cdot (r + \rho(\varepsilon))  \cdot f) \in \pbpopti
\]
This requires more work. We shall need the following simple fact. Recall that given a \pvari \Ds, we write $\compc{\Ds}$ for the \pvari containing all complements of languages in \Ds.

\begin{fct} \label{fct:pbp:thefinitepvari}
	There exists a {\bf finite} \pvari $\Ds \subseteq \pol{\Cs}$ such that,
	\[\pbpopti = \popti{\pol{\compc{\Ds}}}{\alpha,\rho}\]
\end{fct}

\begin{proof}
	For all $s \in M$, let $\Kb_s$ be an optimal $\pbpol{\Cs}$-cover of $\alpha\inv(s)$ (for $\rho$). Recall that we showed in Lemma~\ref{lem:comptrick} that $\pbpol{\Cs} = \pol{\compc{\pol{\Cs}}}$. Thus, since there are finitely many languages in the sets $\Kb_s$ for $s \in M$, there exists some $k \in \nat$ such that all these languages belong to $\pol{\compc{\polk{\Cs}}}$. Here, \polk{\Cs} denotes a stratum in our stratification of \pol{\Cs} into finite \pvaris (see Section~\ref{sec:tools} for details). It now suffices to choose $\Ds = \polk{\Cs}$.
\end{proof}

In view of Fact~\ref{fct:pbp:thefinitepvari}, it now suffices to show that,	
\[
(e, f \cdot (r + \rho(\varepsilon))  \cdot f)  \in \popti{\pol{\compc{\Ds}}}{\alpha,\rho}
\]
By definition, this amounts to proving that $f \cdot (r + \rho(\varepsilon))  \cdot f \in \opti{\pol{\compc{\Ds}}}{\alpha\inv(e),\rho}$. We fix an arbitrary optimal \pol{\compc{\Ds}}-cover \Kb of $\alpha\inv(e)$ and show that $f \cdot (r + \rho(\varepsilon))  \cdot f \in \prin{\rho}{\Kb}$. Since $\compc{\Ds}$ is a finite \pvari, we have a stratification of $\pol{\compc{\Ds}}$ which we introduced in Section~\ref{sec:tools}. Recall that the strata are denoted by \polk{\compc{\Ds}} and the associated canonical preorder relations by \polrelk. Since \Kb is a finite set of languages in \pol{\compc{\Ds}}, we have a stratum $k \in \nat$ such that all languages in \Kb belong to \polk{\compc{\Ds}}.  Consequently, we have $\opti{\polk{\compc{\Ds}}}{\alpha\inv(e),\rho} \subseteq \prin{\rho}{\Kb}$ and it now suffices to show that $f \cdot (r + \rho(\varepsilon))  \cdot f \in \opti{\polk{\compc{\Ds}}}{\alpha\inv(e),\rho}$. By Lemma~\ref{lem:finitedef}, we have to exhibit $w \in \alpha\inv(e)$ and $K \subseteq A^*$ such that $w \polrelk K$ and $f \cdot (r + \rho(\varepsilon))  \cdot f \leq \rho(K)$. This is what we do.

\medskip

We first use our hypotheses that $(r,T) \in \tpocopti$ and $(e,f) \in T$ to introduce a few objects that we need to define $w$ and $K$. We do so in the following lemma. Recall that \canod denotes the canonical preorder associated to the finite \pvari \Ds (which is compatible with word concatenation by Lemma~\ref{lem:canoquo}).

\begin{lemma} \label{lem:pbp:theobjs}
	There exists $u,v \in A^*$ and $H \subseteq A^*$ satisfying the following conditions: \begin{itemize}
		\item $\rho(u) = r$, $\alpha(v) = e$ and $f \leq \rho(H)$.
		\item $u \canod v$ and $v \polrelk H$.
	\end{itemize}
\end{lemma}

\begin{proof}
	We first define $u,v \in A^*$ and $H \subseteq A^*$. Since $(r,T) \in \tpocopti$ and $\Ds \subseteq \pol{\Cs}$ it follows from Fact~\ref{fct:inclus1} that $(r,T) \in \popti{\Ds}{\beta,\tau}$. This implies that $T \in \opti{\Ds}{\beta\inv(r),\tau}$. Therefore, Lemma~\ref{lem:finitedef} yields $u \in \beta\inv(r)$ and $L \subseteq A^*$ such that $T \subseteq \tau(L)$ and $u \canod L$. Moreover, since $(e,f) \in T \subseteq \tau(L)$, we obtain by definition of $\tau$ that,
	\[
	f \in \opti{\pbpol{\Cs}}{L \cap \alpha\inv(e),\rho} 
	\]
	Since $\Ds \subseteq \pol{\Cs}$, we have $\polk{\compc{\Ds}} \subseteq \pol{\compc{\Ds}} \subseteq \pbpol{\Cs}$. Therefore, it follows from Fact~\ref{fct:inclus1} that $f \in \opti{\polk{\compc{\Ds}}}{L \cap \alpha\inv(e),\rho}$. Using Lemma~\ref{lem:finitedef}, this yields $v \in L \cap \alpha\inv(e)$ and $H \subseteq A^*$ such that $v \polrelk H$ and $f \leq \rho(H)$.
	
	Let us verify that $u,v \in A^*$ and $H \subseteq A^*$ satisfy the conditions described in the lemma. We have $\rho(u) = r$ since $u \in \beta\inv(r)$ which means that $\rho(u) =\beta(u) = r$. Moreover, $\alpha(v) = e$,  $f \leq \rho(H)$ and $v \polrelk H$	hold by definition of $v$ and $H$. Finally, $u \canod v$ holds since $v \in L$ and we have $u \canod L$.	
\end{proof}

We are now ready to define the appropriate $w \in \alpha\inv(e)$ and $K \subseteq A^*$ such that $w \polrelk K$ and $f \cdot (r + \rho(\varepsilon))  \cdot f \leq \rho(K)$. We fix $u,v \in A^*$ and $H \subseteq A^*$ as defined in Lemma~\ref{lem:pbp:theobjs}. Let $p \geq 1$ be the period of $\compc{\Ds}$ (see Fact~\ref{fct:omegapower} for the definition) and $\ell = p \times 2^{k+1}$. We define,
\[
w = v^{\ell} \quad \text{and} \quad K = H^{\ell + p-1} u H^\ell \cup H^{\ell}
\]
Clearly, $w \in \alpha\inv(e)$ since $\alpha(v) = e$ and $e$ is an idempotent of $M$. It remains to show that $w \polrelk K$ and $f \cdot (r + \rho(\varepsilon)) \cdot f \leq \rho(K)$. We start with the latter. By definition and since $\rho$ is a \mratm,
\[
\begin{array}{lll}
\rho(K) & = & (\rho(H))^{\ell+p-1} \cdot \rho(u) \cdot (\rho(H))^\ell + (\rho(H))^{\ell} \\
& = & (\rho(H))^{\ell+p-1} \cdot \rho(u) \cdot (\rho(H))^\ell + (\rho(H))^{\ell-1} \cdot \rho(\varepsilon) \cdot \rho(H)
\end{array}
\]
By definition in Lemma~\ref{lem:pbp:theobjs}, we have $\rho(u) = r$ and $f \leq \rho(H)$. Thus, since $f$ is idempotent, we obtain as desired that,
\[
f \cdot (r + \rho(\varepsilon)) \cdot f = f \cdot \rho(u) \cdot f + f \cdot \rho(\varepsilon) \cdot f \leq \rho(K)
\]

We finish with the proof that $w \polrelk K$. Let $x \in K$, we show that $w \polrelk x$. By definition, we have $K = H^{\ell + p-1} u H^\ell \cup H^{\ell}$. Moreover, we know that $v \polrelk H$ by definition in Lemma~\ref{lem:pbp:theobjs}. Thus, since $x \in K$ and \polrelk is compatible with concatenation (see Lemma~\ref{lem:canoquo}), we obtain that one of the two following properties hold:
\begin{itemize}
	\item $v^{\ell} \polrelk x$ (when $x \in H^\ell$) or,
	\item $v^{\ell+p-1} \cdot u \cdot v^{\ell} \polrelk x$ (when $x \in H^{\ell + p-1} u H^\ell$).
\end{itemize}
Thus, by transitivity, it suffices to show that $w \polrelk v^{\ell}$ and $w \polrelk v^{\ell+p-1} \cdot u \cdot v^{\ell}$. The former is immediate: we have $w = v^{\ell}$ by definition.

It remains to show that $w \polrelk v^{\ell+p-1} \cdot u \cdot v^{\ell}$. Recall that \polrelk is the canonical preorder of \polk{\compc{\Ds}}. By definition in Lemma~\ref{lem:pbp:theobjs}, we have $u \canod v$. Moreover, since \canod is compatible with concatenation, this implies $v^{p-1}u \canod v^p$. It follows that $v^p \leqslant_{\compc{\Ds}} v^{p-1}u$. Indeed, $v^{p-1}u \canod v^p$ means that for any $L \in \Ds$, we have $v^{p-1}u \in L \Rightarrow v^p \in L$. The contrapositive then states that for any $L \in \Ds$, we have $v^p \not\in L \Rightarrow v^{p-1}u \not\in L$. Finally, since the languages of $\compc{\Ds}$ are the complements of those in \Ds, it follows that for all $L \in \compc{\Ds}$, we have $v^p \in L \Rightarrow v^{p-1}u \in L$, i.e. $v^p \leqslant_{\compc{\Ds}} v^{p-1}u$. Therefore, since $w = v^\ell$ and $\ell = p2^{k+1}$ where $p$ is the period of $\compc{\Ds}$, we obtain as desired that $w \polrelk v^{\ell} \cdot v^{p-1}u \cdot v^{\ell}$ from Lemma~\ref{lem:hintro:propreo2}.


\section{Completeness in Theorem~\ref{thm:pbp:algo}}
\label{sec:theproof}
This last section is devoted to proving the difficult direction in Theorem~\ref{thm:pbp:algo}. It corresponds to completeness in the least fixpoint algorithm for \pbpol{\Cs}-covering.

\medskip

We let $\alpha: A^* \to M$ as a \Cs-compatible morphism and $\rho: 2^{A^*} \to R$ as a \nice \Cs-compatible \mratm. Recall that Theorem~\ref{thm:pbp:algo} states that \pbpopti  is the smallest \pbpol{\Cs}-saturated subset of $M \times R$. We have already showed in the previous section that \pbpopti is \pbpol{\Cs}-saturated (this was the soundness proof). Therefore, it remains to show that it is the smallest such subset. This is the purpose of this section.

Thus, we fix some arbitrary \pbpol{\Cs}-saturated subset $S$ of $M \times R$ and prove that it includes \pbpopti. That is, we show,
\[
\pbpopti \subseteq S
\]
The argument is more involved than what we did for \pol{\Cs} in Section~\ref{sec:polc}. Indeed, the definition of \pbpol{\Cs}-saturated sets is involved and this is reflected in our argument.

We write $\beta: A^* \to R_{A^*}$ for the canonical morphism associated to $\rho$ (which is \Cs-compatible by Fact~\ref{fct:canoiscompat}). Moreover, we write $\tau: 2^{A^*} \to 2^{M \times R}$ for the \mratm \drata:
\[
\begin{array}{llll}
\tau: & 2^{A^*} & \to     & 2^{M \times R}                                                                 \\
& K       & \mapsto & \{(s,r) \in M \times R \mid \text{$r \in \opti{\pbpol{\Cs}}{K \cap \alpha\inv(s),\rho}$}\}
\end{array}
\]
Since $S$ is \pbpol{\Cs}-saturated, the definition yields a subset $\Ts \subseteq R \times 2^{M \times R}$ such that the pair $(S,\Ts)$ is \pbpol{\Cs}-saturated. We actually prove the two following inclusions,
\[
\pbpopti \subseteq S \quad \text{and} \quad \tpocopti \subseteq \Ts 
\]
Of course, the difficulty is that these two properties are mutually dependent. This is not surprising since the conditions which make $S$ and $\Ts$ \pbpol{\Cs}-saturated are mutually dependent themselves. We cope with this problem by following three steps:
\begin{enumerate}
	\item {\bf Step~1: \pol{\Cs} argument:} First, we investigate the set \Ts. Let us point out that we do not yet show that $\tpocopti \subseteq \Ts$: getting this inclusion requires information on the set $S$ and we do not have any at this point. Instead, we cope with this problem by introducing a new \mratm $\gamma: 2^{A^*} \to 2^{M \times R}$ which is intuitively a weaker variant of $\tau$ parametrized by the set $S$ (for any $K$, $\gamma(K)$ is a restriction of the set $\tau(K) \subseteq M \times R$ which depends on $S$). The main result of this step is as follows:
	\[
	\gpocopti \subseteq \Ts
	\]
	The key idea is that the definition of $\gamma$ depends on $S$. Therefore, later in the proof, whenever we exhibit new elements in $S$, this simultaneously proves the existence of new elements in \gpocopti. Consequently, the above inclusion yields that these new elements belong to \Ts as well (and we may use them to exhibit even more elements in $S$). This is based on Theorem~\ref{thm:half:mainpolc} (our characterization of \pol{\Cs}-optimal pointed \imprints): we apply it to the \mratm $\gamma$.
	
	\item {\bf Step~2: a finite \vari $\Ds \subseteq \bpol{\Cs}$:} This step reexamines the results obtained in Step~1. We use our new \mratm $\gamma$ to define a key object of the proof: a finite \vari $\Ds \subseteq \bpol{\Cs}$. Intuitively, \Ds is designed to contain any language involved in an optimal \pol{\Cs}-cover of $\beta\inv(r)$ for $\gamma$ where $r \in R_{A^*}$.
	
	\item {\bf Step~3: \pbpol{\Cs}-argument:} This step uses all preliminary results that we have obtained in the previous ones to finally prove the desired property:
	\[
	\pbpopti \subseteq S 
	\]
\end{enumerate}

\subsection{\texorpdfstring{Step~1: \pol{\Cs} argument}{Step~1: Pol(C) argument}}

We start by defining the auxiliary \mratm $\gamma: 2^{A^*} \to 2^{M \times R}$ that we announced above. Then, we use Theorem~\ref{thm:half:mainpolc} to prove that,
\[
\gpocopti \subseteq \Ts
\]
As we explained, $\gamma$ is designed as a weaker variant of $\tau$ parametrized by the set $S \subseteq M \times R$. We use this new object to cope with the fact that we do not yet have enough information on $S$ to prove directly that $\tpocopti \subseteq \Ts$. The definition is based on the notion of good languages which we present now.

\medskip
\noindent
{\bf Good languages.} Given some language $G \subseteq A^*$, we shall say that $G$ is \emph{good} when the two following conditions are satisfied:
\begin{itemize}
	\item $G$ is closed under infixes: for any $u,v_1,v_2 \in A^*$, if $v_1uv_2 \in G$, then $u \in G$.
	\item $\tau(G) \subseteq S$.
\end{itemize}

Note that we already know at least one simple good language: the empty one ($\emptyset$) which trivially satisfies both conditions.

\begin{remark} \label{rem:pbp:univgood}
	Our ultimate goal in the proof is to prove that $\pbpopti \subseteq S$. Thus, since $\tau(A^*) = \{(s,r) \mid r \in \opti{\pbpol{\Cs}}{\alpha\inv(s),\rho}\} = \pbpopti$ by definition, this means that the universal language $A^*$ is good. However, we {\bf do not know this yet}. While this will be implicit, a key point in later steps of the proof is that it exhibits more and more elements within $S$. Therefore, we shall obtain increasingly large good languages until we prove that $A^*$ itself is good.
\end{remark}

We shall need the following fact which is used for combining two good languages into a single larger one.

\begin{fact} \label{fct:pbp:concatgood}
	Let $G_1,G_2$ be good languages. Then $G_1G_2 \cup G_1 \cup G_2$ is a good language as well.
\end{fact}

\begin{proof}
	Clearly, $G_1G_2 \cup G_1 \cup G_2$  is closed under infixes since this was the case for both $G_1$ and $G_2$. Hence, we may concentrate on the second property in the definition of good languages. Since $\tau$ is a \mratm, we have,
	\[
	\tau(G_1G_2 \cup G_1 \cup G_2) = \tau(G_1) \cdot \tau(G_2) \cup \tau(G_1) \cup \tau(G_2)
	\]
	We already know that $\tau(G_1) \subseteq S$ and $\tau(G_2) \subseteq S$ since $G_1,G_2$ are good. Thus, it suffices to verify that $\tau(G_1) \cdot \tau(G_2) \subseteq S$.  By definition, $(s,r) \in \tau(G_1) \cdot \tau(G_2)$ satisfies $s = s_1s_2$ and $r \leq r_1r_2$ for $(s_1,r_1) \in \tau(G_1) \subseteq S$ and $(s_2,r_2) \in \tau(G_2) \subseteq S$. Since $S$ is \pbpol{\Cs}-saturated, it is closed under downset and multiplication which yields $(s,r) \in S$.
\end{proof}

\noindent
{\bf Definition of $\gamma$.} We now use good languages to introduce our new \mratm $\gamma: 2^{A^*} \to 2^{M \times R}$. As we explained, it is defined from $\tau$. Recall that $2^{M \times R}$ is a finite hemiring which we already used as the evaluation set for $\tau$. In particular, recall that it uses a special multiplication which we defined above. We define,
\[
\begin{array}{llll}
\gamma: & 2^{A^*} & \to     & 2^{M \times R}                                     \\
& K       & \mapsto & \{(s,r) \mid \text{$(s,r) \in \tau(K \cap G)$ for some good language $G$}\}
\end{array}
\]
Note that the definition of $\gamma$ depends on $S \subseteq M \times R$ since this set constrains which languages are good.

\begin{remark} \label{rem:pbp:thisiswierd}
	This new map $\gamma$ is strongly related to $\tau$. An apparent connection is that for any $K \subseteq A^*$, we have $\gamma(K) \subseteq \tau(K)$. Indeed, given $(s,r) \in \gamma(K)$, we have $(s,r) \in \tau(K \cap G)$ for some good language $G$. Moreover, we know that $\tau(K \cap G) \subseteq \tau(K)$ since $K \cap G \subseteq K$. Thus, we get that $(s,r) \in \tau(K)$. However, the connection between $\gamma$ and $\tau$ is in fact much stronger. Since we shall later obtain that $A^*$ is good (see Remark~\ref{rem:pbp:univgood}), the converse inclusion holds as well: $\gamma(K) \subseteq \tau(K)$. Thus, $\gamma$ and $\tau$ are actually the same object. Unfortunately, this is an information that we are not able to use. Indeed, we only obtain it at the end of our proof: the fact that $A^*$ is good follows from our end goal ($\pbpopti \subseteq S$).
	
	Another way to explain the situation is as follows. We need to use $\gamma$ since we are not able to prove the inclusion $\tpocopti\subseteq \Ts$ directly. However, it turns out that when we shall finally obtain $\pbpopti\subseteq S$ (the proof relies on the inclusion $\gpocopti \subseteq \Ts$),  we will also realize that $\gamma$ was $\tau$ all along. 
\end{remark}

Before we can use $\gamma$, we first need to verify that it is indeed a \mratm. We do so in the next lemma (this is a consequence of the fact that $\tau$ is a \mratm itself).

\begin{lemma} \label{lem:pbp:itisgood}
	The map $\gamma$ is a \mratm.
\end{lemma}

\begin{proof}
	We need to prove that $\gamma$ is a hemiring morphism. We start by proving that it is a monoid morphism for addition. Clearly, $\gamma(\emptyset) = \emptyset$. Indeed,
	\[
	\gamma(\emptyset) = \{(s,r) \mid \text{$(s,r) \in \tau(\emptyset)$}\} = \emptyset
	\]
	We now show that for any $K_1,K_2 \subseteq A^*$, we have $\gamma(K_1 \cup K_2) = \gamma(K_1) \cup \gamma(K_2)$.  Assume first that $(s,r) \in \gamma(K_1) \cup \gamma(K_2)$. By symmetry assume that $(s,r) \in \gamma(K_1)$. By definition, we have a good language $G$ such that $(s,r) \in \tau(K_1 \cap G)$. Thus, since $K_1 \cap G \subseteq (K_1 \cup K_2) \cap G$, we get $(s,r) \in \tau((K_1 \cup K_2) \cap G)$ which means that $(s,r) \in \gamma(K_1 \cup K_2)$ as desired. Conversely, assume that $(s,r) \in \gamma(K_1 \cup K_2)$. By definition, there exists a good language $G$ such that $(s,r) \in \tau((K_1 \cup K_2) \cap G)$. Thus, since $\tau$  is a morphism itself, we have $(s,r) \in \tau(K_1\cap G) \cup \tau(K_2 \cap G)$ and it follows that $(s,r) \in \gamma(K_1) \cup \gamma(K_2)$. 
	
	\medskip
	
	We now show that $\gamma$ is a semigroup morphism for multiplication. Let $K_1,K_2 \subseteq A^*$. We show that $\gamma(K_1K_2) = \gamma(K_1) \cdot \gamma(K_2)$. Assume first that $(s,r) \in \gamma(K_1) \cdot \gamma(K_2)$. Thus, by definition of our multiplication on $2^{M \times R}$, we have $(s_1,r_1) \in \gamma(K_1)$ and $(s_2,r_2) \in \gamma(K_2)$ such that $s = s_1s_2$ and $r \leq r_1r_2$. It then follows from the definition of $\gamma$ that there exist two good languages $G_1,G_2$ such that,
	\[	
	(s_1,r_1) \in \tau(K_1 \cap G_1) \quad \text{and} \quad (s_2,r_2) \in \tau(K_2 \cap G_2)
	\]
	Since we already know that $\tau$ is \tame, it follows that $(s,r) \in \tau((K_1\cap G_1)(K_2\cap G_2))$. Moreover, because $(K_1 \cap G_1)(K_2 \cap G_2) \subseteq K_1K_2 \cap (G_1G_2 \cup G_1 \cup G_2)$, we get $(s,r) \in \tau(K_1K_2 \cap (G_1G_2 \cup G_1 \cup G_2))$.  Finally, since we know that $G_1G_2 \cup G_1 \cup G_2$ is good by Fact~\ref{fct:pbp:concatgood}, it follows that $(s,r) \in \gamma(K_1K_2)$ by definition of $\gamma$.
	
	We finish with the converse inclusion. Let $(s,r) \in \gamma(K_1K_2)$. By definition, we get a good language $G$ such that $(s,r) \in \tau(K_1K_2 \cap G)$. Since $G$ is closed under infixes (by definition of good languages), it is simple to verify that $K_1K_2 \cap G \subseteq (K_1 \cap G)(K_2 \cap G)$. Thus, it follows that,
	\[
	(s,r) \in \tau((K_1 \cap G)(K_2 \cap G)) = \tau(K_1 \cap G) \cdot \tau(K_2 \cap G)
	\]
	Therefore, we have $(s_1,r_1) \in \tau(K_1 \cap G)$ and $(s_2,r_2) \in \tau(K_2 \cap G)$ such that $s =s_1s_2$ and $r \leq r_1r_2$. By definition, $(s_1,r_1)\in \gamma(K_1)$ and $(s_2,r_2) \in \gamma(K_2)$. Thus, we obtain that,
	$(s,r) \in \gamma(K_1) \cdot \gamma(K_2)$.
\end{proof}

We turn to the main result in Step~1. Using Theorem~\ref{thm:half:mainpolc} (i.e. our characterization of \pol{\Cs} optimal pointed \imprints) we prove that \gpocopti is included in \Ts.  Recall that $\Ts \subseteq R \times 2^{M \times R}$ is a set such that the pair $(S,\Ts)$ is \pbpol{\Cs}-saturated (for $\rho$).

\begin{proposition} \label{prop:pbp:level1}
	We have $\gpocopti \subseteq \Ts$.
\end{proposition}

\begin{proof}
	We know that \gpocopti is the least \pol{\Cs}-saturated subset of $R_{A^*} \times 2^{M \times R}$ (for $\gamma$) by Theorem~\ref{thm:half:mainpolc}. Thus, it suffices to show that \Ts is \pol{\Cs}-saturated for $\gamma$. It will then be immediate that $\gpocopti \subseteq \Ts$ as desired.
	
	\begin{remark}
		Note that we are able to use Theorem~\ref{thm:half:mainpolc} since $\beta$ is \Cs-compatible and $\gamma$ is a \mratm. However, let us point out that $\gamma$ need not be \nice in general. Therefore, it is important here that this is not a requirement for applying Theorem~\ref{thm:half:mainpolc}.
	\end{remark}
	
	Recall that by hypothesis, the pair $(S,\Ts)$ is \pbpol{\Cs}-saturated for $\rho$. Hence, we already know that \Ts is closed under downset and multiplication. It remains to show that $\ptriv{\beta,\gamma} \subseteq \Ts$ and that \Ts satisfies \pol{\Cs}-closure.
	
	We start with $\ptriv{\beta,\gamma} \subseteq \Ts$. Let $(r,T) \in \ptriv{\beta,\gamma}$, we show that $(r,T) \in \Ts$. By definition of \ptriv{\beta,\gamma}, there exists $w \in A^*$ such that $\rho(w) = \beta(w) = r$ and $T \subseteq \gamma(w)$. Let $s = \alpha(w)$ and consider the set $T' = \{(s,r') \mid r' \leq r\}$. Since $(S,\Ts)$ is \pbpol{\Cs}-saturated for $\rho$, it is immediate by definition that $(r,T') \in \Ts$. We show that $T \subseteq T'$. Since we already established that \Ts is closed under downset, this will imply that $(r,T) \in \Ts$. Let $(t,r') \in T$. We have to show that $s = t$ and $r' \leq r$. Since $T \subseteq \gamma(w)$, we have $(t,r') \in \gamma(w)$. Furthermore, we have $(t,r') \in \tau(w)$ since we observed in Remark~\ref{rem:pbp:thisiswierd} that $\gamma(w) \subseteq \tau(w)$.  By definition of $\tau$, this means that $r' \in \opti{\pbpol{\Cs}}{\{w\} \cap \alpha\inv(t),\rho}$. Therefore, $\{w\} \cap \alpha\inv(t) \neq \emptyset$ which means that $\{w\} \cap \alpha\inv(t) = \{w\}$ and $t = \alpha(w) = s$. Moreover, since \pbpol{\Cs} contains all finite languages by Lemma~\ref{lem:pbpfinite}, $\Kb = \{\{w\}\}$ is a \pbpol{\Cs}-cover of $\{w\} = \{w\} \cap \alpha\inv(t)$. Hence, we have $r' \in \opti{\pbpol{\Cs}}{\{w\} \cap \alpha\inv(t),\rho} \subseteq \prin{\rho}{\Kb}$ which means any $r' \leq \rho(w) = r$ by definition.
	
	We turn to \pol{\Cs}-closure. Let $(f,E) \in \Ts$ be an idempotent. We have to show that,
	\[
	(f, E \cdot \gamma(\ctype{f})  \cdot E) \in \Ts
	\]
	Let $T = \{(s,r) \in S \mid \ctype{f} = \ctype{s}\}$. Since $(S,\Ts)$ is \pbpol{\Cs}-saturated for $\rho$, we get from nested closure (i.e. Operation~\eqref{op:pbp:upper}) that,
	\[
	(f, E \cdot T  \cdot E) \in \Ts
	\]
	Hence, by closure under downset, it suffices to show that $\gamma(\ctype{f}) \subseteq T$. Let $(s,r) \in \gamma(\ctype{f})$. By definition, we know that $(s,r) \in \tau(\ctype{f} \cap G)$ for some good language $G$. Therefore, $(s,r) \in \tau(G)$ which is a subset of $S$ by the second item in the definition of good languages. We get that $(s,r) \in S$. Moreover, $(s,r) \in \tau(\ctype{f})$ which means that we have $r \in \opti{\pbpol{\Cs}}{\ctype{f} \cap \alpha\inv(s),\rho}$ . Therefore, $\ctype{f} \cap \alpha\inv(s) \neq \emptyset$. Since $\alpha$ is \Cs-compatible, this implies that $\ctype{f} = \ctype{s}$. Altogether, we get that $(s,r) \in T$ which concludes the proof.
\end{proof}

\subsection{\texorpdfstring{Step~2: An auxiliary class $\Ds \subseteq \bpol{\Cs}$}{Step~2: An auxiliary class D}}

We use the \mratm $\gamma$ introduced in the previous step to define a finite \vari \Ds included in \bpol{\Cs}. This object will serve as a key ingredient when proving the third step. Moreover, we reformulate Proposition~\ref{prop:pbp:level1} as a property on the equivalence classes of $\sim_{\Ds}$ (the canonical equivalence associated to \Ds). This is the formulation that we shall actually use later.

\medskip

Given any $r \in R_{A^*}$, we may define $\Hb_r$ as an arbitrary \emph{optimal} \pol{\Cs}-cover of $\beta\inv(r)$ for $\gamma$. Furthermore, we let $\Hb =  \cup_{r \in R_{A^*}} \Hb_r$. Observe that by definition, \Hb is a finite set of languages in \pol{\Cs}. Therefore, there exists a stratum \polk{\Cs} in our stratification of \pol{\Cs} (we may use it since \Cs is a finite \vari, see Section~\ref{sec:tools}) such that all $H \in \Hb$ belong to \polk{\Cs}. We define,
\[
\Ds = \bool{\polk{\Cs}}
\]
By definition, \Ds is a finite \vari and any $H \in \Hb$ belongs to \Ds. We shall work with the canonical equivalence $\sim_\Ds$ (over $A^*$) associated to \Ds. Note that by definition, we have the following fact.

\begin{fact} \label{fct:pbp:cinpbp}
	We have $\Ds \subseteq \bpol{\Cs} \subseteq \pbpol{\Cs}$. In particular, any $\sim_{\Ds}$-class $V$ is a language of \pbpol{\Cs}.
\end{fact}

We now reformulate Proposition~\ref{prop:pbp:level1} as a property on the $\sim_\Ds$-classes. We state it in the following lemma.

\begin{lemma} \label{lem:pbp:level1}
	Consider a $\sim_\Ds$-class $V \subseteq A^*$. Let $v \in V$ and $r = \rho(v)$. Then,
	\[
	(r,\gamma(V)) \in \Ts
	\]
\end{lemma}

\begin{proof}
	By definition, $r = \rho(v) = \beta(v)$. Therefore, $v \in \beta\inv(r)$ and since $\Hb_r$ is a cover of $\beta\inv(r)$ by definition, we have $H \in \Hb_r$ such that $v \in H$. Moreover, $\Hb_r$ is an optimal \pol{\Cs}-cover of $\beta\inv(r)$ for $\gamma$ by definition. Hence,
	\[
	(r,\gamma(H)) \in \gpocopti
	\] 
	By Proposition~\ref{prop:pbp:level1}, this yields $(r,\gamma(H)) \in \Ts$. Finally, we have $H \in \Ds$ by definition. Therefore, since $V$ is $\sim_\Ds$-class such that $H \cap V \neq \emptyset$ ($v$ is in the intersection), we get $V \subseteq H$. It follows that $\gamma(V) \subseteq \gamma(H)$ since $\gamma$ is a \mratm. Finally, we know that \Ts is closed under downset since $(S,\Ts)$ is \pbpol{\Cs}-saturated which yields $(r,\gamma(V)) \in \Ts$, finishing the proof of Lemma~\ref{lem:pbp:level1}.
\end{proof}

\subsection{\texorpdfstring{Step~3: \pbpol{\Cs}-argument}{Step~3: PBPol(C)-argument}}

We now come back to our main objective: proving the inclusion $\pbpopti \subseteq S$. As expected, we shall rely heavily on the objects introduced in the previous steps: the \mratm $\gamma$ and the finite \vari $\Ds \subseteq \bpol{\Cs}$. Our approach is based on the following proposition.

\begin{proposition} \label{prop:pbp:thegoal}
	For any $s \in M$, there exists a \pbpol{\Cs}-cover $\Kb_s$ of $\alpha\inv(s)$ such that for any $K \in \Kb_s$, we have $(s,\rho(K)) \in S$
\end{proposition}

The proof of Proposition~\ref{prop:pbp:thegoal} is constructive: we explain how to build the \pbpol{\Cs}-covers $\Kb_s$. Once we have them in hand, by definition of \pbpol{\Cs}-optimal $\rho$-\imprints, it will follow that for any $s \in M$, we have,
\[
\opti{\pbpol{\Cs}}{\alpha\inv(s),\rho} \subseteq \prin{\rho}{\Kb_s} \subseteq \{r \mid (s,r) \in S\}
\]
Thus, since $\pbpopti = \{(s,r)\in M \times R \mid r \in \opti{\pbpol{\Cs}}{\alpha\inv(s),\rho}\}$ by definition, it will follow that $\pbpopti \subseteq S$ as desired, finishing the proof.

\medskip

We now focus on proving Proposition~\ref{prop:pbp:thegoal}. The argument is based on the factorization forest theorem of Simon. We refer the reader to Section~\ref{sec:tools} for details on this result. In particular, this is where we need the non-standard notion of \emph{idempotent height} which we introduced together with factorization forests.

\medskip

Consider the finite \vari $\Ds \subseteq \bpol{\Cs}$ that we defined in Step~2. Since \Ds is closed under quotients, it follows from Lemma~\ref{lem:canoquo} that $\sim_\Ds$ (the associated canonical equivalence on $A^*$) is a congruence of finite index for word concatenation. Therefore the quotient set ${A^*}/{\sim_\Ds}$ is a finite monoid and the map $w \mapsto \typ{w}{\Ds}$ (which associates its $\sim_\Ds$-class to any word $w \in A^*$) is a morphism. Consequently, the Cartesian product $P = M \times ({A^*}/{\sim_\Ds})$ is a finite monoid for the componentwise multiplication and the following map $\eta: A^* \to P$ is a morphism:
\[
\begin{array}{llll}
\eta: & A^* & \to     & P                        \\
& w   & \mapsto & (\alpha(w),\typ{w}{\Ds})
\end{array}
\]
We shall work with $\eta$-factorization forests. We are now ready to start the proof of Proposition~\ref{prop:pbp:thegoal} and construct the \pbpol{\Cs}-covers $\Kb_s$. We use an induction which we state in the next proposition. Recall that for any $p \in P$ and any $h,m \in \nat$, we write,
\[
F^\eta(p,h,m) \subseteq \eta^{-1}(p) \subseteq A^*
\]
for the language of all words in $\eta^{-1}(p)$ which admit an $\eta$-factorization forest of height at most $h$ and idempotent height at most $m$.

\begin{proposition} \label{prop:pbp:level2}
	Let $h,m \in\nat$ and $p = (s,V)\in P$ (i.e. $s\in M$ and $V$ a $\sim_{\Ds}$-class). There exists a \pbpol{\Cs}-cover \Kb of $F^\eta(p,h,m)$ such that for any $K \in \Kb$, we have $(s,\rho(K)) \in S$.
\end{proposition}

Let us first use this result to prove Proposition~\ref{prop:pbp:thegoal} and finish the completeness argument. Consider $s \in M$. Our objective is to build a \pbpol{\Cs}-cover $\Kb_s$ of $\alpha\inv(s)$ such that,
\[
(s,\rho(K)) \in S \quad \text{for all $K \in \Kb_s$}
\]
Let $h = m = 3|P| - 1$. We obtain the factorization forest theorem (i.e. Theorem~\ref{thm:facto}) that for any $\sim_{\Ds}$-class $V$, we have  $\eta^{-1}((s,V)) = F^\eta((s,V),h,m)$. Thus, for any $\sim_{\Ds}$-class $V$, Proposition~\ref{prop:pbp:level2} yields a \pbpol{\Cs}-cover $\Kb_{s,V}$ of $\eta^{-1}((s,V))$ such that for any $K \in \Kb_{s,V}$, we have $(s,\rho(K)) \in S$. Observe that by definition of $\eta$, it is immediate that,
\[
\alpha\inv(s) = \bigcup_{V \in {A^*}/{\sim_\Ds}} \eta\inv((s,V))
\]
Therefore, it suffices to define $\Kb_s$ as the union of all sets $\Kb_{s,V}$ where $V$ is some $\sim_{\Ds}$-class. It is immediate by definition that $\Kb_s$ is a \pbpol{\Cs}-cover of $\alpha\inv(s)$ and that $(s,\rho(K)) \in S$ for all $K \in \Kb_s$. This concludes the proof of Proposition~\ref{prop:pbp:thegoal}.

\medskip

It now remains to prove Proposition~\ref{prop:pbp:level2}. Let $h,m \in \nat$ and $p = (s,V) \in P$. Our objective is to build a \pbpol{\Cs}-cover \Kb of $F^\eta(p,h,m)$ which satisfies the following property:
\begin{equation} \label{eq:pbp:goalpbpolc}
\text{For all $K \in \Kb$,} \quad (s,\rho(K)) \in S
\end{equation}
We proceed by induction on \emph{two} parameters which are listed below by order of importance:
\begin{enumerate}
	\item The idempotent height $m$.
	\item The height $h$.
\end{enumerate}
We start with the induction base: $h = 0$. Let us point out that for this case, the argument is actually slightly simpler than what we did for the corresponding case in the \pol{\Cs} proof. This is because \pbpol{\Cs} always contains all finite languages (this may or may not be the case for \pol{\Cs} depending on the finite \vari \Cs).

\medskip
\noindent
{\bf Base case: Leaves.} Assume that $h = 0$. All words in $F^\eta(p,0,m)$ are either empty or made of a single letter $a \in A$. In particular $F^\eta(p,0,m)$ is a \emph{finite language}. Thus, the following set \Kb is finite:
\[
\Kb = \{\{w\} \mid w \in F^\eta(p,0,m)\}
\]
Clearly, \Kb is a finite cover of $F^\eta(p,0,m)$. Moreover, it is also a \pbpol{\Cs}-cover since \pbpol{\Cs} contains all finite languages (see Lemma~\ref{lem:pbpfinite}). It remains to show that~\eqref{eq:pbp:goalpbpolc} is satisfied.

Let $K \in \Kb$. We show that $(s,\rho(K)) \in S$. Since $K \in \Kb$, we know that $K = \{w\}$ for some $w \in F^\eta(p,0,m)$. In particular, we have $\eta(w) = p$ which means that $\alpha(w) = s$ by definition of $\eta$ since $p = (s,V)$. Thus, we have,
\[
(s,\rho(K)) = (\alpha(w),\rho(w)) \in \ptriv{\alpha,\rho}
\]
This concludes the proof. Indeed, we know that $S$ is \pbpol{\Cs}-saturated which means that $\ptriv{\alpha,\rho} \subseteq S$ by definition.

\medskip
\noindent
{\bf Inductive case.} We now assume that $h \geq 1$. Recall that our objective is to build a \pbpol{\Cs}-cover \Kb of $F^\eta(p,h,m)$ which satisfies~\eqref{eq:pbp:goalpbpolc}. We decompose $F^\eta(p,h,m)$ as the union of three languages that we cover independently.

Recall that $F^\eta_B(p,h,m)$ (resp. $F^\eta_I(p,h,m)$) denotes the language of all nonempty words in $\eta^{-1}(p)$ which admit an $\eta$-factorization forest of height at most $h$, of idempotent height at most $m$ and whose root is a \emph{binary node} (resp. \emph{idempotent node}). The construction is based on the two following lemmas.

\begin{lemma} \label{lem:pbp:binary}
	There exists a \pbpol{\Cs}-cover $\Kb_B$ of $F^\eta_B(p,h,m)$ such that for all $K \in \Kb_B$, we have $(s,\rho(K)) \in S$.
\end{lemma}

\begin{lemma} \label{lem:pbp:idem}
	There exists a \pbpol{\Cs}-cover $\Kb_I$ of $F^\eta_I(p,h,m)$ such that for all $K \in \Kb_I$, we have $(s,\rho(K)) \in S$.
\end{lemma}

Before we show these two results, let us use them to finish the proof. Let $\Kb_B$ and $\Kb_I$ be as defined in the lemmas. By Fact~\ref{fct:factounion}, $F^\eta(p,h,m)$ is equal to the following union:
\[
F^\eta(p,h,m) = F^\eta_B(p,h,m) \cup F^\eta_I(p,h,m) \cup F^\eta(p,0,0)
\]
Moreover, since $h \geq 1$, induction on $h$ in Proposition~\ref{prop:pbp:level2} yields a \pbpol{\Cs}-cover $\Kb'$ of $F^\eta(p,0,0)$ such that for all $K \in \Kb'$, we have $(s,\rho(K)) \in S$. Thus, it suffices to define $\Kb = \Kb_B \cup \Kb_I \cup \Kb'$. By definition, we know that \Kb is a \pbpol{\Cs}-cover of $F^\eta(p,h,m)$ which satisfies~\eqref{eq:pbp:goalpbpolc}.

It remains to prove the two lemmas. We start with Lemma~\ref{lem:pbp:binary} which is simpler. The argument is essentially identical to the one we used when proving the corresponding result for \pol{\Cs} in Section~\ref{sec:polc}.

\subsection{Proof of Lemma~\ref{lem:pbp:binary}}

For all $q \in P$, consider the language $F^\eta(q,h-1,m)$. Using induction on our second parameter (the height $h$), we get that for all $q = (t,W) \in P$ (i.e. $t \in M$ and $W$ is a $\sim_{\Ds}$-class), there exists a \pbpol{\Cs}-cover $\Ub_q$ of $F^\eta(q,h-1,m)$ such that for any $U \in \Ub_q$, we have $(t,\rho(U)) \in S$. Consider the following set:
\[
\Kb_B = \{K_1K_2 \mid \text{there exists $q_1,q_2 \in P$ such that $p = q_1q_2$, $K_1 \in \Ub_{q_1}$ and $K_2 \in \Ub_{q_2}$}\}
\]
One may now verify that $\Kb_B$ is a \pbpol{\Cs}-cover of $F_B^\eta(p,h,m)$ and that $(s,\rho(K)) \in S$ for $K \in \Kb_B$. The proof is left to the reader as it is identical to that of Lemma~\ref{lem:pbinary} in Section~\ref{sec:polc}. In particular, when proving that $(s,\rho(K)) \in S$ for $K \in \Kb_B$, one uses the fact that $S$ is \pbpol{\Cs}-saturated (specifically closure under multiplication).

\subsection{Proof of Lemma~\ref{lem:pbp:idem}}

We turn to Lemma~\ref{lem:pbp:idem}. The proof is more involved. In particular, let us point out that this is where we finally use the preliminary results that we obtained in Steps~1 and~2 (i.e. the \mratm $\gamma$ and Lemma~\ref{lem:pbp:level1}).

\medskip

Our objective is to construct a \pbpol{\Cs}-cover $\Kb_I$ of $F_I^\eta(p,h,m)$ such that for all $K \in \Kb_I$, we have $(s,\rho(K)) \in S$. Note that we may assume without loss of generality that $m \geq 1$ and $p$ is an idempotent of $P$. Indeed, otherwise we have $F_I^\eta(p,h,m) = \emptyset$ and we may simply choose $\Kb_I = \emptyset$. Note that since we have $p = (s,V)$, this means that $s$ is an idempotent of $M$ and $V$ is an idempotent of ${A^*}/{\sim_{\Ds}}$.

In the sequel, we shall mostly use the hypothesis that $s \in M$ is idempotent. Hence, we rename it $e$ to underline this hypothesis: we want to build a \pbpol{\Cs}-cover $\Kb_I$ of $F_I^\eta(p,h,m)$ (with $p = (e,V)$) such that for all $K \in \Kb_I$, we have $(e,\rho(K)) \in S$.

This proof is where the argument departs from what we did for \pol{\Cs} in Section~\ref{sec:polc}. On one hand, the construction itself remains similar: it is still based on Lemma~\ref{lem:inducidem}. However, there is a subtle difference and more importantly, showing that $(e,\rho(K)) \in S$ for all $K \in \Kb_I$ is much more involved (this is where we shall need to use induction on the idempotent height $m$ and the work we did in Steps~1 and~2).

Recall that applying Lemma~\ref{lem:inducidem} requires a cover \Ub of $F^\eta(p,h-1,m-1)$ and a language containing $\eta^{-1}(p)$. We first define these two objects. We define \Ub as an optimal \pbpol{\Cs}-cover of $F^\eta(p,h-1,m-1)$. Moreover, recall $p = (e,V)$ where $V$ is a $\sim_\Ds$-class. Clearly, we have $\eta^{-1}(p) \subseteq V$ by definition of $\eta$.

\begin{remark}
	There is a subtle but crucial difference with what we did for \pol{\Cs}. The  \pbpol{\Cs}-cover \Ub is {\bf not} obtained from induction. Instead, it is an arbitrary optimal \pbpol{\Cs}-cover of $F^\eta(p,h-1,m-1)$. We shall use induction later for a different purpose.
\end{remark}

Before, we use \Ub and $V$ together with Lemma~\ref{lem:inducidem} to build our \pbpol{\Cs}-cover $\Kb_I$ of $F^\eta_I(p,h,m)$, let us present two important properties of these objects that we shall need later. Both are obtained by induction in Proposition~\ref{prop:pbp:level2}. The first one involves \Ub only.

\begin{lemma} \label{lem:thesimpleone}
	For any $n \geq 1$, $U_1,\dots,U_n \in \Ub$, we have $(e,\rho(U_1 \cdots U_n)) \in S$.
\end{lemma}

\begin{proof}
	We start by handling the special case when $n = 1$ and then use it to treat the general case. Consider $U \in \Ub$. We show that $(e,\rho(U)) \in S$. Using induction in Proposition~\ref{prop:pbp:level2} (the induction parameter that we use here is unimportant since both $h$ and $m$ have decreased), we get a \pbpol{\Cs}-cover $\Ub'$ of $F^\eta(p,h-1,m-1)$ such that $(e,\rho(U')) \in S$ for any $U' \in \Ub'$. Since \Ub is an optimal \pbpol{\Cs}-cover of $F^\eta(p,h-1,m-1)$ by definition, we have,
	\[
	\prin{\rho}{\Ub} \subseteq \prin{\rho}{\Ub'}
	\]
	Moreover, since $U \in \Ub$, we have $\rho(U) \in \prin{\rho}{\Ub}$ and therefore, $\rho(U) \in \prin{\rho}{\Ub'}$ by the above inclusion. This yields $U' \in \Ub'$ such that $\rho(U) \leq \rho(U')$. Finally, by definition of $\Ub'$, we have $(e,\rho(U')) \in S$. Thus, since $S$ is \pbpol{\Cs}-saturated and therefore closed under downset, this yields $(e,\rho(U)) \in S$ as desired.
	
	\medskip
	
	It remains to treat the general case. We consider $U_1,\dots,U_n \in \Ub$. By the special case treated above, we already know that for all $i \leq n$, we have $(e,\rho(U_i)) \in S$. Therefore, since $S$ is \pbpol{\Cs}-saturated and $e$ is idempotent, we obtain from closure under multiplication that,
	\[
	(e,\rho(U_1 \cdots U_n)) = (e^n,\rho(U_1) \cdots \rho(U_n))  \in S 
	\]
	This concludes the proof of Lemma~\ref{lem:thesimpleone}.
\end{proof}

Moreover, we have the following additional property of \Ub and $V$ which is where we need induction on the idempotent height $m$ in Proposition~\ref{prop:pbp:level2} and the machinery introduced in Step~1 and~2 (i.e. the \mratm $\gamma$ and Lemma~\ref{lem:pbp:level1}). Moreover, this is also where $\rho$ being \nice is important. 

\begin{lemma} \label{lem:themainthing}
	Let $U_1,\dots,U_n \in \Ub$ such that $\rho(U_1 \cdots U_n)$ is an idempotent $f$ of $R$.  We have,
	\[
	(e,f \cdot \rho(V) \cdot f) \in S
	\]
\end{lemma}

The proof of Lemma~\ref{lem:themainthing} is quite involved. Thus, we postpone it to the end of our argument.  Let us first use \Ub and $V$ to build our \pbpol{\Cs}-cover of $F_I^\eta(p,h,m)$ and finish the proof of Lemma~\ref{lem:pbp:idem}.

\medskip

We apply Lemma~\ref{lem:inducidem} together with our \pbpol{\Cs}-cover \Ub of $F^\eta(p,h-1,m-1)$ and $V$ which contains $\eta^{-1}(p)$. This yields a cover $\Kb_I$ of $F_I^\eta(p,h,m)$ such that any $K \in \Kb_I$ is a concatenation $K = K_1 \cdots K_n$ where each $K_i$ is of one of the two following kinds:
\begin{enumerate}
	\item $K_i$ is a language in \Ub, or,
	\item $K_i = U_1 \cdots U_\ell V U'_1 \cdots U'_{\ell'}$ where $U_1,\cdots,U_\ell, U'_1,\dots,U'_{\ell'} \in \Ub$ and there exists an \emph{idempotent} $f \in R$ such that $\rho(U_1 \cdots U_\ell) = \rho(U'_1 \cdots U'_{\ell'}) = f$.
\end{enumerate}
Clearly, any $K \in \Kb_I$ belongs to \pbpol{\Cs}. Indeed, $K$ is a concatenation of languages in \pbpol{\Cs} by definition ($V \in \pbpol{\Cs}$ by Fact~\ref{fct:pbp:cinpbp} since it is $\sim_\Ds$-class) and \pbpol{\Cs} is closed under concatenation by Theorem~\ref{thm:pclos}. Thus, $\Kb_I$ is a \pbpol{\Cs}-cover of $F^\eta_I(p,h,m)$.

It remains to prove that for any $K \in \Kb_I$, we have $(e,\rho(K)) \in S$. This is where we use Lemmas~\ref{lem:thesimpleone} and~\ref{lem:themainthing}. Let $K \in \Kb_I$, by definition we have  $K = K_1 \cdots K_n$ where all languages $K_i$ are as described in the two items above. Observe that for all $i \leq n$, we have $(e,\rho(K_i)) \in S$. If $K_i$ is as described in the first item ($K_i \in \Ub$), this is by Lemma~\ref{lem:thesimpleone}. Otherwise, when $K_i$ is as described in the second item, this is by Lemma~\ref{lem:themainthing}. Therefore, since $S$ is \pbpol{\Cs}-saturated and $e$ is idempotent, we get from closure under multiplication that,
\[
(e,\rho(K)) = (e^n,\rho(K_1) \cdots \rho(K_n))  \in S 
\]
This concludes the proof of Lemma~\ref{lem:pbp:idem}. It remains to show Lemma~\ref{lem:themainthing}.

\subsection{Proof of Lemma~\ref{lem:themainthing}}

Let $U_1,\dots,U_n \in \Ub$ such that $\rho(U_1 \cdots U_n)$ is an idempotent $f$ of $R$. We have to show that,
\[
(e,f \cdot \rho(V) \cdot f) \in S
\]
This is where we use induction on the idempotent height $m$ and the preliminary results that we obtained in Steps~1 and~2 (specifically, we rely on the \mratm $\gamma$ and Lemma~\ref{lem:pbp:level1}). More precisely, proving that $(e,f \cdot \rho(V) \cdot f) \in S$ is achieved using \pbpol{\Cs}-closure. Applying this operation requires elements within the set \Ts such that $(S,\Ts)$ is \pbpol{\Cs}-saturated. In Lemma~\ref{lem:pbp:thefinalone} below, we use induction on $m$ to prove the existence of specific elements in $S$. This gives us more information on $\gamma$ (since its definition depends on $S$). Together with Lemma~\ref{lem:pbp:level1}, this information on $\gamma$ allows us to exhibit the right elements in \Ts to apply \pbpol{\Cs}-closure for proving that $(e,f \cdot \rho(V) \cdot f) \in S$.

\begin{lemma} \label{lem:pbp:thefinalone}
	We have $(e,f) \in \gamma(V)$.  
\end{lemma}

Before we prove Lemma~\ref{lem:pbp:thefinalone}, let us show that $(e,f \cdot \rho(V) \cdot f) \in S$ using \pbpol{\Cs}-closure. This is where we use the hypothesis that $\rho$ is {\bf \nice}. It implies that we have words $v_1,\dots,v_\ell \in V$ such that, $\rho(v_1) + \cdots + \rho(v_\ell) = \rho(V)$.

For all $i \leq \ell$, we let $r_i = \rho(v_i)$. Consequently, we have $\rho(V) = r_1 + \cdots + r_\ell$. Since all words $v_i$ belong to the $\sim_{\Ds}$-class $V$ by definition, it follows from Lemma~\ref{lem:pbp:level1} that for all $i \leq \ell$, we have,
\[
(r_i,\gamma(V)) \in \Ts
\]
Moreover, we have $(e,f) \in \gamma(V)$ by Lemma~\ref{lem:pbp:thefinalone}. Therefore, we may use the hypothesis that $(S,\Ts)$ is \pbpol{\Cs}-saturated to apply \pbpol{\Cs}-closure (i.e. Operation~\eqref{op:pbp:lower}) for each $i \leq \ell$ which yields that,
\[
(e,f \cdot (r_i + \rho(\varepsilon)) \cdot f) \in S
\]
Since $e$ is idempotent, we may now use closure under multiplication to obtain,
\[
\left(e,\prod_{1 \leq i \leq \ell}f \cdot (r_i + \rho(\varepsilon)) \cdot f\right) \in S
\]
Recall that by hypothesis, we have $f = \rho(U_1 \dots U_n)$. Thus, we have $\rho(\varepsilon) \cdot f = f \cdot \rho(\varepsilon) = f$. We may now distribute the multiplication in the right component to obtain that,
\[
f \cdot \left(r_1 + \cdots +r_\ell\right) \cdot f \leq \prod_{1 \leq i \leq \ell}\left(f \cdot (r_i + \rho(\varepsilon)) \cdot f\right)
\]
Since $r_1 + \cdots + r_\ell = \rho(V)$, the above can be rewritten as follows,
\[
f \cdot \rho(V) \cdot f \leq \prod_{1 \leq i \leq \ell}\left(f \cdot (r_i + \rho(\varepsilon)) \cdot f\right)
\]
Combined with closure under downset, this yields as desired that,
\[
(e,f \cdot \rho(V) \cdot f) \in S
\]
This concludes the proof of Lemma~\ref{lem:themainthing}. It remains to prove Lemma~\ref{lem:pbp:thefinalone}.

\begin{proof}[Proof of Lemma~\ref{lem:pbp:thefinalone}]	
	Our objective is to show that $(e,f) \in \gamma(V)$. Recall that by definition of $\gamma$, we have,
	\[
	\gamma(V) = \{(s,r) \mid \text{$(s,r) \in \tau(V \cap G)$ for some good language $G$}\}
	\]
	Therefore, we must exhibit some good language $G \subseteq A^*$ such that $(e,f) \in \tau(V \cap G)$.
	
	Recall that $f = \rho(U_1 \cdots U_n)$ for $U_1,\dots,U_n \in \Ub$. We define $G$ as the language of all words which are an infix of some other word in $F^\eta(p,h+n-2,m-1)$. It now remains to show that $G$ is indeed good and that $(e,f) \in \tau(V \cap G)$. The two proofs are independent.
	
	\medskip
	\noindent
	{\bf First property: $G$ is good.} We have to show that $G$ is closed under infixes and $\tau(G) \subseteq S$. The former is immediate by definition of $G$. Therefore, we may concentrate on the latter: $\tau(G) \subseteq S$. This is where we need induction on the idempotent height $m$ in Proposition~\ref{prop:pbp:level2}. 
	
	By definition, all words in $G$ are an infix of some other word admitting an $\eta$-factorization forest of height at most $h+n-2$ and of idempotent height at most $m-1$. Therefore, it follows from Proposition~\ref{prop:idheight} that all words in $G$ admit an $\eta$-factorization forest of height at most $h+n$ and idempotent height at most $m-1$. Thus, we have,
	\[
	G \subseteq \bigcup_{q \in P} F^\eta(q,h+n,m-1)
	\]
	Since $\tau$ is a \mratm (whose evaluation hemiring $2^{M \times R}$ uses union as addition), it follows that,
	\[
	\tau(G) \subseteq \bigcup_{q \in P} \tau(F^\eta(q,h+n,m-1))
	\]
	Therefore, it suffices to show that for all $q \in P$, we have $\tau(F^\eta(q,h+n,m-1)) \subseteq S$. Let $q \in P$, we have $q = (t,W)$ with $t \in M$ and $W$ a $\sim_{\Ds}$-class. Consider $(s,r) \in \tau(F^\eta(q,h+n,m-1))$. We show that $(s,r) \in S$. By definition of $\tau$, we have,
	\[
	r \in \opti{\pbpol{\Cs}}{F^\eta(q,h+n,m-1) \cap \alpha\inv(s),\rho}
	\]
	This implies that $F^\eta(q,h+n,m-1) \cap \alpha\inv(s) \neq \emptyset$. Since $F^\eta(q,h+n,m-1) \subseteq \eta\inv(q)$ which is itself a subset of $\alpha\inv(t)$ by definition of $\eta$, this implies that $s = t$. Consequently, we have $F^\eta(q,h+n,m-1) \subseteq \alpha\inv(s)$. Altogether, we obtain,
	\[
	r \in \opti{\pbpol{\Cs}}{F^\eta(q,h+n,m-1),\rho}
	\]
	Using induction on the idempotent height $m$ in Proposition~\ref{prop:pbp:level2} (our most important parameter), we get a \pbpol{\Cs}-cover $\Kb'$ of $F^\eta(q,h+n,m-1)$ such that $(s,\rho(K'))\in S$ for any $K' \in \Kb'$ (recall that $q = (t,W)$ and $t = s$). Since $\Kb'$ is a \pbpol{\Cs}-cover of $F^\eta(q,h+n,m-1)$,
	\[
	\opti{\pbpol{\Cs}}{F^\eta(q,h+n,m-1),\rho} \subseteq \prin{\rho}{\Kb'}
	\]
	Therefore, we have $r \in \prin{\rho}{\Kb'}$ and we obtain $K' \in \Kb'$ such that $r \leq \rho(K')$. Finally, since $(s,\rho(K')) \in S$ by definition of $\Kb'$ and $S$ is \pbpol{\Cs}-saturated, closure under downset yields $(s,r) \in S$.

	\medskip
	\noindent
	{\bf Second property: $(e,f) \in \tau(V \cap G)$.} Recall that $f = \rho(U_1 \cdots U_n)$ where $U_1,\dots,U_n \in \Ub$. Moreover, by definition, \Ub is an optimal \pbpol{\Cs}-cover of $F^\eta(p,h-1,m-1)$ (for $\rho$). Therefore, we know that for all $i \leq n$, we have,
	\[
	\rho(U_i) \in \opti{\pbpol{\Cs}}{F^\eta(p,h-1,m-1),\rho}
	\]
	Moreover, since $p = (e,V)$, we have $F^\eta(p,h-1,m-1) \subseteq \eta\inv(p) \subseteq \alpha\inv(e)$ which means that $F^\eta(p,h-1,m-1) \cap \alpha\inv(e) = F^\eta(p,h-1,m-1)$. Thus, for all $i \leq n$, we have,
	\[
	\rho(U_i) \in \opti{\pbpol{\Cs}}{F^\eta(p,h-1,m-1) \cap \alpha\inv(e),\rho}
	\]
	By definition of $\tau$, this exactly says that for all $i \leq n$,
	\[
	(e,\rho(U_i)) \in \tau(F^\eta(p,h-1,m-1))
	\]
	Therefore, since $e$ is an idempotent of $M$ and $f = \rho(U_1 \cdots U_n)$, we may use the fact that $\tau$ is a morphism to get,
	\[
	(e,f) = (e^n,\rho(U_1 \cdots U_n)) \in \tau((F^\eta(p,h-1,m-1))^n)
	\]
	We now prove that $(F^\eta(p,h-1,m-1))^n \subseteq V \cap G$. By definition of \ratms, it will then follow that $\tau((F^\eta(p,h-1,m-1))^n) \subseteq \tau(V \cap G)$.  This implies as desired that $(e,f) \in \tau(V \cap G)$.
	
	Let $w \in (F^\eta(p,h-1,m-1))^n$. By definition, $w = w_1 \cdots w_n$ with $w_i \in F^\eta(p,h-1,m-1)$ for all $i$. In particular $\eta(w_i) = p$ for all $i$ and since $p$ is idempotent, this yields $\eta(w) = p = (e,V)$. Consequently, $\typ{w}{\Ds} = V$ by definition of $\eta$ and we get that $w \in V$. It remains to show that $w \in G$. Since all words $w_i$ belong to $F^\eta(p,h-1,m-1)$, they admit $\eta$-factorization forests of height at most $h-1$ and idempotent height at most $m-1$. Therefore $w = w_1 \cdots w_n$ admits an $\eta$-factorization forest of height at most $h-1+n-1$ and idempotent height at most $m-1$. Indeed, it suffices to combine the forests of the factors $w_i$ with binary nodes. This means that $w \in F^\eta(p,h+n-2,m-1)$. Hence, by definition of $G$, we have $w \in G$. This concludes the proof that $(e,f) \in \tau(V \cap G)$ and that of Lemma~\ref{lem:pbp:thefinalone}.
\end{proof}


\section{Conclusion}
\label{sec:conc}
We showed that the separation and covering problems are decidable for all classes of the form \pol{\Cs} and \pbpol{\Cs} where \Cs is an arbitrary finite \vari. As we explained, these results have important consequences for the quantifier alternation hierarchies within first-order logic overs words: we get that separation and covering are decidable for the level \siwt. This result may be lifted to the stronger logic \siws{3} with a transfer result of~\cite{pzsucc}. Finally, one also gets that the membership problem is decidable for \siw{4} and \siws{4} using a reduction theorem of~\cite{pzqalt}.

Let us point out that while this is not apparent in the above summary, we still know more about \pol{\Cs} than we do about \pbpol{\Cs}.  Indeed, our \pol{\Cs}-covering algorithm is based on a characterization of \emph{\pol{\Cs}-optimal pointed \imprints} which holds for {\bf any} \mratm. This result is stronger than the decidability of \pol{\Cs}-covering (one only needs to consider \nice \mratms for this). Moreover, the fact that we have this stronger result is precisely why we are able to handle the larger class \pbpol{\Cs}. On the other hand, while we also have a characterization of \pbpol{\Cs}-optimal pointed \imprints, it only holds for \nice \mratms. Therefore, a natural next move is trying to generalize this characterization to all \emph{all} \mratms. Indeed, such a result could be the key to solving covering for one more level in concatenation hierarchies.

Another interesting (and much simpler) objective is analyzing the complexity of the decision problems that we have just solved. Naturally, this depends on the finite \vari \Cs for both \pol{\Cs} and \pbpol{\Cs}. We leave this for further work.


\bibliographystyle{alpha}
\bibliography{main}

\end{document}